\documentclass{article}

\usepackage{PRIMEarxiv}
\usepackage{type1cm}

\usepackage[utf8]{inputenc} 
\usepackage[T1]{fontenc}\usepackage{lmodern}    
\usepackage{url}            
\usepackage{booktabs}       
\usepackage{amsfonts}       
\usepackage{nicefrac}       
\usepackage{microtype}
\usepackage{lipsum}
\usepackage{fancyhdr}       
\usepackage{graphicx}
\usepackage{placeins}       
\graphicspath{{media/}}     

\usepackage{amsmath,amssymb,amsthm}
\usepackage{algorithm}
\usepackage{algorithmic}
\newtheorem{theorem}{Theorem}
\newtheorem{definition}{Definition}
\newtheorem{assumption}{Assumption}
\newtheorem{proposition}{Proposition}
\newtheorem{lemma}{Lemma}
\newtheorem{remark}{Remark}
\newtheorem{example}{Example}
\newtheorem{corollary}{Corollary}
\usepackage{textcomp}
\usepackage{multirow}
\usepackage{bm}
 \usepackage{natbib}
\usepackage[hidelinks]{hyperref}       

\newcommand{\R}{\mathbb{R}}
\newcommand{\E}{\mathbb{E}}
\newcommand{\1}{\mathbf{1}}

\newcommand{\cP}{\mathcal{P}}

\newcommand{\cK}{\mathcal{K}}

\newcommand{\cF}{\mathcal{F}}
\newcommand{\cS}{\mathcal{S}}
\newcommand{\cM}{\mathcal{M}}
\newcommand{\cA}{\mathcal{A}}
\newcommand{\cX}{\mathcal{X}}
\newcommand{\cU}{\mathcal{U}}
\newcommand{\cC}{\mathcal{C}}
\newcommand{\cG}{\mathcal{G}}

\title{Endogenous Information in Routing Games: Memory-Constrained Equilibria, Recall Braess Paradoxes, and Memory Design}

\author{
  Saad Alqithami \\
  \texttt{alqithami@gmail.com} 
}

\date{23 September 2026}
\begin{document}
\maketitle

\begin{abstract}
    We study routing games in which travelers (or an AI guidance layer mediating their choices) optimize over an action set that is itself produced by bounded memory, surfacing, and forgetting policies. We present two coupled layers. First, we give a fully specified micro model in which each agent carries a finite memory state that evolves as a Markov chain (e.g., LRU eviction) and selects among recalled and surfaced routes via a logit response; a stationary \emph{Forgetful Wardrop Equilibrium} (FWE) exists as a fixed point between congestion and the memory chain's stationary law. Second---and as the main design layer---we introduce \emph{stationary salience policies} (additive utility biases over routes) that summarize the effect of memory and interface ranking on choice. Salience-weighted stochastic user equilibrium is the unique minimizer of a strictly convex potential, yielding uniqueness and supporting convergent potential-minimization algorithms without contraction assumptions, and enabling inverse-equilibrium analysis: we characterize implementable flows under ratio budgets and affine tying (fairness) constraints and derive a scalar-search reduction with convex inner problems on parallel networks and structure-exploiting oracles on series-parallel networks. We tighten the bridge between the layers. For last-choice memory ($B=1$) the micro model is \emph{exactly} equivalent to the salience model with salience proportional to the surfacing distribution; hence any interior salience vector is operationally realizable by an appropriate surfacing policy. For larger memories ($B>1$) with LRU eviction, we propose an LRU$\rightarrow$TTL$\rightarrow$availability$\rightarrow$salience heuristic. We distinguish an external independent-request cache theorem from conditional random-menu bounds; an end-to-end error rate for endogenous LRU choices is not established. We define a social-cost \emph{Recall Braess Paradox}---improving recall can increase equilibrium delay without changing physical capacity---and prove a network-class theorem: every two-terminal network with at least two distinct $s$--$t$ paths admits latencies and a nested expansion of static admissible menus that induces the paradox in the deterministic Wardrop baseline. A reproducible six-experiment suite verifies exact small-instance reductions and inverse design, separates stationary predictions from oscillatory population dynamics, and evaluates the limits of the larger-memory surrogate. The experiments are computational and do not constitute behavioral calibration.
\end{abstract}

\keywords{routing games\and stochastic user equilibrium\and imperfect recall\and endogenous information\and Braess paradox\and implementability\and memory governance\and series-parallel networks.}


\section{Introduction}

Traffic routing is a canonical noncooperative game: each traveler selects a route to minimize personal travel time, while congestion couples decisions through shared edge latencies.
The resulting Wardrop equilibrium (the non-atomic analogue of Nash equilibrium) is stable to unilateral deviations but can be socially inefficient, and classic Braess' paradox shows that adding a physical ``shortcut'' can increase equilibrium delay.

This paper focuses on a different mechanism that is increasingly operational in AI-mediated mobility: \emph{endogenous information created by memory}.
In practice, travelers (and navigation assistants that mediate their choices) do not optimize over a fixed, fully known set of routes.
They act on a small, evolving subset shaped by experience, attention, and explicit memory-management policies.
In modern AI systems, memory is a first-class design variable governed by retention budgets and eviction policies (FIFO/LRU/priority-decay/summary-based retention) that trade off utility, compute, and privacy.

\subsection{Running example: calibrated forgetting as a Pigouvian correction}
\label{sec:running-example}
The following two-link instance illustrates both the opportunity and the conceptual novelty.
Consider a single OD pair with total demand $d=1$ and two parallel routes:
route $a$ has constant latency $\ell_a(x)\equiv 1$ and route $b$ has latency $\ell_b(x)=x$.
With full recall, the Wardrop equilibrium routes all demand through $b$ (since $b$ is initially shorter), yielding equilibrium social cost $\mathrm{SC}=1$.
The system-optimal split routes half the demand through $a$ and half through $b$, yielding $\mathrm{SC}=3/4$.

\begin{example}[Pigouvian forgetting]
\label{ex:pigou-forgetting}
Suppose an AI system (or a cognitive limitation) induces \emph{calibrated forgetting}: an $\alpha$ fraction of the population does not recall route $b$ and therefore uses $a$.
Then $x_b=1-\alpha$ and the social cost is $\mathrm{SC}(\alpha)=\alpha\cdot 1 + (1-\alpha)\cdot (1-\alpha)=1-\alpha+\alpha^2$.
This is minimized at $\alpha^\star=1/2$ with $\mathrm{SC}(\alpha^\star)=3/4$, exactly matching the system optimum \emph{without changing network capacity and without tolls}.
\end{example}

Example~\ref{ex:pigou-forgetting} motivates the central mechanism studied in this paper:
\emph{recall is an endogenous action-set constraint that can act as a non-monetary congestion-control instrument.}
The static example is closely related to informational Braess effects under restricted menus. The additional modeling question studied here concerns information \emph{generated and erased dynamically} by a memory process coupled to congestion through experience and guidance; the static illustration alone does not establish that dynamic mechanism.

\subsection{Problem identification: the endogenous-information externality}
Limiting recall restricts a player's feasible action set, which in congestion games can change equilibrium selection and welfare.
The scientific question is therefore not ``does more capacity help?'' but rather:
\emph{how does improving recall change equilibrium, and what recall policy should an AI system implement under governance constraints?}

\subsection{Three core questions}
We organize the paper around three targets that, taken together, organize the analysis:
\begin{enumerate}
    \item \textbf{Well-posedness and stability.} Does the coupled flow--memory system admit a \emph{unique} stationary equilibrium, and do natural day-to-day dynamics converge to it?
    \item \textbf{Recall Braess phenomena.} When can improving recall (larger memory budgets or weaker forgetting) \emph{worsen} equilibrium welfare without changing network capacity?
    \item \textbf{Equilibrium-aware memory design.} How can an AI route guidance system choose retention and surfacing policies to reduce congestion externalities subject to governance constraints (privacy, fairness, compute)?
\end{enumerate}

\subsection{Contributions and roadmap}
\label{sec:contrib}
\begin{enumerate}
    \item \textbf{A fully specified forgetful routing game (Sections~\ref{sec:prelim}--\ref{sec:dynamic-recall-model}).}
    We introduce a stationary mean-field routing model with endogenous recall in which each agent has a finite memory state, receives surfaced alternatives, chooses via a bounded-rational response, and updates memory under an explicit eviction/forgetting policy.
    \item \textbf{Equilibrium concept and existence (Section~\ref{sec:existence}).}
    We define \emph{Forgetful Wardrop Equilibrium} (FWE) as the fixed point coupling within-period congestion consistency with stationarity of the memory Markov chain, and we prove existence under mild regularity.
    \item \textbf{A convex-potential ``design layer'' beyond contraction (Section~\ref{sec:salience}, especially Section~\ref{sec:micro2salience}).}
    We introduce a policy-relevant reduced-form equilibrium model in which memory/guidance act through \emph{route salience}.
    We show that \emph{salience-weighted SUE} is the unique minimizer of a strictly convex potential, yielding uniqueness and supporting potential-minimization algorithms without contraction assumptions; this does not guarantee convergence of arbitrary day-to-day dynamics.
    We further give (i) an \emph{exact} micro-foundation for $B=1$ last-choice memory, and (ii) an explicit LRU$\to$TTL$\to$salience approximation for $B>1$ LRU-type memory with a cited $\sqrt{\log B/B}$ rate for individual hits in an external Poisson-request cache; no end-to-end endogenous-choice rate is claimed (Section~\ref{sec:micro2salience-B}). We also quantify the accuracy of the Poissonization device for discrete-time day-to-day departures under diffused route popularity (Lemma~\ref{lem:poissonization-accuracy} and Corollary~\ref{cor:diffuse-poisson}).
    \item \textbf{Implementability and governance constraints (Section~\ref{sec:implementability}).}
    We formalize \emph{implementability} as an inverse-equilibrium problem for memory/guidance policies.
    In the salience model, every interior feasible flow is implementable by an explicit inverse construction, and we give a sharp constrained-implementability theorem under influence budgets. We interpret this as implementability by stationary interface/ranking salience; it is \emph{micro-realizable without loss} in the $B=1$ regime via Corollary~\ref{cor:micro-realizability}.
    We then show that governed salience design reduces to a single-level implementable-flow optimization problem over this constrained set (Theorem~\ref{thm:governed-design-reduction}). On the series-parallel class, we further show that a low-dimensional decomposition-tied salience policy can implement any interior $s$--$t$ flow by a linear-time inverse construction (Theorem~\ref{thm:sp-implementability}), and local influence budgets become linear-time implementability tests (Corollary~\ref{cor:sp-local-budgets}).
    We then characterize implementability under \emph{fairness/tying} constraints (Section~\ref{sec:fair-impl}).
    \item \textbf{Recall Braess paradox and a network-class theorem (Sections~\ref{sec:vor}--\ref{sec:rbp}).}
    We define a social-cost \emph{Recall Braess Paradox} and prove: (i) a sharp analytic Pigou instance where calibrated forgetting strictly improves welfare over full recall, and (ii) a network-class theorem showing social-cost RBP under nested static admissible menus on every two-terminal network with at least two distinct $s$--$t$ paths (the only immunity class within that formulation is a unique-path network).
    \item \textbf{Equilibrium-aware design with closed-form and network-class results (Section~\ref{sec:design}).}
    We pose equilibrium-aware memory/guidance design as a bilevel optimization problem under governance constraints and provide an implicit-differentiation sensitivity formula enabling gradient-based policy optimization.
    On parallel networks, we derive a constructive reduction of bounded-influence optimal salience design to a one-dimensional search plus convex subproblems (Section~\ref{sec:parallel-design}).
    On two-terminal series-parallel networks, we show how to evaluate the logit partition function and edge marginals in linear time (Theorem~\ref{thm:sp-partition}), derive an equivalent split-flow convex program with first-order convergence guarantees (Theorem~\ref{thm:sp-split-convex} and Proposition~\ref{prop:sp-gradient}), and (in the implementability layer) give a constructive inverse design for decomposition-tied salience with linear-time feasibility under local budgets (Theorem~\ref{thm:sp-implementability}).
\end{enumerate}

\paragraph{Paper organization and reading guide.}
Part~I (Sections~\ref{sec:salience}--\ref{sec:design}) develops the stationary salience design layer,
implementability under governance, and equilibrium-aware design algorithms. Part~II
(Sections~\ref{sec:dynamic-recall-model}--\ref{sec:rbp}) introduces the explicit micro memory model,
proves existence/stability results for Forgetful Wardrop Equilibrium, develops the LRU$\to$TTL$\to$salience
approximation pipeline, and establishes Recall Braess phenomena. Part~III (Sections~\ref{sec:experiments}--\ref{sec:discussion})
provides experimental evaluation and a discussion of limitations.

\subsection{Modeling layers}
\label{sec:layer-map}
The stationary salience model and the finite-state memory model support different equilibrium and design results.
The main mechanism-design contributions (implementability, governed design under constraints, and tractable network-class algorithms) are proved in the \emph{stationary salience} design layer (Sections~\ref{sec:salience}, \ref{sec:implementability}, and \ref{sec:design}).
The finite-state memory Markov chain model (Sections~\ref{sec:dynamic-recall-model}--\ref{sec:existence}) serves as a micro-foundation: it is \emph{exactly} equivalent to the salience model for $B=1$ (Theorem~\ref{thm:micro2salience}), and it motivates a heuristic salience abstraction for $B>1$ via an explicit approximation pipeline (Section~\ref{sec:micro2salience-B}).
Table~\ref{tab:exact-approx} summarizes what is exact, what is approximate, and what is specific to the reduced-form design layer.

\begin{table}[t]
\centering \scriptsize
\begin{tabular}{@{}p{0.54\linewidth}p{0.13\linewidth}p{0.12\linewidth}p{0.12\linewidth}@{}}
\toprule
\textbf{Statement} & \textbf{Layer} & \textbf{Status} & \textbf{Reference} \\
\midrule
Existence of a stationary Forgetful Wardrop Equilibrium (FWE) & Micro (Markov) & Exact & Thm.~\ref{thm:fwe-existence} \\
Exact equivalence of $B=1$ last-choice memory and salience-weighted logit & Micro$\rightarrow$salience & Exact & Thm.~\ref{thm:micro2salience} \\
(Operational) realizability of any interior salience vector via surfacing when $B=1$ & Micro design & Exact & Cor.~\ref{cor:micro-realizability} \\
LRU$\rightarrow$TTL marginal-hit bound (external independent Poisson requests) & External cache & Conditional rate & Thm.~\ref{thm:lru-ttl} \\
Random-menu logit $\approx$ availability-weighted logit under explicit menu assumptions & Menu surrogate & Coordinate bound & Prop.~\ref{prop:awl-approx}, Cor.~\ref{cor:awl-rate} \\
Convex potential and uniqueness of SW-SUE & Salience & Exact & Prop.~\ref{prop:swsue-potential} \\
Full implementability and governed implementability geometry & Salience design & Exact & Thm.~\ref{thm:implementability-salience}, Thm.~\ref{thm:governed-implementability} \\
\bottomrule
\end{tabular}
\caption{Exact vs.\ approximate statements by modeling layer.
Unless explicitly stated otherwise, implementability and design results are proved in the salience layer; the micro memory model provides an exact foundation for $B=1$ and a heuristic pathway requiring validation for $B>1$.}
\label{tab:exact-approx}
\end{table}

\subsection{Scope}
The analysis centers on convex equilibrium structure, implementability, and constrained design in the stationary salience model. The finite-state memory model supplies a behavioral mechanism linking surfacing and recall to salience.
The finite-state micro model is explicitly specified and yields an exact bridge to salience for $B=1$, but its exact computation is combinatorial for large route sets (Section~\ref{sec:existence}); for $B>1$ we therefore emphasize approximations and scalable surrogates (Section~\ref{sec:micro2salience-B}).
The evaluation is computational rather than behavioral: it does not estimate memory policies from observed travelers or test a deployed AI guidance system.

\section{Related Work}\label{sec:related_work}

\paragraph{Traffic equilibria, stochastic user equilibrium, and day-to-day dynamics.}
Non-atomic routing games originate with Wardrop's user-equilibrium principle \citep{wardrop1952}.
Under standard monotonicity/continuity assumptions, Wardrop equilibria admit equivalent variational
inequality and convex potential formulations \citep{beckmann1956,dafermos1969,smith1979}, which
underpin much of traffic assignment theory \citep{sheffi1985,patriksson1994}.
Stochastic user equilibrium (SUE) and logit-based route choice models are classical tools for
capturing dispersed preferences, perception errors, and within-period randomness
\citep{dial1971,benakiva1985,sheffi1985}. Day-to-day adjustment and learning dynamics have been
studied extensively in transportation science; a representative route-swapping family is developed
in \citep{cascetta1991}.

\paragraph{Congestion games, potential structure, and efficiency loss.}
Congestion games are canonical potential games \citep{rosenthal1973,monderer1996}.
In nonatomic routing, the efficiency loss from selfish behavior is formalized by the price of
anarchy (PoA), with tight bounds for broad latency classes \citep{roughgarden2002,roughgarden2005book}
and refinements via smoothness/variational techniques \citep{correa2004,correa2008,koutsoupias1999}.
We also leverage structural parallels with finite congestion games when discussing policy and
governance constraints \citep{christodoulou2005}.

\paragraph{Braess-type paradoxes and network classes.}
Braess's paradox---the possibility that adding capacity worsens equilibrium travel times---was
first identified in \citep{braess1968} and popularized in the transportation literature in
\citep{murchland1970}. Subsequent work characterized when paradoxes can and cannot occur under
different modeling assumptions and network structures \citep{steinberg1983,pas1997,dafermos1984,milchtaich2006}.
Series-parallel structure plays a recurring role in both algorithm design and topology-based
characterizations \citep{duffin1965}.

\paragraph{Information design, endogenous information, and consideration sets.}
The informational Braess paradox (IBP) shows that providing additional route information to a
subpopulation can worsen overall performance \citep{acemoglu2018ibp}. This connects to the broader
literature on information design / Bayesian persuasion \citep{kamenica2011,bergemann2019,dughmi2017}
and to recent algorithmic information-design work specialized to congestion games
\citep{zhou2022aid}. Our paper is complementary: we focus on endogenous \emph{consideration sets}
generated by memory and surfacing, rather than belief noise or exogenous information structures.
This links to random-attention and consideration-set models \citep{masatlioglu2012revealed,manzini2014,cattaneo2020ram},
and to rational-inattention foundations for multinomial logit \citep{matejka2015,caplin2015}.

\paragraph{Caching-based approximations for limited memory.}
To bridge explicit bounded-memory dynamics to a stationary ``salience'' design layer, we draw on
classical and modern cache approximations for LRU/TTL policies, including asymptotic miss-ratio
analysis and characteristic-time approximations \citep{fagin1977,che2002,fricker2012lru,gast2017ttl,jiang2018ttl}.

\paragraph{Algorithmic governance and constrained influence.}
Our governance constraints (ratio budgets and tying/affine structure) are motivated by how modern
AI systems surface options under auditability and fairness constraints.
They are also closely related, mathematically, to exposure-based fairness constraints in ranking
and recommender systems \citep{singh2018,biega2018,zehlike2017,celis2018}.
Finally, we build on recent work that uses AI agents and large language models to study day-to-day
route choice and bounded rationality at scale \citep{Wang_2025,wardropbr2024}, and on our own
recent work on imperfect recall and cognitive memory architectures \citep{alqithami2026snam,alqithami2025fifa}.
\section{Preliminaries: Non-Atomic Routing Games}
\label{sec:prelim}
We review the standard non-atomic model; we focus on a single origin--destination pair for clarity and note extensions later.

\subsection{Network and flows}
Let $G=(V,E)$ be a directed graph with origin $s$ and destination $t$.
Each edge $e\in E$ has a latency (travel time) function $\ell_e:\R_{\ge 0}\to \R_{\ge 0}$ that is continuous and nondecreasing.
A (simple) $s$--$t$ path is denoted $p\in\cP$.
A (non-atomic) flow is a vector $f=(f_p)_{p\in\cP}$ with $f_p\ge 0$ and $\sum_{p\in\cP} f_p = d$, where $d>0$ is total demand.
Edge loads are $x_e(f) = \sum_{p\ni e} f_p$.
Path latency is $L_p(f) = \sum_{e\in p}\ell_e(x_e(f))$.

\subsection{Wardrop equilibrium and social optimum}
\begin{definition}[Wardrop equilibrium]
A feasible flow $f^\star$ is a Wardrop equilibrium if for every path $p$ with $f^\star_p>0$, we have
\[
L_p(f^\star) \le L_{p'}(f^\star)\quad \text{for all } p'\in\cP.
\]
\end{definition}

The social cost (total latency) is
\[
C(f) = \sum_{e\in E} x_e(f)\,\ell_e(x_e(f)) = \sum_{p\in\cP} f_p\, L_p(f).
\]
A socially optimal flow minimizes $C(f)$ over feasible flows.

\section{Design layer: stationary salience policies}\label{sec:salience-layer}
This section introduces the reduced-form \emph{stationary salience} model that serves as the main
design layer of the paper. The model abstracts recall and guidance as multiplicative weights on
routes and yields a strictly convex equilibrium characterization. Subsequent sections use this layer
to derive implementability tests and equilibrium-aware design algorithms. Later sections connect this design layer back to the explicit micro memory model.

\subsection{Stationary salience policies: a convex potential and uniqueness}
\label{sec:salience}
The explicit memory model of Section~\ref{sec:dynamic-recall-model} induces a high-dimensional
stochastic process, and establishing global uniqueness/stability of its stationary flow can require
strong contraction assumptions (Theorem~\ref{thm:unique}).
To obtain sharp equilibrium structure and enable tractable optimization, we introduce a reduced-form
\emph{stationary salience policy} model that represents recall and guidance via deterministic
multiplicative weights on routes. This abstraction yields a strictly convex potential, a unique
equilibrium, and an explicit inverse mapping from target flows to implementing salience
(Section~\ref{sec:implementability}). We return to the micro-to-salience connection in
Section~\ref{sec:micro2salience-B}.

\paragraph{Interpretation.}
In many AI-mediated choice systems, the user does not literally face a strict feasibility constraint; instead, the system retrieves/surfaces options with different \emph{prominence} and the user responds stochastically.
This motivates modeling memory as route-dependent salience weights that tilt logit choice.

\subsubsection{Stationary salience policies}
Fix a finite path set $\cP_k$ for each commodity $k$.

\begin{definition}[Stationary salience policy]
A stationary salience policy is a collection of strictly positive weights
\[
s \;=\; \big(s_{k,p}\big)_{k\in\cK,\;p\in\cP_k}\in \R_{>0}^{\sum_k|\cP_k|},
\]
where $s_{k,p}$ represents the (policy-induced) \emph{salience} of path $p$ for commodity $k$.
Equivalently, define additive biases $a_{k,p}\triangleq \frac{1}{\beta}\log s_{k,p}$.
\end{definition}

Given congestion $x$, the salience-weighted logit choice probability is
\begin{equation}
\label{eq:salience-logit}
\Pr(p\mid k,x;s)
\;=\;
\frac{s_{k,p}\exp\!\big(-\beta L_p(x)\big)}
{\sum_{r\in\cP_k} s_{k,r}\exp\!\big(-\beta L_r(x)\big)}
\qquad (p\in\cP_k).
\end{equation}
This coincides with logit choice on utilities $-L_p(x)+a_{k,p}$.

\begin{definition}[Salience-weighted stochastic user equilibrium (SW-SUE)]
\label{def:swsue}
Fix $(\ell_e)_{e\in E}$ and a salience policy $s$.
A feasible path-flow vector $f^\star=(f_{k,p}^\star)$ is a \emph{SW-SUE} if for every commodity $k$ and path $p\in\cP_k$,
\begin{equation}
\label{eq:swsue}
f_{k,p}^\star
\;=\;
d_k\,\Pr(p\mid k,x(f^\star);s),
\end{equation}
where $x(f^\star)$ is defined by \eqref{eq:edge-loads}.
\end{definition}

SW-SUE is a stochastic user equilibrium with alternative-specific constants.
While the strict convexity/potential formulation is well known in logit SUE (it is essentially Beckmann's potential plus an entropy regularizer), we record it here because it is the technical backbone of our governance and implementability results: once salience is treated as a policy lever, the equilibrium map becomes an explicitly solvable convex program.
Crucially, it admits a strictly convex potential characterization, which yields uniqueness and algorithmic stability \emph{without} requiring a contraction bound on a reduced fixed-point map.

\subsubsection{Potential formulation and uniqueness}
Define the feasible set
\[
\cF \;\triangleq\; \Big\{f\ge 0:\; \sum_{p\in\cP_k} f_{k,p}=d_k\ \ \forall k\in\cK\Big\}.
\]
For a salience policy $s$, define the objective
\begin{equation}
\label{eq:salience-potential}
\Phi_s(f)
\;\triangleq\;
\sum_{e\in E}\int_{0}^{x_e(f)} \ell_e(u)\,du
\;+\;
\frac{1}{\beta}\sum_{k\in\cK}\sum_{p\in\cP_k} f_{k,p}\Big(\log f_{k,p}-\log s_{k,p}\Big),
\end{equation}
with the convention $0\log 0 = 0$.

\begin{proposition}[Convex program for SW-SUE]
\label{prop:swsue-potential}
Assume $\beta>0$ and each $\ell_e$ is continuous and nondecreasing.
Then $\Phi_s$ is strictly convex on $\cF$ and admits a unique minimizer $f^\star\in\cF$.
Moreover, $f^\star$ is the unique SW-SUE in Definition~\ref{def:swsue}.
\end{proposition}

\begin{proof}[Proof idea.]
Write the SW-SUE fixed point as the first-order optimality conditions of the Beckmann potential augmented with an entropic regularizer shifted by $\log s$. The entropy term makes the objective strictly convex over $\cF$, yielding existence and uniqueness; the KKT conditions recover the salience-weighted logit form.
Full proof is deferred to Appendix~\ref{app:proof-prop-swsue-potential}.
\end{proof}

\subsubsection{Micro-foundation: last-choice memory yields stationary salience}
\label{sec:micro2salience}
The salience abstraction in \eqref{eq:salience-logit} can be given an exact micro-foundation as a special case of the dynamic recall model in Section~\ref{sec:dynamic-recall-model}.
This strengthens the interpretation of $s_{k,p}$ as an \emph{endogenous information mechanism} rather than an ad hoc reduced form.

\paragraph{B=1 last-choice memory with surfacing.}
Fix a commodity $k$ and suppose the memory budget is $B_k=1$.
The memory state is simply the \emph{last chosen} route $m\in\cP_k$.
Each period, a candidate route $q\sim \rho_k$ is surfaced, and the traveler chooses from $\{m,q\}$ using the logit rule \eqref{eq:logit-choice} (with costs evaluated at a fixed congestion vector $x$).
After choosing a route $p$, memory updates deterministically to $m^+=p$.

For fixed congestion $x$, this induces a Markov chain on $\cP_k$ with transition probabilities
\begin{equation}
\label{eq:transition-b1}
P_x^{(k)}(p'\mid p)
\;=\;
\sum_{q\in\cP_k}\rho_k(q)\,
\frac{\exp\!\big(-\beta L_{p'}(x)\big)\,\mathbf{1}\{p'\in\{p,q\}\}}
{\sum_{r\in\{p,q\}}\exp\!\big(-\beta L_r(x)\big)}.
\end{equation}
(When $q=p$, the set $\{p,q\}$ has one element: the denominator is $\exp(-\beta L_p(x))$ and the conditional probability of staying at $p$ is one.)

\begin{theorem}[Exact reduction: stationary last-choice memory induces salience-weighted logit]
\label{thm:micro2salience}
Fix $k$ and a congestion vector $x$.
Assume $\rho_k$ has full support and $\beta>0$.
Then the Markov chain \eqref{eq:transition-b1} is irreducible and aperiodic, and it is reversible with unique stationary distribution
\begin{equation}
\label{eq:stationary-b1}
\pi_{k,x}(p)
\;=\;
\frac{\rho_k(p)\exp\!\big(-\beta L_p(x)\big)}
{\sum_{r\in\cP_k}\rho_k(r)\exp\!\big(-\beta L_r(x)\big)}.
\end{equation}
Moreover, in stationarity the marginal probability of choosing route $p$ in a period equals $\pi_{k,x}(p)$.
Equivalently, the stationary choice rule is exactly the salience-weighted logit \eqref{eq:salience-logit} with salience weights $s_{k,p}\propto \rho_k(p)$.
\end{theorem}

\begin{proof}
Irreducibility follows from full support of $\rho_k$ and the fact that from any current route $p$ the candidate $q=p'$ occurs with positive probability and is chosen with positive probability under logit.
Aperiodicity holds because $P_x^{(k)}(p\mid p)>0$ for all $p$ (take $q=p$).

To show reversibility, define $\tilde\pi(p)\propto \rho_k(p)e^{-\beta L_p(x)}$.
For distinct $p\neq p'$, the only way to move from $p$ to $p'$ in one step is that the surfaced candidate equals $p'$ and the logit chooses $p'$.
Thus
\[
P_x^{(k)}(p'\mid p)
\;=\;
\rho_k(p')\,
\frac{e^{-\beta L_{p'}(x)}}{e^{-\beta L_{p}(x)}+e^{-\beta L_{p'}(x)}}.
\]
Hence, for $p\neq p'$,
\[
\tilde\pi(p)\,P_x^{(k)}(p'\mid p)
=
\rho_k(p)e^{-\beta L_p(x)}\cdot
\rho_k(p')\frac{e^{-\beta L_{p'}(x)}}{e^{-\beta L_{p}(x)}+e^{-\beta L_{p'}(x)}}
=
\tilde\pi(p')\,P_x^{(k)}(p\mid p'),
\]
so detailed balance holds.
Therefore $\pi_{k,x}$ in \eqref{eq:stationary-b1} is stationary; uniqueness follows from irreducibility.

Finally, when the process is stationary, the next-period memory state equals the chosen route, so the stationary distribution of memory coincides with the stationary distribution of choices, yielding the claim.
\end{proof}

\begin{corollary}[Coupled routing equilibrium equals SW-SUE for $B_k=1$]
\label{cor:micro2salience-equilibrium}
Consider the non-atomic routing game with last-choice memory ($B_k=1$ for all $k$) and surfacing distributions $\{\rho_k\}$.
If the induced equilibrium exists, then its flow component is exactly a SW-SUE with salience weights $s_{k,p}\propto \rho_k(p)$.
In particular, uniqueness and algorithmic stability follow from Proposition~\ref{prop:swsue-potential}.
\end{corollary}

\begin{proof}[Proof idea.]
By Theorem~\ref{thm:micro2salience}, in the $B_k=1$ micro model the stationary per-period choice probabilities equal a salience-weighted logit with weights $s_{k,p}\propto \rho_k(p)$. Substituting these probabilities into the non-atomic flow definition yields exactly the SW-SUE fixed point, and uniqueness follows from strict convexity of the SW-SUE potential.
Full proof is deferred to Appendix~\ref{app:proof-cor-micro2salience-equilibrium}.
\end{proof}

\begin{corollary}[Operational realizability of salience via surfacing when $B_k=1$]
\label{cor:micro-realizability}
Fix any collection of strictly positive salience weights $\{s_{k,p}\}_{k\in\cK,\,p\in\cP_k}$.
Define, for each commodity $k$, a surfacing distribution
\[
\rho_k(p)\;\triangleq\;\frac{s_{k,p}}{\sum_{r\in\cP_k} s_{k,r}}.
\]
Consider the $B_k=1$ last-choice memory micro model with these surfacing distributions.
Then the induced stationary within-period choice probabilities coincide \emph{exactly} with the salience-weighted logit model with salience $s$, and the induced network equilibrium flow is the unique SW-SUE for $s$.
Consequently, in the $B_k=1$ regime, any interior flow implementable by stationary salience (Theorem~\ref{thm:implementability-salience}) is also implementable by an explicit micro policy (choice of $\rho$).
\end{corollary}

\begin{proof}[Proof idea.]
Theorem~\ref{thm:micro2salience} shows that for $B_k=1$ the stationary choice law is salience-weighted logit with salience proportional to the surfacing distribution: $s_{k,p}\propto \rho_k(p)$.
Since salience is defined only up to a per-commodity multiplicative constant, choosing $\rho_k$ proportional to a target $s_k$ realizes that target exactly.
The equilibrium statement then follows from Corollary~\ref{cor:micro2salience-equilibrium}.
Full proof is deferred to Appendix~\ref{app:proof-cor-micro-realizability}.
\end{proof}

\begin{remark}
Theorem~\ref{thm:micro2salience} provides an explicit bridge between the micro memory kernel and the salience mechanism.
Richer memory states ($B_k>1$, LRU lists, summary-based retention) expand the policy space beyond pure surfacing weights; however, the $B_k=1$ case already shows that a simple AI ``suggestion layer'' can induce an equilibrium that is exactly the optimizer of a strictly convex potential.
This is the key technical reason salience policies provide a strong design handle \emph{beyond contraction}.
\end{remark}

\section{Implementability: inverse equilibrium and memory-as-control}
\label{sec:implementability}
The stationary salience layer (Section~\ref{sec:salience}) yields a unique equilibrium for every
choice of salience policy $s$. This raises an inverse question that is central for design:
\emph{given a target equilibrium behavior, when is it implementable by some salience policy, and how
much ``influence'' is required?}

This section provides (i) an explicit inverse mapping from interior target flows to salience
parameters, (ii) sharp feasibility tests under influence budgets and tying/fairness constraints, and
(iii) a ``governed'' implementability characterization that will underwrite the single-level design
reductions in Section~\ref{sec:design}.

\subsection{Full implementability under salience policies}
We now show that the stationary salience model in Section~\ref{sec:salience} yields an explicit and very strong implementability guarantee.

\noindent\textbf{Comment on novelty.} The unconstrained ``full implementability'' statement below is, in hindsight, an explicit inversion of the salience-weighted logit equilibrium conditions. Its value in this paper is as a building block: it (i) cleanly separates what is a property of the reduced-form salience layer from what is micro-founded, and (ii) becomes nontrivial once we impose governance constraints (ratio budgets, tying/fairness, feature constraints) and seek tractable network-class design algorithms.

\begin{theorem}[Full implementability of interior flows under salience]
\label{thm:implementability-salience}
Consider a fixed network with continuous nondecreasing latencies and a fixed $\beta>0$.
Let $\bar f\in\cF$ be an \emph{interior} feasible flow, i.e., $\bar f_{k,p}>0$ for all $k$ and all $p\in\cP_k$.
Define salience weights by
\begin{equation}
\label{eq:inverse-salience}
s_{k,p}
\;\triangleq\;
\bar f_{k,p}\,\exp\!\Big(\beta\,L_p\big(x(\bar f)\big)\Big).
\end{equation}
Then $\bar f$ is the unique SW-SUE (Definition~\ref{def:swsue}) induced by $s$.
In particular, every interior feasible flow is implementable by stationary salience policies.
\end{theorem}

\begin{proof}
Let $s$ be defined by \eqref{eq:inverse-salience}.
Then
\[
s_{k,p}\exp\!\Big(-\beta L_p(x(\bar f))\Big) \;=\; \bar f_{k,p}.
\]
Summing over $p\in\cP_k$ yields $\sum_{p} s_{k,p}e^{-\beta L_p(x(\bar f))} = \sum_p \bar f_{k,p} = d_k$.
Substituting into \eqref{eq:salience-logit} gives
\[
d_k\,\Pr(p\mid k,x(\bar f);s)
\;=\;
d_k\,\frac{\bar f_{k,p}}{d_k}
\;=\;
\bar f_{k,p},
\]
so $\bar f$ satisfies the SW-SUE fixed point \eqref{eq:swsue}.
Uniqueness follows from Proposition~\ref{prop:swsue-potential}.
\end{proof}

\begin{corollary}[Implementing the system optimum without tolls]
\label{cor:implement-so}
If the system-optimal flow $f^{\mathrm{SO}}$ is interior, then it is implementable by a stationary salience policy via \eqref{eq:inverse-salience}.
If $f^{\mathrm{SO}}$ is not interior, it is $\varepsilon$-implementable for any $\varepsilon>0$ by perturbing $f^{\mathrm{SO}}$ to an interior flow and applying Theorem~\ref{thm:implementability-salience}.
\end{corollary}

\begin{proof}[Proof idea.]
If $f^{\mathrm{SO}}$ is interior, apply Theorem~\ref{thm:implementability-salience} directly. If not, perturb $f^{\mathrm{SO}}$ to an interior flow within $\varepsilon$; implementability of the perturbed flow and continuity of costs yield $\varepsilon$-implementability of $f^{\mathrm{SO}}$.
Full proof is deferred to Appendix~\ref{app:proof-cor-implement-so}.
\end{proof}

\subsection{Constrained implementability under influence budgets and governance}
\label{sec:constrained-impl}
Theorem~\ref{thm:implementability-salience} is intentionally strong: unconstrained salience can implement any interior flow.
Applications may impose \emph{governance} constraints on how strongly an AI system may bias or rank alternatives.
We formalize this via \emph{influence budgets} and derive sharp implementability characterizations.

\paragraph{Scale invariance and influence budgets.}
Only \emph{relative} salience matters: for a fixed commodity $k$, replacing $s_{k,p}$ by $c_k s_{k,p}$ for any constant $c_k>0$ leaves \eqref{eq:salience-logit} unchanged.
Accordingly, a natural constraint is a bound on within-commodity salience ratios.

\begin{definition}[Influence budget (bounded salience ratios)]
\label{def:influence-budget}
Fix $R_k\ge 1$ for each commodity $k$.
A salience vector $s$ satisfies the influence budget if
\begin{equation}
\label{eq:salience-budget}
\max_{p,r\in\cP_k}\frac{s_{k,p}}{s_{k,r}} \;\le\; R_k
\qquad\text{for all } k\in\cK.
\end{equation}
Equivalently, $\max_{p}\log s_{k,p}-\min_{p}\log s_{k,p}\le \log R_k$.
\end{definition}

\begin{theorem}[Exact constrained implementability under ratio budgets]
\label{thm:constrained-implementability}
Fix $\beta>0$ and continuous nondecreasing latencies.
Let $\bar f\in\cF$ be an interior feasible flow.
Define the \emph{required log-salience up to scale} by
\begin{equation}
\label{eq:req-log-salience}
a_{k,p}(\bar f) \;\triangleq\; \log \bar f_{k,p} + \beta\,L_p\!\big(x(\bar f)\big),
\qquad k\in\cK,\ p\in\cP_k.
\end{equation}
Then $\bar f$ is implementable by stationary salience policies satisfying the influence budgets \eqref{eq:salience-budget} if and only if
\begin{equation}
\label{eq:range-condition}
\max_{p\in\cP_k} a_{k,p}(\bar f)\;-\;\min_{p\in\cP_k} a_{k,p}(\bar f)
\;\le\; \log R_k
\qquad \text{for all } k\in\cK.
\end{equation}
Moreover, the minimal required budget for commodity $k$ is
$R_k^{\min}(\bar f) = \exp\!\big(\max_p a_{k,p}(\bar f)-\min_p a_{k,p}(\bar f)\big)$.
\end{theorem}

\begin{proof}
(\emph{If}.) Suppose \eqref{eq:range-condition} holds.
Set $s_{k,p}\triangleq \exp(a_{k,p}(\bar f))$.
Then $s_{k,p}\propto \bar f_{k,p}\exp(\beta L_p(x(\bar f)))$, so by the same calculation as in Theorem~\ref{thm:implementability-salience}, $\bar f$ is the unique SW-SUE induced by $s$.
Finally, $\log s_{k,p}=a_{k,p}(\bar f)$, so the ratio bound follows directly from \eqref{eq:range-condition}.

(\emph{Only if}.) If $\bar f$ is implementable by some $s$ satisfying \eqref{eq:salience-budget}, then at equilibrium
$\bar f_{k,p}\propto s_{k,p}\exp(-\beta L_p(x(\bar f)))$, i.e.,
$\log s_{k,p} = \log \bar f_{k,p} + \beta L_p(x(\bar f)) + c_k$
for some constant $c_k$ (normalization) depending on $k$ only.
Thus the range of $\log s_{k,p}$ over $p$ equals the range of $a_{k,p}(\bar f)$, so \eqref{eq:salience-budget} implies \eqref{eq:range-condition}.
\end{proof}

\begin{remark}[Menu-size and fairness constraints]
Beyond ratio budgets, governance may constrain the \emph{menu size} (how many routes can be surfaced) or impose \emph{fairness} (e.g., salience cannot depend on protected attributes).
In the stationary salience model, imposing $s_{k,p}=0$ forces $f_{k,p}^\star=0$.
Thus menu-size constraints translate into sparsity constraints on $s$ (bounded support), and fairness constraints translate into tying parameters across groups/commodities.
Theorem~\ref{thm:constrained-implementability} isolates a first-order, scale-invariant ``influence'' constraint that already yields nontrivial implementability geometry.
\end{remark}

\subsection{Implementability under tying and fairness constraints}
\label{sec:fair-impl}
Influence budgets (Definition~\ref{def:influence-budget}) constrain \emph{how much} an AI system can tilt attention.
A different and equally important class of governance constraints restricts \emph{which distinctions} the system is allowed to encode.
In particular, many fairness and compliance regimes require the guidance policy to be \emph{group-blind} (or to satisfy bounded disparity) with respect to protected attributes, which naturally induces \emph{parameter tying} across subpopulations.

We formalize such constraints as \emph{affine restrictions on log-salience}.
Let $u_{k,p}\triangleq \log s_{k,p}$ and stack these into a vector $u\in\R^{\sum_k|\cP_k|}$.
A broad class of constrained salience policies can be written as
\begin{equation}
\label{eq:tied-salience-class}
u \;\in\; \cU \;\triangleq\; \{A\theta + b : \theta\in\R^d\},
\end{equation}
where $A$ encodes tying or feature-based parameterization and $b$ is a fixed offset.
Examples include:
(i) \emph{group-blind} policies that force $u_{k,p}$ to be identical across protected groups $k$ that share the same route $p$;
(ii) \emph{feature-based} policies $u_{k,p}=\theta^\top \phi_{k,p}$ (shared $\theta$) that only depend on approved route features $\phi_{k,p}$.

Because salience is scale-invariant within each commodity, we also allow commodity-specific intercepts.
Let
\begin{equation}
\label{eq:intercept-space}
\cC \;\triangleq\; \{c\in\R^{\sum_k|\cP_k|} : c_{k,p}=c_k \ \text{for each }k\text{ and all }p\in\cP_k\}.
\end{equation}

\begin{theorem}[Exact implementability under affine tying constraints]
\label{thm:tied-implementability}
Fix $\beta>0$ and continuous nondecreasing latencies, and let $\bar f\in\cF$ be interior.
Define the \emph{required log-salience} vector (up to scale) by
\[
a(\bar f)_{k,p} \;\triangleq\; \log \bar f_{k,p} + \beta\,L_p\!\big(x(\bar f)\big).
\]
Then $\bar f$ is implementable by stationary salience policies with log-salience $u\in \cU$ if and only if
\begin{equation}
\label{eq:tied-membership}
a(\bar f) \;\in\; \cU + \cC.
\end{equation}
Equivalently, there exist $\theta\in\R^d$ and intercepts $\{c_k\}$ such that for all $k$ and $p\in\cP_k$,
\begin{equation}
\label{eq:tied-eq}
(A\theta+b)_{k,p} + c_k \;=\; \log \bar f_{k,p} + \beta\,L_p\!\big(x(\bar f)\big).
\end{equation}
\end{theorem}

\begin{proof}
At any SW-SUE induced by $u=\log s$, the fixed point implies
\[\bar f_{k,p}\propto \exp(u_{k,p})\exp(-\beta L_p(x(\bar f)))\], i.e.,
$u_{k,p} = \log \bar f_{k,p} + \beta L_p(x(\bar f)) + c_k$ for some commodity-specific constants $c_k$.
Thus $\bar f$ is implementable with $u\in\cU$ if and only if $a(\bar f)=u-c$ for some $u\in\cU$ and $c\in\cC$, i.e.\ $a(\bar f)\in \cU+\cC$.
\end{proof}

\begin{corollary}[Group-blind salience forces identical route shares]
\label{cor:group-blind}
Suppose commodities $k$ correspond to protected groups that share the same feasible route set $\cP$ and face the same latencies.
If the fairness constraint enforces \emph{group-blind} salience, i.e.\ $u_{k,p}=u_{k',p}$ for all groups $k,k'$ and all $p\in\cP$, then at any SW-SUE all groups induce the same route-share vector:
\[
\frac{f_{k,p}}{d_k} \;=\; \frac{f_{k',p}}{d_{k'}} \qquad\text{for all }k,k',\ p\in\cP.
\]
Consequently, any target flow that assigns different route shares across groups is not implementable under group-blind salience.
\end{corollary}

\begin{proof}[Proof idea.]
With identical costs and group-blind salience, each group faces the same salience-weighted logit rule at the same congestion, so their route-choice distributions coincide. Since each group's flow is its demand times this common distribution, all groups induce identical route shares.
Full proof is deferred to Appendix~\ref{app:proof-cor-group-blind}.
\end{proof}

\begin{remark}[Approximate implementability and projection]
When \eqref{eq:tied-membership} fails, a natural ``best-effort'' policy solves the convex regression problem
\[
\min_{\theta,c\in\cC}\ \|A\theta+b+c-a(\bar f)\|^2,
\]
or, more structurally, minimizes the equilibrium social cost over the constrained policy class $\cU$ using the sensitivity tools in Section~\ref{sec:design}.
This yields a quantitative \emph{distance to implementability} under governance constraints, and highlights when fairness restrictions make certain welfare targets infeasible.
\end{remark}

\begin{remark}[What is ``salience'' in a memory system?]
In an AI memory architecture, $s_{k,p}$ can be interpreted as a stationary \emph{retrieval intensity} for route $p$ (how often it is surfaced or made salient in the user's context) or, more broadly, as an alternative-specific bias induced by ranking, summarization, and retrieval.
Theorem~\ref{thm:implementability-salience} shows that, at least at the level of a stationary abstraction, memory can act as a full-fledged mechanism-design primitive.
\end{remark}

\subsection{Putting it together: governed implementability (tying \emph{and} influence budgets)}
\label{sec:governed-implementability}
In many deployments, governance combines both \emph{intensity} limits (influence budgets) and \emph{structure} limits (fairness/tying).
Theorem~\ref{thm:constrained-implementability} and Theorem~\ref{thm:tied-implementability} compose cleanly into a single geometric characterization.

For convenience, define the per-commodity range operator
\[
\mathrm{range}_k(v)\;\triangleq\;\max_{p\in\cP_k} v_{k,p}-\min_{p\in\cP_k} v_{k,p},
\]
and the governed set of admissible log-salience vectors (up to commodity-wise scale)
\begin{equation}
\label{eq:governed-set}
\cG(\cU,R)\;\triangleq\;
\Big\{\,u\in\cU+\cC:\ \mathrm{range}_k(u)\le \log R_k\ \ \forall k\in\cK\,\Big\}.
\end{equation}

\begin{theorem}[Governed implementability under affine tying and ratio budgets]
\label{thm:governed-implementability}
Fix $\beta>0$ and continuous nondecreasing latencies.
Let $\cU=\{A\theta+b:\theta\in\R^d\}$ encode the policy's allowed log-salience structure, and let $R=\{R_k\}_{k\in\cK}$ encode per-commodity influence budgets.
For any interior feasible flow $\bar f\in\cF$, define
$a(\bar f)_{k,p}\triangleq \log \bar f_{k,p} + \beta\,L_p(x(\bar f))$.
Then $\bar f$ is implementable by a stationary salience policy that satisfies both (i) the tying constraint $u\in\cU$ and (ii) the ratio budgets \eqref{eq:salience-budget} if and only if
\begin{equation}
\label{eq:governed-membership}
a(\bar f)\;\in\;\cG(\cU,R).
\end{equation}
\end{theorem}

\begin{proof}
By Theorem~\ref{thm:tied-implementability}, $\bar f$ is implementable under tying if and only if $a(\bar f)\in\cU+\cC$.
By Theorem~\ref{thm:constrained-implementability}, $\bar f$ is implementable under ratio budgets if and only if $\mathrm{range}_k(a(\bar f))\le \log R_k$ for all $k$.
Because adding commodity-wise intercepts does not change ranges, the two constraints are compatible and their conjunction is exactly \eqref{eq:governed-membership}.
\end{proof}

\begin{theorem}[Governed salience mechanism design reduces to implementable-flow optimization]
\label{thm:governed-design-reduction}
Consider the governed salience design problem
\begin{equation}
\label{eq:bilevel-governed}
\min_{u\in \cG(\cU,R)}\ \mathrm{SC}\big(f^\star(u)\big),
\end{equation}
where $u=\log s$ is the log-salience vector, $\cG(\cU,R)$ is the governed feasible set \eqref{eq:governed-set}, and
$f^\star(u)$ is the (unique) SW-SUE induced by $u$ (Proposition~\ref{prop:swsue-potential}).

Assume equilibria are interior ($f^\star(u)\in\cF^\circ$) for all feasible $u$.
Then \eqref{eq:bilevel-governed} is equivalent to the single-level program
\begin{equation}
\label{eq:single-level-design}
\min_{f\in\cF^\circ}\ \mathrm{SC}(f)
\qquad\text{s.t.}\qquad
a(f)\in \cG(\cU,R),
\end{equation}
and any optimizer $f^\star$ of \eqref{eq:single-level-design} can be implemented by some $u^\star\in\cG(\cU,R)$ via the inverse formula \eqref{eq:inverse-salience} (up to commodity-wise intercepts).
\end{theorem}

\begin{proof}
If $u\in\cG(\cU,R)$ is feasible in \eqref{eq:bilevel-governed}, then $f^\star(u)$ is implementable under the same governance constraints.
By Theorem~\ref{thm:governed-implementability}, this implies $a(f^\star(u))\in\cG(\cU,R)$, so $f^\star(u)$ is feasible for \eqref{eq:single-level-design}.
Conversely, if $f\in\cF^\circ$ satisfies $a(f)\in\cG(\cU,R)$, then again by Theorem~\ref{thm:governed-implementability} it is implementable by some governed log-salience $u\in\cG(\cU,R)$.
Because the SW-SUE induced by a fixed $u$ is unique (Proposition~\ref{prop:swsue-potential}), the induced equilibrium must equal $f$.
Therefore the feasible objective values of \eqref{eq:bilevel-governed} and \eqref{eq:single-level-design} coincide, and the constructions above map optimizers to optimizers.
\end{proof}

\begin{remark}[Geometry vs.\ tractability]
Theorem~\ref{thm:governed-design-reduction} separates (i) network physics, via $L_p(x(f))$, from (ii) governance geometry, via $\cG(\cU,R)$.
On general networks, \eqref{eq:single-level-design} remains nonconvex because of the $\log f$ term and the nonlinear dependence of $L$ on $f$; however, on certain network classes it becomes tractable (Section~\ref{sec:parallel-design}).
\end{remark}

\subsection{A network-class theorem: series-parallel implementability under decomposition-tied salience}
\label{sec:sp-implementability}
The full implementability result (Theorem~\ref{thm:implementability-salience}) treats the log-salience vector $u=\log s$ as fully flexible at the \emph{path} level.
In many networks, however, the number of $s$--$t$ paths is enormous, and governance may prefer a \emph{structured} control surface.
On the two-terminal series-parallel (SP) class, the network itself provides a natural low-dimensional parameterization: \emph{local biases at parallel composition nodes}.

\paragraph{Decomposition-tied salience.}
Fix a single commodity (one OD pair) on a two-terminal SP network with an SP decomposition tree $\mathcal{T}$.
Each internal node of $\mathcal{T}$ is either a series composition ($\otimes$) or a parallel composition ($\oplus$) of two subgraphs.
For each parallel node $v\in\mathcal{T}$ with children $(v_L,v_R)$, introduce a \emph{relative} log-salience parameter $\delta_v\in\R$ that biases the left child vs.\ the right child.
Operationally, $\delta_v$ is a ``branch-level'' salience signal: how strongly the platform steers users into the left subnetwork relative to the right at that decision point.

Formally, under \emph{decomposition-tied} salience, the unnormalized logit weight of any $s$--$t$ path $p$ is
\[
\exp\!\big(-\beta L_p(x)\big)\cdot \exp\!\Big(\sum_{v\in\mathcal{V}_\oplus(p)} \delta_v\Big),
\]
where $\mathcal{V}_\oplus(p)$ is the set of parallel nodes at which $p$ takes the left branch.
Equivalently, each parallel node contributes a multiplicative factor $e^{\delta_v}$ to all paths that go left at that node, and a factor $1$ to all paths that go right.
(This fixes the right-branch bias to zero. A common shift to both children cancels in the conditional split at that node, but at a nested node it can change the probability of entering that subnetwork. It is not a global invariance.)

\paragraph{Inclusive values on an SP tree.}
For any edge-flow vector $x$ and any set of $\{\delta_v\}$, define the (log) inclusive value of a subgraph $H$ in the SP tree by
\[
V_H(x,\delta) \;\triangleq\; \log\!\sum_{p\in\cP(H)} \exp\!\big(-\beta L_p(x)\big)\cdot \exp\!\Big(\sum_{v\in\mathcal{V}_\oplus(p)} \delta_v\Big),
\]
where $\cP(H)$ denotes the $s$--$t$ paths \emph{within} subgraph $H$.
The inclusive values satisfy the standard SP recursion:
for a leaf edge $e$, $V_e=-\beta \ell_e(x_e)$; for a series node $H=H_1\otimes H_2$, $V_H=V_{H_1}+V_{H_2}$; and for a parallel node $H=H_L\oplus H_R$ with relative bias $\delta_H$ on the left branch,
\begin{equation}
\label{eq:sp-inclusive-parallel}
V_H \;=\; \log\!\Big(e^{\delta_H+V_{H_L}} + e^{V_{H_R}}\Big).
\end{equation}

\begin{theorem}[Constructive implementability on series-parallel networks under decomposition-tied salience]
\label{thm:sp-implementability}
Fix $\beta>0$ and a single-commodity two-terminal SP network with decomposition tree $\mathcal{T}$.
Let $\bar x$ be any \emph{strictly interior} feasible $s$--$t$ edge flow, meaning that at every parallel node $H=H_L\oplus H_R$ in $\mathcal{T}$, the induced branch flows $\bar d_{H_L},\bar d_{H_R}$ satisfy $\bar d_{H_L}>0$ and $\bar d_{H_R}>0$.
Then there exists a vector of branch-bias parameters $\delta=\{\delta_H\}$ such that $\bar x$ is the unique SW-SUE edge flow induced by decomposition-tied salience $\delta$.
Moreover, $\delta$ can be computed in $O(|E|)$ time by a single bottom-up pass on $\mathcal{T}$ via the local inverse formula
\begin{equation}
\label{eq:sp-inverse-delta}
\delta_H
\;=\;
\log\frac{\bar d_{H_L}}{\bar d_{H_R}}
\;-\;
\Big(V_{H_L}(\bar x,\delta)-V_{H_R}(\bar x,\delta)\Big),
\qquad \text{for each parallel node }H.
\end{equation}
\end{theorem}

\begin{proof}
Fix an interior feasible target flow $\bar x$.
Because the network is two-terminal SP with decomposition tree $\mathcal{T}$, $\bar x$ induces a well-defined \emph{subgraph demand} $\bar d_H$ for every subgraph node $H\in\mathcal{T}$:
at the root, $\bar d_{\text{root}}=d$; at a series node $H=H_1\otimes H_2$, we have $\bar d_{H_1}=\bar d_{H_2}=\bar d_H$; and at a parallel node $H=H_L\oplus H_R$, feasibility implies a flow split
$\bar d_H=\bar d_{H_L}+\bar d_{H_R}$ with $\bar d_{H_L},\bar d_{H_R}>0$ by the interior assumption.

We compute $(\delta,V)$ bottom-up on $\mathcal{T}$ as follows.
For a leaf edge $e$, set $V_e=-\beta \ell_e(\bar x_e)$.
For a series node $H=H_1\otimes H_2$, set $V_H=V_{H_1}+V_{H_2}$.
For a parallel node $H=H_L\oplus H_R$, first assume $V_{H_L},V_{H_R}$ have been computed; then set
\[
\delta_H
=
\log\frac{\bar d_{H_L}}{\bar d_{H_R}}
-
\big(V_{H_L}-V_{H_R}\big),
\qquad
V_H=\log\!\big(e^{\delta_H+V_{H_L}}+e^{V_{H_R}}\big).
\]
This is exactly \eqref{eq:sp-inverse-delta} and \eqref{eq:sp-inclusive-parallel}.
The computation visits each tree node once and is therefore $O(|E|)$.

Now define the induced salience weights on $s$--$t$ paths by
$s_p\triangleq \exp\!\big(\sum_{v\in\mathcal{V}_\oplus(p)}\delta_v\big)$.
Consider the logit path-flow induced by $(\bar x,s)$:
\begin{equation}
\label{eq:sp-induced-pathflow}
\bar f_p
\;\triangleq\;
d\cdot
\frac{s_p\,\exp(-\beta L_p(\bar x))}{\sum_{q\in\cP}\,s_q\,\exp(-\beta L_q(\bar x))}.
\end{equation}
We claim that the resulting edge flow satisfies $x(\bar f)=\bar x$.
To show this, it is convenient to work recursively on $\mathcal{T}$.
For any subgraph node $H\in\mathcal{T}$, let $\cP(H)$ denote its internal $s$--$t$ paths and define the \emph{conditional} logit distribution within $H$ by
\[
\bar q_H(p)
\;\triangleq\;
\frac{\exp(-\beta L_p(\bar x))\,\exp(\sum_{v\in\mathcal{V}_\oplus(p)}\delta_v)}{\exp(V_H)},
\qquad p\in\cP(H),
\]
where $V_H$ is the inclusive value computed above.
By construction, $\bar q_H$ is a probability distribution on $\cP(H)$.

We prove by induction on $H$ that $\bar q_H$ routes exactly the target subflow $\bar d_H$ through subgraph $H$:
for every leaf edge $e$ in the subtree of $H$, the induced edge flow equals $\bar x_e$.
The base case is a leaf edge: $\cP(e)$ contains a single path and thus all flow $\bar d_e$ traverses $e$, matching $\bar x_e$.

For a series node $H=H_1\otimes H_2$, every path in $H$ is a concatenation $p=p_1\otimes p_2$ with $p_i\in\cP(H_i)$.
Because costs add and biases in the two subtrees are disjoint, we have
\[
\bar q_H(p_1\otimes p_2)
=
\frac{e^{-\beta(L_{p_1}+L_{p_2})}\,e^{(\sum\delta\text{ in }H_1)+(\sum\delta\text{ in }H_2)}}{e^{V_{H_1}+V_{H_2}}}
=
\bar q_{H_1}(p_1)\,\bar q_{H_2}(p_2),
\]
i.e., the conditional distribution factorizes.
Since the total flow through each child equals $\bar d_H$ in a series composition, the induction hypothesis implies the induced edge flows within each child match $\bar x$ restricted to that child.

For a parallel node $H=H_L\oplus H_R$, every path is either in the left child or the right child.
Summing the unnormalized weights over $\cP(H_L)$ and $\cP(H_R)$ gives
\[
\Pr_{\bar q_H}\{\text{choose }H_L\}
=
\frac{e^{\delta_H+V_{H_L}}}{e^{\delta_H+V_{H_L}}+e^{V_{H_R}}},
\qquad
\Pr_{\bar q_H}\{\text{choose }H_R\}
=
\frac{e^{V_{H_R}}}{e^{\delta_H+V_{H_L}}+e^{V_{H_R}}}.
\]
By the inverse choice of $\delta_H$, we have $e^{\delta_H+V_{H_L}-V_{H_R}}=\bar d_{H_L}/\bar d_{H_R}$, hence
$\Pr_{\bar q_H}\{H_L\}=\bar d_{H_L}/(\bar d_{H_L}+\bar d_{H_R})$.
Multiplying by the parent flow $\bar d_H=\bar d_{H_L}+\bar d_{H_R}$ yields exactly the target branch flows $\bar d_{H_L}$ and $\bar d_{H_R}$.
Conditioned on choosing $H_L$ (resp.\ $H_R$), the distribution within that child is exactly $\bar q_{H_L}$ (resp.\ $\bar q_{H_R}$), so the induction hypothesis ensures the internal edge flows match $\bar x$ on each side.
This completes the induction.

Applying the induction to the root node shows that the full path flow \eqref{eq:sp-induced-pathflow} induces $x(\bar f)=\bar x$.
But \eqref{eq:sp-induced-pathflow} is exactly the SW-SUE fixed-point condition \eqref{eq:swsue} at $x=\bar x$ with salience weights $s$.
Therefore $\bar x$ is an SW-SUE edge flow under decomposition-tied salience.
Uniqueness follows from strict convexity of the SW-SUE potential (Proposition~\ref{prop:swsue-potential}).
\end{proof}

\begin{corollary}[Governance as local budget constraints on SP trees]
\label{cor:sp-local-budgets}
Suppose governance restricts branch-level influence at each parallel node $H$ to a ratio budget $R_H\ge 1$, i.e.\ $|\delta_H|\le \log R_H$.
Then an interior target flow $\bar x$ is implementable under these budgets if and only if its induced branch splits satisfy the local inequalities
\[
\Big|\log\frac{\bar d_{H_L}}{\bar d_{H_R}}
-
\big(V_{H_L}(\bar x,\delta)-V_{H_R}(\bar x,\delta)\big)\Big|
\le \log R_H
\qquad \text{for all parallel nodes }H,
\]
where $\delta$ is the recursively defined inverse-bias vector from \eqref{eq:sp-inverse-delta}.
Consequently, implementability testing under per-node budgets is linear-time in $|E|$ on SP networks.
\end{corollary}

\begin{proof}[Proof idea.]
The SP implementability theorem expresses the required bias $\delta_H$ at each parallel node as the log ratio between target split masses corrected by downstream soft-values. Imposing $|\delta_H|\le \log R_H$ is therefore equivalent to a pair of local inequalities. Conversely, if the inequalities hold, choosing admissible $\delta_H$ and reconstructing edge-additive salience implements the target flow.
Full proof is deferred to Appendix~\ref{app:proof-cor-sp-local-budgets}.
\end{proof}

\begin{remark}[Why this is useful]
Theorem~\ref{thm:sp-implementability} identifies a \emph{network class} on which implementability is both constructive and low-dimensional: the control variables live at parallel splits rather than at the path level.
This matches UI-level intervention points (``take the highway vs.\ the arterial''), and it yields a direct bridge to algorithmic design:
one can optimize over SP-tree split patterns subject to local budget constraints without enumerating an exponential path set.
\end{remark}

Section~\ref{sec:design} returns to richer policy classes (explicit memory kernels and surfacing distributions) and to equilibrium-aware optimization under governance constraints.

\noindent Governance constraints here describe allowed salience parameters. They do not by themselves establish individual non-harm, equitable welfare, or privacy. Hard route exclusions require restricting the admissible path set before applying formulas with positive salience.

\section{Equilibrium-Aware Memory and Guidance Design}
\label{sec:design}
Sections~\ref{sec:salience} and~\ref{sec:implementability} provide (i) a strictly convex equilibrium
characterization for SW-SUE under stationary salience policies and (ii) sharp implementability and
governance constraints.
We now formulate and solve \emph{equilibrium-aware design} problems: choose a policy within a
governance class (influence budgets, tying/fairness, or both) to optimize welfare at the induced
equilibrium.

The main technical theme is that strict convexity enables efficient single-level reductions and
first-order methods (including implicit differentiation of the equilibrium map). Later we return
to the explicit memory model and identify assumptions needed to approximate its interventions by
salience designs.
\subsection{Policy parameterization}
Let $\theta\in\Theta$ parameterize a class of policies.
A policy may control, for each commodity $k$:
\begin{itemize}
    \item a surfacing distribution $\rho_{k,\theta}$ over candidate routes (which routes are suggested or explored), and
    \item a memory update rule (eviction/forgetting) $U_{k,\theta}$ or, more generally, a memory kernel $Q_{k,\theta}(m'\mid m,p)$ that may include resets, summarization, or privacy-driven deletion.
\end{itemize}
This mirrors real AI memory architectures in which the system selects what to store, what to summarize, and what to discard under explicit budgets and governance constraints.

Given $\theta$, the induced FWE (when unique) is $(x^\star(\theta),\mu^\star(\theta))$.

\subsection{Bilevel optimization objective}
We consider an equilibrium-aware objective of the form
\begin{equation} 
\label{eq:bilevel}
\min_{\theta\in\Theta}\;\; \underbrace{\mathrm{SC}\!\left(x^\star(\theta)\right)}_{\text{equilibrium congestion}}
\;+\;\lambda\,\underbrace{\mathrm{Cost}(\theta)}_{\text{compute / privacy / UX}}
\qquad \text{s.t.}\qquad
\underbrace{\mathrm{Fair}(\theta)\le 0}_{\text{fairness / individual harm constraints}}.
\end{equation}
Here $\mathrm{Cost}(\theta)$ can encode, e.g., expected memory usage, summarization overhead, or a privacy-risk proxy; and $\mathrm{Fair}(\theta)$ can enforce constraints such as ``no user type is made worse off by the policy'' or ``bounded disparity across groups.''

\subsection{A closed-form optimal design on a canonical network class}
\label{sec:optimal-pigou}
To emphasize that the ``memory as mechanism'' perspective yields analytically tractable design problems (not only numerical bilevel programs), we give a closed-form optimal salience result on the canonical Pigou network.
This provides a baseline for more complex network classes.

\paragraph{Pigou instance.}
Consider a single OD pair ($d=1$) with two parallel routes $a$ and $b$.
Route $a$ has constant latency $\ell_a(x)=1$ and route $b$ has latency $\ell_b(x)=x$.
Let $\beta>0$ and consider stationary salience policies with weights $(s_a,s_b)$; only the ratio $r\triangleq s_b/s_a$ matters.

Under SW-SUE, the equilibrium flow on route $b$ is a scalar $f_b\in(0,1)$ satisfying
\begin{equation}
\label{eq:pigou-ratio}
f_b
=
\frac{r\,\exp(-\beta f_b)}{\exp(-\beta)+r\,\exp(-\beta f_b)}
\quad\Longleftrightarrow\quad
r
=
\frac{f_b}{1-f_b}\exp\!\big(\beta(f_b-1)\big).
\end{equation}
The equilibrium social cost is
\begin{equation}
\label{eq:pigou-sc}
\mathrm{SC}(f_b) \;=\; (1-f_b)\cdot 1 + f_b^2 \;=\; 1-f_b+f_b^2,
\end{equation}
which is uniquely minimized at the system-optimal split $f_b^{\mathrm{SO}}=1/2$.

\begin{theorem}[Optimal bounded-influence salience on the Pigou network]
\label{thm:optimal-pigou}
Fix $\beta>0$ and an influence budget $R\ge 1$ so that $r\in[1/R,R]$.
Then:
\begin{enumerate}
\item The mapping $f_b\mapsto r(f_b)=\frac{f_b}{1-f_b}\exp(\beta(f_b-1))$ is strictly increasing on $(0,1)$.
Hence for each $r>0$ there is a unique equilibrium $f_b(r)\in(0,1)$.
\item Let $f_{\min}\triangleq f_b(1/R)$ and $f_{\max}\triangleq f_b(R)$.
The social-cost-minimizing feasible equilibrium split is
\begin{equation}
\label{eq:clamp}
f_b^\star \;=\; \mathrm{clip}\!\left(\tfrac12;\; f_{\min},\,f_{\max}\right)
\;\triangleq\;
\min\!\Big\{\max\!\big\{\tfrac12,f_{\min}\big\},\,f_{\max}\Big\},
\end{equation}
and the optimal salience ratio is $r^\star=r(f_b^\star)$.
\item The system optimum $f_b^{\mathrm{SO}}=\tfrac12$ is exactly implementable under budget $R$ if and only if $R\ge \exp(\beta/2)$ (equivalently, $1/R \le \exp(-\beta/2)$), in which case an optimal ratio is $r^\star=\exp(-\beta/2)$.
\end{enumerate}
\end{theorem}

\begin{proof}
(1) Differentiate $\log r(f_b)=\log f_b-\log(1-f_b)+\beta(f_b-1)$, giving
\[
\frac{d}{df_b}\log r(f_b) = \frac{1}{f_b}+\frac{1}{1-f_b}+\beta>0,
\]
so $r(f_b)$ is strictly increasing.

(2) Because $f_b(r)$ is increasing in $r$ by (1), the feasible equilibria correspond to $f_b\in[f_{\min},f_{\max}]$.
Since $\mathrm{SC}(f_b)$ in \eqref{eq:pigou-sc} is convex with unique minimizer $1/2$, the constrained minimizer over the interval is the projection (clipping) of $1/2$ onto that interval, yielding \eqref{eq:clamp}.

(3) Plugging $f_b=1/2$ into \eqref{eq:pigou-ratio} yields the unique ratio $r^{\mathrm{SO}}=\exp(-\beta/2)$.
Feasibility under $r\in[1/R,R]$ requires $1/R \le r^{\mathrm{SO}}\le R$, and since $r^{\mathrm{SO}}\le 1$, this reduces to $R\ge \exp(\beta/2)$.
\end{proof}

\begin{remark}
Theorem~\ref{thm:optimal-pigou} makes explicit how a governance knob (influence budget $R$) controls how closely a salience-based memory/guidance layer can approximate the system optimum.
In richer networks, the same logic applies but without closed-form inversion; the strict convexity of $\Phi_s$ (Proposition~\ref{prop:swsue-potential}) still yields a well-posed inner problem, and Theorem~\ref{thm:constrained-implementability} provides a tractable feasibility test for exact implementation of a candidate target flow.
\end{remark}

\subsection{Parallel networks: bounded-influence design beyond Pigou}
\label{sec:parallel-design}
The Pigou network (two parallel routes) admits a closed-form inversion (Theorem~\ref{thm:optimal-pigou}).
We now extend this to a broader \emph{network class}: a single OD pair with $m\ge 2$ \emph{parallel} routes, each with its own latency function.
This is a natural next step because parallel networks isolate the essence of ``shortcut'' competition without path-overlap complications, and because many real mobility decisions (e.g., highway vs.\ arterial) are well-approximated by a small set of parallel alternatives.

\paragraph{Model.}
Consider $m$ parallel routes indexed by $i\in\{1,\dots,m\}$ with demand normalized to $1$.
Let $f_i\in(0,1)$ denote the equilibrium flow on route $i$, so $\sum_i f_i=1$, and latency is $\ell_i(f_i)$ (continuous, nondecreasing).
Fix $\beta>0$ and a salience ratio budget $R\ge 1$ so that $u_i=\log s_i$ satisfies $\max_i u_i-\min_i u_i\le \log R$.

Define the strictly increasing transform
\begin{equation}
\label{eq:gi}
g_i(z)\;\triangleq\;\log z+\beta\,\ell_i(z),\qquad z\in(0,1).
\end{equation}
(Strict monotonicity follows because $g_i'(z)=1/z+\beta\,\ell_i'(z)>0$ whenever $\ell_i$ is differentiable, and holds more generally in the sense of monotone derivatives for continuous nondecreasing $\ell_i$.)

\begin{proposition}[Implementable equilibria on parallel networks under ratio budgets]
\label{prop:parallel-impl}
An interior flow vector $f\in(0,1)^m$ with $\sum_i f_i=1$ is implementable as the SW-SUE induced by some salience weights $s$ satisfying the ratio budget $R$ if and only if
\begin{equation}
\label{eq:parallel-range}
\max_{i} g_i(f_i)\;-\;\min_{i} g_i(f_i)\;\le\;\log R.
\end{equation}
Moreover, whenever \eqref{eq:parallel-range} holds, one feasible implementing policy is
\begin{equation}
\label{eq:parallel-inverse}
s_i \;\propto\; f_i\,\exp\!\big(\beta\,\ell_i(f_i)\big),
\end{equation}
with the proportionality constant chosen so that $\max_i \log s_i-\min_i\log s_i\le \log R$ (which is possible exactly when \eqref{eq:parallel-range} holds).
\end{proposition}

\begin{proof}
This is Theorem~\ref{thm:constrained-implementability} specialized to a single commodity with parallel routes.
The required log-salience up to scale is $a_i(f)=\log f_i+\beta\ell_i(f_i)=g_i(f_i)$.
The ratio budget requires the range of $a_i(f)$ to be at most $\log R$, which is \eqref{eq:parallel-range}.
The inverse construction \eqref{eq:parallel-inverse} follows from \eqref{eq:inverse-salience}.
\end{proof}

\paragraph{A constructive reduction of optimal design to one-dimensional search.}
The social cost on a parallel network is
\[
\mathrm{SC}(f)=\sum_{i=1}^m f_i\,\ell_i(f_i).
\]
Under very mild conditions (e.g., each $\ell_i$ nondecreasing), $\mathrm{SC}$ is continuous; if each $\ell_i$ is convex, then $\mathrm{SC}$ is convex in $f$.

By Proposition~\ref{prop:parallel-impl}, bounded-influence salience design on a parallel network reduces to
\begin{equation}
\label{eq:parallel-design-problem}
\min_{f\in(0,1)^m}\ \mathrm{SC}(f)
\quad \text{s.t.}\quad \sum_{i=1}^m f_i=1,\qquad
\max_i g_i(f_i)-\min_i g_i(f_i)\le \log R.
\end{equation}
While \eqref{eq:parallel-design-problem} is not obviously convex in full generality, it admits a useful structural decomposition.

\begin{theorem}[One-dimensional reduction for bounded-influence optimal design on parallel networks]
\label{thm:parallel-1d}
Assume each latency is continuous and nondecreasing on $[0,1]$, so $g_i$ is continuous and strictly increasing on $(0,1]$.
Then the constraint \eqref{eq:parallel-range} holds if and only if there exists a scalar $t\in\R$ such that for all $i$,
\begin{equation}
\label{eq:band}
t \;\le\; g_i(f_i) \;\le\; t+\log R.
\end{equation}
Restrict $t\le\min_i g_i(1)$ and discard values for which $\sum_i\underline f_i(t)>1$ or $\sum_i\overline f_i(t)<1$. For each remaining $t$, define bounds
\begin{equation}
\label{eq:bounds}
\underline f_i(t)\triangleq g_i^{-1}(t),\qquad \overline f_i(t)\triangleq g_i^{-1}\!\big(\min\{t+\log R,g_i(1)\}\big),
\end{equation}
and consider the inner problem
\begin{equation}
\label{eq:inner-t}
V(t)\;\triangleq\;\min_{f\in\R^m}\ \mathrm{SC}(f)\quad \text{s.t.}\quad
\sum_i f_i=1,\qquad \underline f_i(t)\le f_i\le \overline f_i(t)\ \ \forall i.
\end{equation}
Then the optimal value of \eqref{eq:parallel-design-problem} equals $\min_t V(t)$, and any minimizer $f^\star$ of \eqref{eq:parallel-design-problem} is obtained by solving \eqref{eq:inner-t} at some $t^\star\in\arg\min_t V(t)$.

If, in addition, each $\ell_i$ is convex, then for every fixed $t$ the inner problem \eqref{eq:inner-t} is a convex program. Its minimizer is unique if the social-cost objective is strictly convex on its feasible set. Convexity of latencies alone does not imply uniqueness.
\end{theorem}

\begin{proof}
The band representation \eqref{eq:band} is equivalent to the range condition \eqref{eq:parallel-range} by taking $t=\min_i g_i(f_i)$.
Because each $g_i$ is strictly increasing, \eqref{eq:band} is equivalent to the interval constraints $\underline f_i(t)\le f_i\le \overline f_i(t)$, yielding \eqref{eq:inner-t}.
Taking the minimum over $t$ recovers \eqref{eq:parallel-design-problem}.

When $\ell_i$ is convex, the function $f_i\mapsto f_i\ell_i(f_i)$ is convex, hence $\mathrm{SC}(f)$ is convex, and the feasible set in \eqref{eq:inner-t} is a polytope (simplex with box constraints), so \eqref{eq:inner-t} is convex.
Strict convexity of the sum on feasible directions yields uniqueness. A sufficient condition is strict convexity of every $z\ell_i(z)$ on the relevant interval; strict convexity of just one coordinate is not generally sufficient.
\end{proof}

\begin{proposition}[Optimality characterization on the parallel network class: clipped marginal-cost equalization]
\label{prop:parallel-kkt}
Assume each latency $\ell_i$ is continuously differentiable and convex.
Fix any band parameter $t$ for which the inner problem \eqref{eq:inner-t} is feasible, and let $f^\star(t)$ be any minimizer.
Then there exists a scalar $\lambda(t)\in\R$ such that, for every route $i$,
\begin{equation}
\label{eq:parallel-kkt}
m_i\!\big(f_i^\star(t)\big)\;\begin{cases}
= \lambda(t), & \text{if }\ \underline f_i(t) < f_i^\star(t) < \overline f_i(t),\\[2pt]
\ge \lambda(t), & \text{if }\ f_i^\star(t)=\underline f_i(t),\\[2pt]
\le \lambda(t), & \text{if }\ f_i^\star(t)=\overline f_i(t),
\end{cases}
\end{equation}
Coordinates with equal lower and upper bounds are fixed and are excluded from the displayed cases. The function $m_i(z)\triangleq \ell_i(z)+z\,\ell_i'(z)$ denotes the marginal social cost on route~$i$.

Consequently, any globally optimal solution of the bounded-influence design problem \eqref{eq:parallel-design-problem} equalizes marginal social costs across all \emph{non-saturated} routes, and ``clips'' this equalization only when forced by the salience band constraint \eqref{eq:band}.
\end{proposition}

\begin{proof}
For fixed $t$, the inner problem \eqref{eq:inner-t} is a convex program on a simplex with box constraints.
The KKT conditions yield the existence of a Lagrange multiplier $\lambda(t)$ for the equality constraint $\sum_i f_i=1$ and nonnegative multipliers for the box constraints; these imply \eqref{eq:parallel-kkt} by complementary slackness.
\end{proof}

\begin{remark}[Algorithmic implication]
Theorem~\ref{thm:parallel-1d} reduces bounded-influence optimal salience design on a parallel network to a one-dimensional outer search in $t$ plus an efficiently solvable inner convex program \eqref{eq:inner-t}.
Once an optimal $f^\star$ is computed, an implementing salience policy follows from \eqref{eq:parallel-inverse}.
This gives an exact reduction, but a global guarantee still requires solving the scalar outer problem globally. Unimodality and polynomial complexity of that outer search are not established here.
\end{remark}

\begin{remark}[Extending the network-class theorem to affine tying]
Theorem~\ref{thm:parallel-1d} gives an exact and computationally efficient characterization for \emph{unconstrained} salience vectors under a pure ratio budget.
If the policy class further imposes affine tying $u=A\theta+b$ (Section~\ref{sec:governed-implementability}), then implementability on a parallel network becomes the feasibility of
\[
g_i(f_i)=u_i+c\quad\text{for some}\quad u\in\cU,\ c\in\R,
\qquad\text{together with}\qquad \max_i u_i-\min_i u_i\le \log R.
\]
These are generally nonlinear equalities in $f$ because $g_i$ contains $\log f_i$ and latency terms. The pure ratio-budget box reduction does not make arbitrary affine tying convex in flow variables. Efficient algorithms for special tying classes require separate analysis.

\end{remark}

\subsection{Series-parallel networks: fast equilibrium evaluation and gradients without path enumeration}
\label{sec:sp-design}
The parallel-network class admits a particularly sharp characterization because each route's latency depends only on its own flow.
We can go beyond parallel networks by exploiting structural decomposability of the path set.
A prominent example is the two-terminal \emph{series-parallel} (SP) class, which admits an SP-tree decomposition and supports dynamic-programming evaluation of logit path distributions.

\paragraph{Edge-additive salience (feature tying).}
To avoid explicit path-level salience vectors on exponentially many paths, we consider a policy class in which log-salience is \emph{edge-additive}:
there exists a vector $u=(u_e)_{e\in E}$ such that
\begin{equation}
\label{eq:edge-additive-salience}
\log s_p \;=\; \sum_{e\in p} u_e.
\end{equation}
This is a natural tying constraint: the platform attaches persistent salience biases to links (or to low-dimensional route features that sum along a route), which is typical of ranking and guidance systems.

Under \eqref{eq:edge-additive-salience}, the salience-weighted logit distribution over paths can be written in terms of \emph{generalized edge costs}
\[
c_e(x,u) \;\triangleq\; \ell_e(x_e) - \frac{u_e}{\beta},
\]
since
\[
s_p\exp(-\beta L_p(x)) = \exp\Big(-\beta\sum_{e\in p} c_e(x,u)\Big).
\]

\paragraph{A Gibbs distribution over $s$--$t$ paths and its partition function.}
For fixed $(x,u)$, define the partition function
\begin{equation}
\label{eq:partition}
Z(x,u) \;\triangleq\; \sum_{p\in\cP} \exp\Big(-\beta\sum_{e\in p} c_e(x,u)\Big).
\end{equation}
The induced path distribution is $P(p)=Z(x,u)^{-1}\exp(-\beta\sum_{e\in p} c_e(x,u))$.
The marginal probability that edge $e$ is used equals
\begin{equation}
\label{eq:edge-marginal}
\pi_e(x,u) \;\triangleq\; \Pr\{e\in p\} \;=\; -\frac{1}{\beta}\frac{\partial}{\partial c_e}\log Z(x,u),
\end{equation}
and thus the expected edge load under demand $d$ is $x'_e=d\,\pi_e(x,u)$.

\begin{definition}[Two-terminal series-parallel network]
A directed two-terminal network $(G,s,t)$ is \emph{series-parallel} if it can be obtained from a single edge by recursively applying series composition and parallel composition of two-terminal subnetworks.
\end{definition}

\begin{theorem}[Linear-time evaluation of partition functions and edge marginals on SP networks]
\label{thm:sp-partition}
Let $(G,s,t)$ be a two-terminal series-parallel network with an SP decomposition tree of size $O(|E|)$.
For any generalized edge costs $\{c_e\}$, the partition function $Z$ in \eqref{eq:partition} and all edge marginals $\{\pi_e\}$ in \eqref{eq:edge-marginal} can be computed in $O(|E|)$ time.
\end{theorem}

\begin{proof}[Proof idea.]
For series-parallel graphs, the path partition function factorizes: series composition multiplies partition functions and parallel composition adds them. A bottom-up pass computes $Z$; edge marginals are obtained by differentiating $\log Z$ through the same recursion (equivalently, a reverse-mode pass on the SP tree).
Full proof is deferred to Appendix~\ref{app:proof-thm-sp-partition}.
\end{proof}

\begin{remark}[Computing SW-SUE without enumerating paths]
Theorem~\ref{thm:sp-partition} provides an efficient oracle for the map $(x,u)\mapsto x'$ defined by $x'_e=d\,\pi_e(x,u)$.
Together with a suitable convex solver, it avoids path enumeration in equilibrium computation. The $O(|E|)$ bound applies to one oracle evaluation, not to the total number of optimization iterations, projection cost, or numerical bit complexity.
This extends the tractable network class beyond pure parallel networks while remaining compatible with feature-based governance constraints.
\end{remark}

\subsubsection{A convex split-flow formulation and first-order convergence on SP networks}
Theorem~\ref{thm:sp-partition} provides an efficient oracle for evaluating logit probabilities for \emph{fixed} generalized edge costs.
To go beyond mere evaluation and obtain a deterministic algorithm with convergence guarantees, we can exploit the SP-tree structure to express the SW-SUE convex potential (Proposition~\ref{prop:swsue-potential}) in a low-dimensional set of \emph{split variables}.

\paragraph{Split variables.}
Fix a single commodity with demand $d$ on a two-terminal SP network with decomposition tree $\mathcal{T}$.
Let $\mathcal{V}_\oplus$ be the set of parallel nodes of $\mathcal{T}$.
For each $H\in\mathcal{V}_\oplus$ with children $(H_L,H_R)$, introduce a split-flow variable $y_H\in(0,d_H)$ representing the flow sent through $H_L$, where $d_H$ is the total flow entering subgraph $H$ (determined recursively by the splits above $H$).
Given $y=(y_H)_{H\in\mathcal{V}_\oplus}$, the induced edge-flow vector $x(y)$ is obtained by a single top-down pass on $\mathcal{T}$ and can be computed in $O(|E|)$ time.

\paragraph{Entropy decomposes on SP trees.}
On a series-parallel decomposition, path choice can be viewed as a hierarchical sequence of binary decisions at parallel nodes.
Accordingly, the path-flow entropy term in the SW-SUE potential decomposes into a sum of \emph{local} entropies at parallel nodes (up to an additive constant depending only on total demand).
This yields the following low-dimensional convex program.

\begin{theorem}[Split-flow convex program for SW-SUE on SP networks]
\label{thm:sp-split-convex}
Consider a single commodity of demand $d$ on a two-terminal SP network with continuous nondecreasing edge latencies.
Assume edge-additive salience (feature tying) with parameters $u=(u_e)_{e\in E}$ as in \eqref{eq:edge-additive-salience}.
Define the split-flow objective
\begin{equation} 
\label{eq:sp-split-objective}
\Psi_u(y)
\;\triangleq\;
\sum_{e\in E}\int_{0}^{x_e(y)} \ell_e(z)\,dz
\;-\;
\frac{1}{\beta}\sum_{e\in E} u_e\,x_e(y)
\;+\;
\frac{1}{\beta}\sum_{H\in\mathcal{V}_\oplus}
\Big(
y_H\log y_H + (d_H-y_H)\log(d_H-y_H) - d_H\log d_H
\Big),
\end{equation}
where $d_H$ is the flow entering subgraph $H$ under splits $y$.
Then $\Psi_u$ is strictly convex on its feasible region, and its unique interior minimizer $y^\star$ induces the unique SW-SUE edge flow $x^\star=x(y^\star)$ under $u$.
\end{theorem}

\begin{proof}
We start from the SW-SUE convex potential in \emph{path} flows under edge-additive salience.
For a single commodity of demand $d$, Proposition~\ref{prop:swsue-potential} implies that the SW-SUE path flow is the unique minimizer of
\begin{equation}
\label{eq:sp-swsue-path-potential}
\Phi_u(f)
=
\sum_{e\in E}\int_{0}^{x_e(f)} \ell_e(z)\,dz
-
\frac{1}{\beta}\sum_{e\in E}u_e\,x_e(f)
+
\frac{1}{\beta}\sum_{p\in\cP} f_p\log f_p,
\end{equation}
over feasible path flows $f$ (up to an additive constant depending only on $d$).
The key observation is that on a two-terminal SP network, the objective \eqref{eq:sp-swsue-path-potential} admits an exact \emph{dynamic-programming elimination} on the SP tree, reducing it to split variables.

\paragraph{A value-function recursion on the SP tree.}
For each SP subgraph node $H\in\mathcal{T}$ and scalar $d_H\ge 0$, define the \emph{reduced} value function
\begin{equation} 
\label{eq:sp-value-function}
\begin{aligned}
F_H(d_H)\triangleq\min_{f^H\in\cF(H;d_H)}\Bigg\{&
\sum_{e\in E(H)}\int_{0}^{x_e(f^H)}\ell_e(z)\,dz
-\frac1\beta\sum_{e\in E(H)}u_e x_e(f^H)\\
&+\frac1\beta\Big(\sum_{p\in\cP(H)}f_p^H\log f_p^H-d_H\log d_H\Big)\Bigg\}.
\end{aligned}
\end{equation}
where $\cF(H;d_H)$ is the set of feasible $s$--$t$ path flows \emph{within} subgraph $H$ with total flow $\sum_{p\in\cP(H)} f_p^H=d_H$.
(The subtraction of $d_H\log d_H$ removes the normalization constant at each subproblem and will expose local entropy terms at parallel splits.)
At the root node, $d_{\mathrm{root}}=d$ is fixed, and minimizing \eqref{eq:sp-swsue-path-potential} is equivalent (up to the constant $d\log d/\beta$) to minimizing $F_{\mathrm{root}}(d)$.

We now show by structural induction on the SP tree that $F_H$ obeys the following recursion:
\begin{enumerate}
    \item \textbf{Leaf edge.} If $H$ is a single edge $e$, then there is only one path in $\cP(H)$ and the entropy term cancels, yielding
    \[
    F_e(d)
    =
    \int_{0}^{d}\ell_e(z)\,dz
    -
    \frac{1}{\beta}u_e\,d.
    \]
    \item \textbf{Series composition.} If $H=H_1\otimes H_2$ is a series composition, then
    \[
    F_H(d)=F_{H_1}(d)+F_{H_2}(d).
    \]
    \item \textbf{Parallel composition.} If $H=H_L\oplus H_R$ is a parallel composition, then
    \[
    F_H(d)
    =
    \min_{0\le y\le d}
    \Big\{
    F_{H_L}(y)+F_{H_R}(d-y)
    +\frac{1}{\beta}\big(y\log y+(d-y)\log(d-y)-d\log d\big)
    \Big\}.
    \]
\end{enumerate}

The leaf and parallel cases follow directly from the definition \eqref{eq:sp-value-function} because (i) a leaf has a single path, and (ii) in a parallel composition the path sets are disjoint so both the congestion term and the entropy term split additively across children.
The series case requires an ``independence is optimal'' argument for the entropy term.
Let $H=H_1\otimes H_2$ and consider any feasible joint path flow $f$ on $\cP(H)=\cP(H_1)\times \cP(H_2)$ with total mass $d$.
Let $f^{(1)}$ and $f^{(2)}$ be the induced marginals on $\cP(H_1)$ and $\cP(H_2)$ (each has total mass $d$).
Define the product coupling $g_{p_1,p_2}\triangleq f^{(1)}_{p_1}f^{(2)}_{p_2}/d$.
Nonnegativity of KL divergence gives
\[
0\le D(f\|g)=\sum_{p_1,p_2} f_{p_1,p_2}\log\frac{f_{p_1,p_2}}{g_{p_1,p_2}}
\quad\Longrightarrow\quad
\sum_{p_1,p_2} f_{p_1,p_2}\log f_{p_1,p_2}
\ge
\sum_{p_1,p_2} f_{p_1,p_2}\log g_{p_1,p_2}.
\]
Expanding $\log g_{p_1,p_2}=\log f^{(1)}_{p_1}+\log f^{(2)}_{p_2}-\log d$ shows
\[
\sum_{p_1,p_2} f_{p_1,p_2}\log g_{p_1,p_2}
=
\sum_{p_1} f^{(1)}_{p_1}\log f^{(1)}_{p_1}
+
\sum_{p_2} f^{(2)}_{p_2}\log f^{(2)}_{p_2}
-
d\log d.
\]
Thus the entropy part $\sum f\log f-d\log d$ is \emph{minimized} (equivalently, entropy is maximized) by the independent coupling $f=g$, and the reduced objective \eqref{eq:sp-value-function} decomposes as $F_{H_1}(d)+F_{H_2}(d)$.
This establishes the recursion.

\paragraph{From the recursion to the split-flow objective.}
Unrolling the recursion on the full SP tree introduces exactly one scalar split decision $y_H$ for each parallel node $H\in\mathcal{V}_\oplus$, and the resulting reduced objective equals \eqref{eq:sp-split-objective} with $x(y)$ computed by the induced split flows.
Hence minimizing the original path-flow potential \eqref{eq:sp-swsue-path-potential} is equivalent to minimizing $\Psi_u(y)$.

\paragraph{Convexity and uniqueness.}
The induced edge and incoming subgraph flows are affine functions of the absolute split flows. For a local entropy term $\phi(y,z)=y\log y+(z-y)\log(z-y)-z\log z$, its second variation at an interior point is
\[
 D^2\phi[(a,b),(a,b)]
 =\frac{(za-yb)^2}{zy(z-y)}\ge0.
\]
The root demand is fixed. If the sum of these quadratic forms vanishes, the first parallel split has zero variation; inductively, every incoming-flow variation and every parallel split variation is zero. Thus the entropy Hessian is positive definite on nonzero feasible split directions in the interior. The Beckmann terms are convex. The path minimizer is interior by Proposition~\ref{prop:swsue-potential}; dynamic-programming elimination therefore yields the unique interior minimizer of $\Psi_u$ and the corresponding unique SW-SUE edge flow.

\end{proof}

\begin{proposition}[First-order convergence guarantee (projected gradient)]
\label{prop:sp-gradient}
Assume each $\ell_e$ is Lipschitz on $[0,d]$ and consider a compact interior domain in which all split variables satisfy $\epsilon\le y_H\le d_H-\epsilon$ for some $\epsilon>0$.
On this domain, $\Psi_u$ in \eqref{eq:sp-split-objective} has an $L$-Lipschitz gradient and is $\mu$-strongly convex for some $L,\mu>0$.
On a nonempty closed convex such domain $K_\epsilon$, projected gradient descent with step size $1/L$ converges linearly to the minimizer $y^\star_\epsilon$ over $K_\epsilon$ (write $y^\star=y^\star_\epsilon$ in the following bound):
\[
\Psi_u(y^{(t)})-\Psi_u(y^\star)\ \le\ (1-\mu/L)^t\big(\Psi_u(y^{(0)})-\Psi_u(y^\star)\big).
\]
This is the unrestricted SW-SUE only if its minimizer belongs to $K_\epsilon$. Each gradient evaluation costs $O(|E|)$; projection and iteration counts must be accounted for separately.
\end{proposition}

\begin{proof}[Proof idea.]
On a compact interior domain, the split-flow objective has bounded Hessian eigenvalues because each parallel node contributes a strictly convex entropy term and each series segment contributes a smooth convex term. This yields $L$-smoothness and $\mu$-strong convexity. Standard projected gradient descent guarantees then give linear convergence to the unique minimizer.
Full proof is deferred to Appendix~\ref{app:proof-prop-sp-gradient}.
\end{proof}

\begin{remark}[Beyond contraction]
Theorems~\ref{thm:sp-split-convex} and Proposition~\ref{prop:sp-gradient} provide a convergence guarantee for equilibrium computation on a nontrivial network class (two-terminal SP) without relying on contraction of the equilibrium map.
The key is strict convexity of the potential plus a low-dimensional split-flow representation.
\end{remark}

\subsection{Differentiating through equilibrium}
When the equilibrium map is continuously differentiable and $I-\nabla_xT$ is nonsingular as required by Theorem~\ref{thm:implicit}, the mapping $\theta\mapsto x^\star(\theta)$ can be differentiated via implicit differentiation (Theorem~\ref{thm:implicit}).
For any differentiable scalar objective $J(\theta)=\mathrm{SC}(x^\star(\theta))+\lambda\,\mathrm{Cost}(\theta)$, the chain rule yields
\begin{equation}
\label{eq:gradJ}
\nabla_\theta J(\theta)
\;=\;
\nabla_x \mathrm{SC}\!\left(x^\star(\theta)\right)\,\frac{d x^\star}{d\theta}
\;+\;
\lambda\,\nabla_\theta \mathrm{Cost}(\theta),
\end{equation}
with $\tfrac{d x^\star}{d\theta}$ given by \eqref{eq:implicit-derivative}.
This is a standard ``differentiate through fixed points'' primitive, but it becomes nontrivial here because $T_\theta$ depends on $\theta$ both directly (surfacing) and indirectly (through the stationary distribution of a policy-induced memory Markov chain).

\subsection{Algorithmic template: equilibrium-aware policy gradient}
Algorithm~\ref{alg:eqgrad} sketches a practical approach when $T_\theta$ is a contraction and the stationary distributions can be computed (or estimated) efficiently.

\begin{algorithm}[t]
\caption{Equilibrium-aware memory/guidance optimization (template)}
\label{alg:eqgrad}
\begin{algorithmic}[1]
\STATE \textbf{Input:} initial $\theta^{(0)}$, step size schedule $\{\gamma_t\}$, tolerance $\varepsilon$
\FOR{$t=0,1,2,\dots$}
    \STATE \textbf{Equilibrium solve:} compute $x^\star(\theta^{(t)})$ (and $\mu^\star(\theta^{(t)})$) via Algorithm~\ref{alg:fwe}
    \STATE \textbf{Sensitivity:} compute/estimate $\nabla_x T_{\theta^{(t)}}(x^\star)$ and $\nabla_\theta T_{\theta^{(t)}}(x^\star)$
    \STATE Solve $(I-\nabla_x T_{\theta^{(t)}}(x^\star))^\top\,v = \nabla_x \mathrm{SC}(x^\star)$ for $v$ (linear system)
    \STATE Form gradient estimate $\widehat{\nabla_\theta J} \leftarrow v^\top \nabla_\theta T_{\theta^{(t)}}(x^\star) + \lambda \nabla_\theta \mathrm{Cost}(\theta^{(t)})$
    \STATE \textbf{Update:} $\theta^{(t+1)} \leftarrow \Pi_{\Theta}\big(\theta^{(t)}-\gamma_t\,\widehat{\nabla_\theta J}\big)$ (projected step)
    \IF{ $\|\theta^{(t+1)}-\theta^{(t)}\|\le \varepsilon$ }
        \STATE \textbf{break}
    \ENDIF
\ENDFOR
\STATE \textbf{Output:} optimized $\theta$ and resulting equilibrium $(x^\star,\mu^\star)$
\end{algorithmic}
\end{algorithm}

\begin{remark}[Implementation notes]
(i) In large route sets, $\nabla_\theta T_\theta$ can be estimated by sampling surfaced routes and using automatic differentiation through the softmax/logit components.
(ii) The linear solve in Algorithm~\ref{alg:eqgrad} can be done iteratively (e.g., GMRES for the generally nonsymmetric system). Contraction ensures invertibility in the matching norm, but the inverse bound grows as the contraction factor approaches one; conjugate gradients requires additional symmetry and positive definiteness.
(iii) If fairness constraints are required, one can use projected-gradient or primal--dual updates with type-level cost estimates extracted from the equilibrium.
\end{remark}

\subsection{From toy instances to general networks}
The Pigou results in Section~\ref{sec:rbp} demonstrate a central point: the optimal amount of recall is generally \emph{interior}, not ``as much as possible.''
In general networks, the design variables $\theta$ (surfacing and forgetting) allow the system to trade off:
\begin{itemize}
    \item \emph{efficiency:} reduce congestion externalities by dampening herding onto low-latency routes,
    \item \emph{stability:} avoid oscillations driven by overreactive guidance,
    \item \emph{governance:} satisfy privacy and cost constraints inherent to memory-bounded AI systems.
\end{itemize}
An open technical question is to characterize when such policies can guarantee improvement over full recall and how close they can get to the social optimum under natural constraints (non-discrimination, limited control authority, and bounded memory).

\section{From Information to Recall: Modeling Imperfect Recall}
This section is the main modeling fork relative to classic routing games.
We define a recall model that is compatible with the informational Braess framework while enabling \emph{endogenous} recall under memory policies.

\subsection{Recall sets as a baseline abstraction}
We first recall the informational-Braess style baseline in which each traveler type has access only to a fixed subset of routes.
This subsection is used only as a point of contact with ICWE/IBP; our main model in Section~\ref{sec:dynamic-recall-model} replaces fixed recall sets with an endogenous memory process.

We assume a finite set of traveler types $i\in\{1,\dots,I\}$ with demands $d_i$ summing to $d$.
Each type $i$ has a nonempty admissible path set $\cP_i\subseteq\cP$. An edge-information set $E_i\subseteq E$ induces all feasible paths using those edges, which is a special case; arbitrary path subsets need not admit that representation.
Type $i$ can only route on paths in $\cP_i$.

A \emph{recall policy parameter} $\theta$ may index a family $\{\cP_i(\theta)\}$ (e.g., varying memory budget, recall suppression, or information exposure).
In the \emph{static recall} view, we treat $\cP_i$ as fixed and study the induced equilibrium.
In the \emph{dynamic recall} view, recall sets are generated endogenously by a stochastic memory process together with an AI surfacing policy (Section~\ref{sec:dynamic-recall-model}); here $\theta$ corresponds to controllable memory/guidance parameters rather than an exogenous menu.

\subsection{Information-constrained Wardrop equilibrium as a baseline}
\begin{definition}[Recall/Information-constrained Wardrop equilibrium (ICWE/RCWE)]
Given recall sets $\cP_i$, a flow profile $f=(f^{(i)}_p)$ is a recall-constrained Wardrop equilibrium if for each type $i$ and every path $p\in \cP_i$ with $f^{(i)}_p>0$,
\[
L_p(f) \le L_{p'}(f)\quad \text{for all } p'\in \cP_i.
\]
\end{definition}

When $\cP_i=\cP$ for all $i$, this reduces to Wardrop equilibrium.
When the admissible paths are induced by fixed edge-information sets, this specializes to the ICWE of~\citep{acemoglu2018ibp}.

\begin{proposition}[Existence via potential minimization (standard)]
\label{prop:icwe-existence}
Assume each $\ell_e$ is continuous and nondecreasing.
For fixed recall sets $\{\cP_i\}_{i=1}^I$, an ICWE/RCWE exists.
Moreover, if there is a single origin--destination pair and each $\ell_e$ is strictly increasing on $[0,d]$, then the induced equilibrium edge loads are unique (though path flows need not be).
\end{proposition}

\begin{proof}[Proof idea.]
Define the Beckmann potential over feasible path flows restricted to the fixed recall sets and minimize it over a compact simplex. Convexity and continuity yield existence; KKT conditions recover the (information-constrained) Wardrop inequalities. For a single OD with strictly increasing latencies, strict convexity of the potential in edge loads yields uniqueness of equilibrium edge loads.
Full proof is deferred to Appendix~\ref{app:proof-prop-icwe-existence}.
\end{proof}

\begin{remark}
We state the proposition for completeness; it follows from standard convexity arguments (Beckmann-type potentials) and is established explicitly in the informational Braess framework~\citep{acemoglu2018ibp}.
\end{remark}

\subsection{Dynamic recall: a fully specified stochastic model}
\label{sec:dynamic-recall-model}
This subsection specifies the dynamic recall model.
The goal is to capture \emph{endogenous recall} under explicit memory budgets and eviction policies while retaining a clean non-atomic routing foundation.

\subsubsection{Primitives}
We consider a finite set of origin--destination (OD) commodities $\cK$.
Each commodity $k\in\cK$ has origin $s_k$, destination $t_k$, and a fixed demand (population mass) $d_k>0$.
Let $\cP_k$ denote a \emph{finite} set of feasible $s_k$--$t_k$ paths.\footnote{The finiteness assumption is standard when working with explicit path sets; in computational experiments one can take $\cP_k$ as a candidate route set generated by $k$-shortest paths, column generation, or sampling.}

The physical network is $G=(V,E)$ with continuous, nondecreasing edge latency functions $\ell_e:\R_{\ge 0}\to\R_{\ge 0}$.
Given an edge-load vector $x\in\R_{\ge 0}^{|E|}$, the latency of path $p$ is
\begin{equation}
\label{eq:path-latency}
L_p(x) \;=\; \sum_{e\in p}\ell_e(x_e).
\end{equation}

\subsubsection{Memory state space and recall}
Fix a memory budget $B_k\in\{1,2,\dots,|\cP_k|\}$ for each commodity $k$.

\begin{definition}[Memory state space]
For commodity $k$, a memory state is an \emph{ordered} list
\[
m=(p^{(1)},p^{(2)},\dots,p^{(B_k)}) \in \cM_k,
\]
where $p^{(j)}\in\cP_k$ are distinct paths.
The order encodes recency/priority, with $p^{(1)}$ interpreted as the most ``salient'' or most recently activated route.
The \emph{recalled set} induced by $m$ is
\[
\cS(m) \;=\; \{p^{(1)},\dots,p^{(B_k)}\}\subseteq \cP_k.
\]
\end{definition}

\paragraph{Exogenous discovery (route surfacing).}
To allow routes that are not currently recalled to become available, we assume a simple discovery/surfacing process.
For each commodity $k$, fix a distribution $\rho_k$ over $\cP_k$ with full support, i.e., $\rho_k(p)>0$ for all $p\in\cP_k$.
At each period (day) and for each traveler of commodity $k$, an independent ``surfaced'' route $q$ is drawn as
\[
q \sim \rho_k.
\]
The traveler can choose from the \emph{available} set
\begin{equation}
\label{eq:available-set}
\cA(m,q) \;=\; \cS(m)\cup\{q\}.
\end{equation}
This is a minimal, analytically convenient abstraction of (i) occasional exploration, or (ii) an interface (e.g., navigation app) that surfaces a candidate alternative outside the traveler's currently recalled set.

\subsubsection{Route choice given congestion}
Fix an inverse-temperature (rationality) parameter $\beta\ge 0$.
Given a memory state $m$, surfaced route $q$, and congestion $x$, the traveler's mixed strategy over available routes is the logit response
\begin{equation}
\label{eq:logit-choice}
\sigma(p\mid m,q,x)
\;=\;
\frac{\exp\!\big(-\beta\,L_p(x)\big)}{\sum_{r\in\cA(m,q)} \exp\!\big(-\beta\,L_r(x)\big)}
\qquad \text{for } p\in\cA(m,q),
\end{equation}
and $\sigma(p\mid m,q,x)=0$ otherwise.
The limit $\beta\to\infty$ recovers (tie-broken) best response on the available set.

\paragraph{Information structure.}
We treat $x$ (and hence the travel-time estimates $L_p(x)$) as a contemporaneous signal observed by the agent or provided by the guidance platform for the routes in the available set $\cA(m,q)$.
Thus, memory affects \emph{consideration} (which routes can be chosen) rather than beliefs about costs.
This aligns the model with ``random attention''/consideration-set foundations and isolates the externality created by policy-controlled recall.
Extensions in which memory also stores (possibly biased) cost estimates can be layered on top, but are not needed for the implementability and governance results proved in the stationary salience design layer.

\subsubsection{Eviction/forgetting policy: LRU as a canonical baseline}
A memory policy specifies how the state updates after a route is chosen.
We give a canonical policy that is (i) widely used in computer systems, (ii) cognitively interpretable, and (iii) yields a finite-state Markov chain: \emph{least-recently used} (LRU) eviction.

\begin{definition}[LRU update map]
\label{def:lru}
For commodity $k$, define the deterministic update map $U_k:\cM_k\times \cP_k \to \cM_k$ as follows.
Let $m=(p^{(1)},\dots,p^{(B_k)})$ and let $p\in\cP_k$ be the route \emph{chosen} in the current period.
\begin{enumerate}
    \item If $p\in \cS(m)$ and $p=p^{(j)}$ for some $j$, then $U_k(m,p)$ is obtained by moving $p$ to the front and shifting the earlier items back one position:
    \[
    U_k(m,p) \;=\; (p,p^{(1)},\dots,p^{(j-1)},p^{(j+1)},\dots,p^{(B_k)}).
    \]
    \item If $p\notin \cS(m)$, then $p$ is inserted at the front and the last element is dropped:
    \[
    U_k(m,p) \;=\; (p,p^{(1)},\dots,p^{(B_k-1)}).
    \]
\end{enumerate}
\end{definition}

\begin{remark}[Other policies]
The same framework accommodates FIFO, random replacement, score-based priority eviction, and summary/hybrid policies.
We use LRU to obtain a clean, fully discrete baseline; later sections can compare policies empirically and (where possible) analytically.
\end{remark}

\subsubsection{Population state, induced flows, and the within-period fixed point}
Travelers are non-atomic and persistent.
For each commodity $k$, let $\mu_k$ be a probability distribution over memory states $\cM_k$ (equivalently, a population share over memory states).
A \emph{population memory profile} is $\mu=(\mu_k)_{k\in\cK}$.

Given a memory profile $\mu$ and congestion $x$, the induced expected path flow on $p\in\cP_k$ is
\begin{equation}
\label{eq:path-flow-induced}
f_{k,p}(\mu,x)
\;=\;
d_k \sum_{m\in\cM_k}\mu_k(m)\sum_{q\in\cP_k}\rho_k(q)\,\sigma(p\mid m,q,x).
\end{equation}
Let $f(\mu,x)$ denote the concatenation over all commodities. The induced edge-load vector $x(f)\in\R_{\ge 0}^{|E|}$ is defined componentwise by
\begin{equation}
\label{eq:edge-loads}
x_e(f) \;=\; \sum_{k\in\cK}\;\sum_{p\in\cP_k} f_{k,p}\,\mathbf{1}\{e\in p\}
\qquad \text{for each } e\in E.
\end{equation}

\begin{definition}[Congestion consistency for a memory profile]
\label{def:congestion-consistency}
A congestion vector $x$ is \emph{consistent} with memory profile $\mu$ if
\begin{equation}
\label{eq:congestion-fixed-point}
x \;=\; x\big(f(\mu,x)\big).
\end{equation}
We write $x\in\Xi(\mu)$ for the (possibly set-valued) correspondence of solutions.
\end{definition}

\begin{remark}
Under logit choice (finite $\beta$) and standard regularity of latencies, \eqref{eq:congestion-fixed-point} is the equilibrium condition of a smooth stochastic user equilibrium on restricted choice sets; in many settings it admits a unique solution.
For the purposes of model definition, we take $\Xi(\mu)$ as the equilibrium correspondence that maps population memory to congestion.
\end{remark}

\subsubsection{Memory-state Markov kernel induced by congestion}
Fix a congestion vector $x$.
For each commodity $k$, the LRU update together with logit choice and route surfacing induces a Markov chain on $\cM_k$.
Its transition kernel is
\begin{equation}
\label{eq:memory-kernel}
P^{(k)}_{x}(m'\mid m)
\;=\;
\sum_{q\in\cP_k}\rho_k(q)\sum_{p\in\cA(m,q)} \sigma(p\mid m,q,x)\,\mathbf{1}\{m' = U_k(m,p)\}.
\end{equation}

Given a current population distribution $\mu_k$, the next-period distribution is the pushforward
\begin{equation}
\label{eq:mu-update}
\mu_k^{+}(m') \;=\; \sum_{m\in\cM_k}\mu_k(m)\,P^{(k)}_{x}(m'\mid m).
\end{equation}
We write $\mu^{+}=\Phi_x(\mu)$ for the full profile update across commodities.

\subsubsection{Forgetful Wardrop equilibrium (FWE)}
We can now define a stationary equilibrium that couples (i) congestion consistency and (ii) stationary memory dynamics.

\begin{definition}[Forgetful Wardrop equilibrium (FWE)]
\label{def:fwe}
A pair $(x^\star,\mu^\star)$ is a \emph{forgetful Wardrop equilibrium} if:
\begin{enumerate}
    \item \textbf{Congestion consistency:} $x^\star \in \Xi(\mu^\star)$, i.e., $x^\star$ satisfies \eqref{eq:congestion-fixed-point} for $\mu^\star$.
    \item \textbf{Memory stationarity:} for every commodity $k\in\cK$, $\mu_k^\star$ is stationary under the kernel induced by $x^\star$:
    \[
    \mu_k^\star \;=\; \sum_{m\in\cM_k}\mu_k^\star(m)\,P^{(k)}_{x^\star}(\cdot\mid m).
    \]
    Equivalently, $\mu^\star = \Phi_{x^\star}(\mu^\star)$.
\end{enumerate}
\end{definition}

\begin{remark}[Relation to classical models]
If $B_k=|\cP_k|$ and $\rho_k$ is irrelevant (or if $\cA(m,q)=\cP_k$ always), FWE reduces to a stochastic user equilibrium (and in the limit $\beta\to\infty$ to a Wardrop equilibrium).
\noindent With fixed exogenous menus, finite-temperature FWE becomes restricted-menu SUE. The deterministic ICWE/RCWE baseline is obtained in a suitable zero-noise limit, not at finite $\beta$.
\end{remark}

\section{Existence, Characterization, and Computation of Forgetful Wardrop Equilibrium}
\label{sec:existence}

This section records baseline theoretical properties of the model in Section~\ref{sec:dynamic-recall-model}.
The main purpose is twofold: (i) to establish that the coupled flow--memory equilibrium is well-posed, and (ii) to provide a concrete computational template that will later support empirical results and algorithmic policy design.

\subsection{Feasible set and standing assumptions}
Let $D \triangleq \sum_{k\in\cK} d_k$ denote total demand.
Because each traveler selects exactly one path per period and paths are simple, every edge load satisfies $0\le x_e \le D$.
We therefore work on the compact convex set
\begin{equation}
\label{eq:feasible-x}
\cX \;\triangleq\; [0,D]^{|E|}.
\end{equation}

\begin{assumption}[Regularity and exploration]
\label{ass:regularity}
For each edge $e\in E$, the latency function $\ell_e(\cdot)$ is continuous and nondecreasing on $[0,D]$.
For each commodity $k$, the surfacing distribution $\rho_k$ has full support on $\cP_k$.
Finally, the logit parameter satisfies $\beta<\infty$.
\end{assumption}

Assumption~\ref{ass:regularity} implies that (i) every available route has strictly positive choice probability under \eqref{eq:logit-choice}, and (ii) every route can be surfaced with positive probability.

\subsection{Ergodicity of the memory Markov chain}
\begin{lemma}[Ergodicity and uniqueness of the stationary memory distribution]
\label{lem:ergodic}
Fix a commodity $k$ and a congestion vector $x\in\cX$.
Under Assumption~\ref{ass:regularity}, the Markov chain on $\cM_k$ with transition kernel $P^{(k)}_x$ defined in \eqref{eq:memory-kernel} is irreducible and aperiodic.
Consequently, it admits a \emph{unique} stationary distribution, denoted $\pi_k(x)\in\Delta(\cM_k)$, with full support on $\cM_k$.
\end{lemma}

\begin{proof}[Proof idea.]
Irreducibility is shown by constructing a positive-probability sequence of surfaced routes that (under logit) can be chosen to ``write'' any target LRU list into memory via repeated insertions. Aperiodicity follows from the positive self-loop obtained when the traveler selects the most recent recalled route, leaving the ordered list unchanged.
Full proof is deferred to Appendix~\ref{app:proof-lem-ergodic}.
\end{proof}

\begin{lemma}[Continuity of the stationary memory map]
\label{lem:pi-continuous}
Under Assumption~\ref{ass:regularity}, for each commodity $k$ the mapping $x\mapsto \pi_k(x)$ is continuous on $\cX$.
\end{lemma}

\begin{proof}[Proof idea.]
The transition matrix entries depend continuously on $x$ through continuous latencies and the smooth logit map. The stationary distribution is the unique solution of a linear system whose coefficient matrix remains nonsingular under ergodicity; continuity then follows from continuity of matrix inversion on the set of nonsingular matrices.
Full proof is deferred to Appendix~\ref{app:proof-lem-pi-continuous}.
\end{proof}

Lemma~\ref{lem:ergodic} yields an important simplification: in a forgetful Wardrop equilibrium, the memory profile is pinned down by congestion via $\mu_k^\star=\pi_k(x^\star)$.
Thus the equilibrium can be characterized by a fixed point in edge-load space alone.

\subsection{Within-period congestion consistency}
\begin{proposition}[Existence of a congestion-consistent flow for fixed memory]
\label{prop:xi-nonempty}
Fix a memory profile $\mu$.
Under Assumption~\ref{ass:regularity}, the set $\Xi(\mu)$ of congestion vectors consistent with $\mu$ (Definition~\ref{def:congestion-consistency}) is nonempty.
\end{proposition}

\begin{proof}[Proof idea.]
For fixed $\mu$, the induced flow map $x\mapsto f(\mu,x)$ is continuous, hence so is the induced edge-load map $T_\mu(x)=x(f(\mu,x))$. Since $T_\mu$ maps the compact convex set $\cX$ to itself, Brouwer's fixed point theorem yields $x=T_\mu(x)$.
Full proof is deferred to Appendix~\ref{app:proof-prop-xi-nonempty}.
\end{proof}

\subsection{Existence of Forgetful Wardrop equilibrium}
Define $\pi(x)\triangleq(\pi_k(x))_{k\in\cK}$ and the \emph{reduced} fixed-point map $T:\cX\to\cX$ by
\begin{equation}
\label{eq:T-reduced}
T(x)\;\triangleq\; x\!\left(f(\pi(x),x)\right).
\end{equation}

\begin{theorem}[Existence of FWE]
\label{thm:fwe-existence}
Under Assumption~\ref{ass:regularity}, there exists at least one forgetful Wardrop equilibrium $(x^\star,\mu^\star)$.
Moreover, any $x^\star\in\cX$ satisfying $x^\star=T(x^\star)$ together with $\mu_k^\star=\pi_k(x^\star)$ for each $k$ constitutes a FWE.
\end{theorem}

\begin{proof}[Proof idea.]
Combine Lemma~\ref{lem:ergodic} and Lemma~\ref{lem:pi-continuous} to obtain a continuous stationary-memory map $x\mapsto \pi(x)$. Substituting $\mu=\pi(x)$ into the congestion consistency map yields a continuous self-map $T$ on $\cX$, so Brouwer yields a fixed point $x^\star=T(x^\star)$ and hence an FWE.
Full proof is deferred to Appendix~\ref{app:proof-thm-fwe-existence}.
\end{proof}

\subsection{A practical fixed-point computation template}
Theorem~\ref{thm:fwe-existence} suggests a direct numerical approach: iterate the reduced map $T(\cdot)$ while recomputing the stationary memory distributions.
In finite state spaces, $\pi_k(x)$ can be computed via power iteration on $P^{(k)}_x$ (or more stable linear-algebra routines for the eigenvector corresponding to eigenvalue $1$).

\begin{algorithm}[t]
\caption{Fixed-point iteration for FWE (basic template)}
\label{alg:fwe}
\begin{algorithmic}[1]
\STATE \textbf{Input:} initial $x^{(0)}\in\cX$, damping $\eta\in(0,1]$, tolerance $\varepsilon>0$
\FOR{$t=0,1,2,\dots$ until convergence}
    \FOR{each commodity $k\in\cK$}
        \STATE Build transition matrix $P^{(k)}_{x^{(t)}}$ from \eqref{eq:memory-kernel}
        \STATE Compute stationary distribution $\pi_k(x^{(t)})$ (e.g., power iteration)
    \ENDFOR
    \STATE Form $\pi(x^{(t)})=(\pi_k(x^{(t)}))_{k\in\cK}$
    \STATE Compute induced path flows $f(\pi(x^{(t)}),x^{(t)})$ via \eqref{eq:path-flow-induced}
    \STATE Compute updated edge loads $\tilde x^{(t+1)} \leftarrow x(f(\pi(x^{(t)}),x^{(t)}))$ via \eqref{eq:edge-loads}
    \STATE Damped update: $x^{(t+1)} \leftarrow (1-\eta)x^{(t)}+\eta\,\tilde x^{(t+1)}$
    \IF{$\|\tilde x^{(t+1)}-x^{(t)}\|_\infty \le \varepsilon$}
        \STATE \textbf{break}
    \ENDIF
\ENDFOR
\STATE \textbf{Output:} $x^{(t+1)}$ and $\mu^\star=\pi(x^{(t+1)})$
\end{algorithmic}
\end{algorithm}

Algorithm~\ref{alg:fwe} is a numerical template. Its use requires a residual check, an iteration limit, and a convergence argument or empirical diagnostic for the selected instance.

\begin{remark}[State-space explosion and scalable surrogates]
\label{rem:state-space}
Algorithm~\ref{alg:fwe} is primarily a conceptual fixed-point template.
The memory state space for commodity $k$ has cardinality
$
|\cM_k| = |\cP_k|\cdot(|\cP_k|-1)\cdots(|\cP_k|-B_k+1),
$
which is a falling factorial in $|\cP_k|$ of order $B_k$ and quickly becomes intractable even for moderate candidate route sets.
In large networks, one should therefore avoid explicit enumeration of $\cM_k$.
Two scalable alternatives are: (i) Monte Carlo estimation of stationary choice frequencies by simulating the memory chain directly (without storing the full transition matrix), and (ii) mean-field closures that track only per-route recall probabilities (e.g., the LRU$\rightarrow$TTL$\rightarrow$salience surrogate developed in Section~\ref{sec:micro2salience-B}), which yield a low-dimensional nonlinear closure. A frozen-salience inner problem is convex, but the endogenous TTL feedback need not inherit that convexity.
\end{remark}

\subsection{Uniqueness and global stability under a contraction condition}
Existence via Brouwer is not the end of the story: for algorithmic design and comparative statics we need \emph{uniqueness} and \emph{stability}.
This subsection provides a sufficient condition under which the reduced equilibrium map $T$ is a contraction on $\cX$, implying a unique FWE and global convergence of natural day-to-day dynamics.

\begin{assumption}[Lipschitz latencies]
\label{ass:lipschitz}
Each $\ell_e(\cdot)$ is continuously differentiable on $[0,D]$ with derivative bounded by $L_e$, and we define $L\triangleq \max_{e\in E} L_e$.
\end{assumption}

\begin{assumption}[Uniform mixing of memory dynamics (Doeblin condition)]
\label{ass:doeblin}
For each commodity $k$ there exist $\varepsilon_k\in(0,1]$ and a distribution $\nu_k\in\Delta(\cM_k)$ such that for all $x\in\cX$ and all $m\in\cM_k$,
\begin{equation}
\label{eq:doeblin}
P_x^{(k)}(\cdot\mid m)\;\ge\; \varepsilon_k\,\nu_k(\cdot)
\qquad\text{(componentwise).}
\end{equation}
\end{assumption}

\begin{remark}[How to enforce Assumption~\ref{ass:doeblin} by design]
A simple sufficient mechanism is an \emph{exogenous reset}: with probability $\varepsilon_k$ per period, the memory state is redrawn from $\nu_k$ independently of the chosen route.
For the explicitly modified kernel $(1-\varepsilon_k)P_x+\varepsilon_k\mathbf{1}\nu_k^\top$, the condition follows directly. Pure LRU with $B>1$ need not satisfy a one-step minorization; an appropriate multi-step argument is a different bound.
\end{remark}

Define the maximum path length
\[
H \triangleq \max_{k\in\cK}\max_{p\in\cP_k} |p|
\]
(where $|p|$ is the number of edges on path $p$), and recall that total demand is $D=\sum_k d_k$.

\begin{lemma}[Logit sensitivity]
\label{lem:logit-lip}
Fix any finite action set $\cA$.
The logit map $c\mapsto \sigma(\cdot\mid c)$ defined by $\sigma(a\mid c)\propto \exp(-\beta c_a)$ is $\beta$-Lipschitz from $(\R^{|\cA|},\|\cdot\|_\infty)$ to $(\Delta(\cA),\|\cdot\|_1)$:
\[
\big\|\sigma(\cdot\mid c)-\sigma(\cdot\mid c')\big\|_1 \;\le\; \beta\,\|c-c'\|_\infty.
\]
\end{lemma}

\begin{proof}[Proof idea.]
Differentiate the logit map: its Jacobian has entries bounded in magnitude by $\beta$ times a probability product. Bounding the induced operator norm from $\ell_\infty$ to $\ell_1$ yields a global Lipschitz constant $\beta$. The claim then follows from the mean value theorem.
Full proof is deferred to Appendix~\ref{app:proof-lem-logit-lip}.
\end{proof}

\begin{lemma}[Lipschitz continuity of $x\mapsto \pi_k(x)$ with explicit dependence on mixing]
\label{lem:pi-lip}
Under Assumptions~\ref{ass:regularity}--\ref{ass:doeblin} and \ref{ass:lipschitz}, for each commodity $k$ the stationary distribution $\pi_k(x)$ is unique and satisfies, for all $x,y\in\cX$,
\[
\|\pi_k(x)-\pi_k(y)\|_1 \;\le\; \frac{1}{\varepsilon_k}\,\sup_{m\in\cM_k}\big\|P_x^{(k)}(\cdot\mid m)-P_y^{(k)}(\cdot\mid m)\big\|_1.
\]
Moreover, using Lemma~\ref{lem:logit-lip} and the Lipschitz bound on path costs, one may take
\[
C_k \;=\; \frac{\beta H L}{\varepsilon_k},
\]
so that $\|\pi_k(x)-\pi_k(y)\|_1 \le C_k\,\|x-y\|_\infty$.
\end{lemma}

\begin{proof}[Proof idea.]
Uniform mixing (Doeblin) implies a contraction of total-variation distances under the kernel and yields standard perturbation bounds for stationary distributions. Bounding $\|P_x-P_y\|$ via logit sensitivity and Lipschitz path costs gives $\|\pi(x)-\pi(y)\|_1=O(\|x-y\|_\infty)$.
Full proof is deferred to Appendix~\ref{app:proof-lem-pi-lip}.
\end{proof}

\begin{proposition}[Lipschitz bound for the reduced equilibrium map]
\label{prop:T-lip}
Under Assumptions~\ref{ass:regularity}, \ref{ass:lipschitz}, and \ref{ass:doeblin}, the reduced map $T$ in \eqref{eq:T-reduced} is Lipschitz on $\cX$.
In particular, there exists $\kappa>0$ such that for all $x,y\in\cX$,
\[
\|T(x)-T(y)\|_\infty \;\le\; \kappa\,\|x-y\|_\infty,
\]
and one admissible (conservative) choice is
\begin{equation}
\label{eq:kappa}
\kappa \;=\; D\Big(\beta\,H\,L+\max_{k\in\cK} C_k\Big).
\end{equation}
\end{proposition}

\begin{proof}[Proof idea.]
Bound changes in path costs by Lipschitz latencies and the fact that edge loads are linear in path flows. Use Lemma~\ref{lem:pi-lip} to control how the stationary memory distribution changes with $x$, and Lemma~\ref{lem:logit-lip} to control how logit probabilities change with costs. Aggregating over commodities and paths yields the stated Lipschitz factor $\kappa$.
Full proof is deferred to Appendix~\ref{app:proof-prop-T-lip}.
\end{proof}

\begin{theorem}[Uniqueness and global convergence]
\label{thm:unique}
Suppose Assumptions~\ref{ass:regularity}, \ref{ass:lipschitz}, and \ref{ass:doeblin} hold and the Lipschitz factor $\kappa$ in Proposition~\ref{prop:T-lip} satisfies $\kappa<1$.
Then:
\begin{enumerate}
    \item The reduced map $T$ has a \emph{unique} fixed point $x^\star$, hence the FWE $(x^\star,\mu^\star)$ is unique with $\mu^\star=\pi(x^\star)$.
    \item The fixed-point iteration $x^{(t+1)}\leftarrow T(x^{(t)})$ converges to $x^\star$ from any initialization, at a linear rate bounded by $\kappa$.
    \item The coupled day-to-day dynamics
    \[
    x_{t+1}=x\!\left(f(\mu_t,x_t)\right),\qquad
    \mu_{t+1}=\Phi_{x_t}(\mu_t)
    \]
    are globally convergent under the stated uniform one-step mixing and Lipschitz bounds; the same holds when the $x$-update is damped by any $\eta\in(0,1]$. See the explicit comparison argument in Appendix~\ref{app:proof-thm-unique}.
\end{enumerate}
\end{theorem}

\begin{proof}[Proof idea.]
Under $\kappa<1$, Proposition~\ref{prop:T-lip} implies $T$ is a contraction on the complete metric space $(\cX,\|\cdot\|_\infty)$. Banach's fixed point theorem yields a unique fixed point $x^\star$ and global convergence of the fixed-point iteration; mapping back via $\mu^\star=\pi(x^\star)$ yields uniqueness of the FWE.
Full proof is deferred to Appendix~\ref{app:proof-thm-unique}.
\end{proof}

\begin{remark}[Interpretation of the contraction regime]
The contraction condition $\kappa<1$ is satisfied when congestion costs are not overly sensitive (small $L$), agents are not overly deterministic (moderate $\beta$), or the memory dynamics mix rapidly (large $\varepsilon_k$).
The stability conclusion is conditional on the displayed bounds. Resetting does not by itself establish a formal privacy guarantee. A reset-modified kernel also requires its own stationary-law analysis; the no-reset $B=1$ formula does not automatically persist.
\end{remark}

\subsection{Comparative statics and differentiability for policy optimization}
For equilibrium-aware design we require that the equilibrium depends smoothly on policy parameters.
Let $\theta\in\Theta$ parameterize a family of surfacing distributions and memory policies (e.g., reset rate, eviction rule, or summary frequency), inducing kernels $P^{(k)}_{x,\theta}$, stationary distributions $\pi_{k,\theta}(x)$, and a reduced map
\[
T_\theta(x)\triangleq x\!\left(f(\pi_\theta(x),x;\theta)\right),
\qquad \pi_\theta(x)=(\pi_{k,\theta}(x))_{k\in\cK}.
\]

\begin{theorem}[Implicit differentiation of an isolated FWE]
\label{thm:implicit}
Assume $\Theta$ is an open set and that $T_\theta$ is continuously differentiable in $(x,\theta)$.
Fix $\theta\in\Theta$ and suppose $x^\star$ is a fixed point of $T_\theta$ such that
\begin{equation}
\label{eq:regular-fixed-point}
\det\!\Big(I-\nabla_x T_\theta(x^\star)\Big)\neq 0.
\end{equation}
Then there exists a neighborhood $U$ of $\theta$ and a unique continuously differentiable map $\theta'\mapsto x^\star(\theta')$ on $U$ such that $x^\star(\theta)=x^\star$ and $x^\star(\theta')=T_{\theta'}(x^\star(\theta'))$ for all $\theta'\in U$.
Moreover,
\begin{equation}
\label{eq:implicit-derivative}
\frac{d x^\star}{d\theta}
\;=\;
\Big(I-\nabla_x T_\theta(x^\star)\Big)^{-1}\nabla_\theta T_\theta(x^\star).
\end{equation}
\end{theorem}

\begin{proof}[Proof idea.]
Apply the implicit function theorem to $F(x,\theta)=T_\theta(x)-x$. The Jacobian with respect to $x$ is $\nabla_xT_\theta(x)-I$, which is invertible at $x^\star$ by \eqref{eq:regular-fixed-point}. This yields a locally unique differentiable selection $x^\star(\theta)$ and the derivative formula.
Full proof is deferred to Appendix~\ref{app:proof-thm-implicit}.
\end{proof}

Equation~\eqref{eq:implicit-derivative} is the key technical enabler for gradient-based policy optimization in Section~\ref{sec:design}.

\section{Micro-to-salience surrogates for richer memories}\label{sec:bridge}
The stationary salience layer is exact for $B=1$ last-choice memory (Theorem~\ref{thm:micro2salience}),
but realistic recall involves larger memory budgets and more complex update rules.
This section develops a mean-field approximation pipeline that maps explicit LRU memory to an
endogenized TTL model and then to stationary salience parameters, yielding a practical surrogate for
stationary behavior without simulating the full memory Markov chain.

\subsection{Beyond B=1: a mean-field micro-to-salience approximation for larger memories}
\label{sec:micro2salience-B}
The exact reduction in Theorem~\ref{thm:micro2salience} relies on the special structure of $B_k=1$ (``last choice'').
For larger memory budgets $B_k>1$ and realistic eviction rules such as LRU, the stationary choice probabilities no longer admit a closed form.
Nevertheless, the salience abstraction remains useful: it can be proposed as a heuristic approximation that compresses rich memory dynamics into \emph{route availability frequencies}.

\begin{center}
\begin{minipage}{0.95\linewidth}
\small
\textbf{Approximation assumptions for the micro$\rightarrow$salience surrogate (for $B_k>1$).}
Our $B_k>1$ bridge relies on four explicit approximation ingredients:
(i) replacement by an external independent Poisson-request surrogate (Assumption~\ref{ass:poissonized-requests} and Lemma~\ref{lem:poisson-thinning}),
(ii) an LRU$\rightarrow$TTL (characteristic-time) approximation in the large-cache/large-catalog regime (Theorem~\ref{thm:lru-ttl}),
(iii) an independent-availability menu surrogate as an additional surrogate assumption with a fixed nonempty baseline (Assumption~\ref{assump:independent-menu}), and
(iv) random-denominator concentration for ``large'' available sets (Proposition~\ref{prop:awl-approx} and Corollary~\ref{cor:awl-rate}).
These assumptions are \emph{not} implied by the fully specified discrete-time micro Markov model; rather, their end-to-end accuracy for endogenous choices remains unproved. Section~\ref{sec:experiments} evaluates these distinct approximation stages using a reproducible experiment suite.
\end{minipage}
\end{center}

\paragraph{Setup: LRU recall sets.}
Fix commodity $k$ and congestion vector $x$.
Let the memory state be an ordered list of the $B_k$ most recently chosen routes (LRU stack), so the recalled set is $S(m)\subseteq \cP_k$ with $|S(m)|=B_k$.
Each period, a candidate route $q\sim\rho_k$ is surfaced, and the traveler chooses from the available set $A=S(m)\cup\{q\}$ via the logit rule \eqref{eq:logit-choice}, after which the chosen route moves to the top of the LRU stack.

Let $\pi_{k,x}(\cdot)$ denote the stationary marginal distribution of the \emph{chosen} route for the resulting Markov chain (when it exists and is unique).

\paragraph{Step 0: from nonatomic repeated choice to a Poisson request stream.}
The cache-theoretic results we invoke (LRU$\rightarrow$TTL approximations) are stated for \emph{request processes} in continuous time.
To connect them to repeated route-choice, we introduce an external independent-request surrogate. Independence of request labels is an additional modeling assumption, not a consequence of non-atomicity.

\begin{assumption}[External Poissonized request surrogate]
\label{ass:poissonized-requests}
For a representative cache at fixed congestion, let decision epochs be a Poisson process of rate $r_k>0$. Label epochs independently with fixed probabilities $q_{k,p}>0$. These labels update that cache independently of its contents. This is an external surrogate, not the endogenous route-choice process of Section~\ref{sec:dynamic-recall-model}.
\end{assumption}
\begin{lemma}[Independent marking of a Poisson process]
\label{lem:poisson-thinning}
Under Assumption~\ref{ass:poissonized-requests}, the route-label streams are independent Poisson processes with intensities $\lambda_{k,p}=r_k q_{k,p}$.
\end{lemma}
\begin{proof}
Independent marking splits a finite-rate Poisson process into independent streams at the stated rates. No superposition of unit-rate clocks over an uncountable population is used. In the heuristic closure only, $q$ is set equal to a candidate stationary choice marginal; this does not establish temporal independence of the original choices.
\end{proof}

\paragraph{Step 0a: Discrete-time departures and the accuracy of Poissonization.}
Many day-to-day route-choice models are discrete: each agent departs once per period and draws a route label i.i.d.\ with probabilities $\pi_p$.
In that setting, each route's request stream is Bernoulli on a grid and inter-request times are geometric.
The Poissonization device above replaces geometric inter-request times by exponential ones; the next lemma quantifies the induced error at the level of the TTL/working-set ``hit'' probability that drives our micro$\rightarrow$salience surrogate.

\begin{lemma}[Working-set vs.\ Poisson TTL hit probabilities]
\label{lem:poissonization-accuracy}
Fix a route with per-period request probability $\pi\in(0,1)$ in a discrete-time i.i.d.\ request stream, and let $W\in\mathbb{N}$ be a window length.
The discrete-time working-set hit probability is
$H^{\mathrm{WS}}(W)=1-(1-\pi)^W$.
The continuous-time Poisson TTL approximation with rate $\lambda=\pi$ and horizon $T=W$ is
$H^{\mathrm{TTL}}(W)=1-e^{-\pi W}$.
Then
\[
0 \le H^{\mathrm{WS}}(W)-H^{\mathrm{TTL}}(W)
= e^{-\pi W}-(1-\pi)^W
\le \frac{\pi^2 W}{1-\pi}\,e^{-\pi W}
\le \frac{\pi^2 W}{1-\pi}.
\]
\end{lemma}

\begin{proof}
The bound $(1-\pi)^W\le e^{-\pi W}$ follows from $\log(1-\pi)\le -\pi$.
For the other direction, $\log(1-\pi)\ge -\pi/(1-\pi)$ implies $(1-\pi)^W\ge e^{-\pi W/(1-\pi)}$.
Therefore
\[
e^{-\pi W}-(1-\pi)^W \le e^{-\pi W}-e^{-\pi W/(1-\pi)}
= e^{-\pi W}\Big(1-e^{-\pi^2 W/(1-\pi)}\Big)
\le \frac{\pi^2 W}{1-\pi}e^{-\pi W},
\]
where the last step uses $1-e^{-y}\le y$.
\end{proof}

\begin{corollary}[Diffused popularity justifies Poissonization]
\label{cor:diffuse-poisson}
Let $\pi_{\max}\triangleq\max_p \pi_p$ and assume a regime in which $\pi_{\max}\to 0$ (e.g., a large choice set with no dominant route).
If the relevant window scales so that $\sup_p \pi_p W = O(1)$, then
\[
\sup_p \big|H^{\mathrm{WS}}_p(W)-H^{\mathrm{TTL}}_p(W)\big| = O(\pi_{\max}).
\]
\end{corollary}

\begin{proof}[Proof idea.]
Apply Lemma~\ref{lem:poissonization-accuracy} to each route $p$ and use $\pi_p^2 W\le (\pi_p W)\pi_p$. Under the regime $\sup_p \pi_p W=O(1)$ and $\pi_{\max}\to 0$, the uniform difference between working-set and Poisson TTL hit probabilities is $O(\pi_{\max})$.
Full proof is deferred to Appendix~\ref{app:proof-cor-diffuse-poisson}.
\end{proof}

\begin{remark}[Normalization]
Under the Poissonized model, the TTL in-cache probability is $H^{\mathrm{TTL}}_{k,p}(t)=1-e^{-\lambda_{k,p}(x)t}$.
Equivalently, by rescaling time by $r_k$ one can write $H^{\mathrm{TTL}}_{k,p}(t/r_k)=1-e^{-q_{k,p}t}$, which matches the form in \eqref{eq:ttl-hit}.
\end{remark}

\paragraph{Step 1: approximating recall frequencies by TTL/Che approximation.}
Under i.i.d.\ ``request'' models, LRU caches admit accurate and in some regimes asymptotically exact approximations in terms of a characteristic time (often called the TTL or Che approximation).
In our setting, the ``request'' process is the stationary sequence of chosen routes, whose marginal is $\pi_{k,x}$.
Motivated by the TTL approximation for LRU caches \citep{fricker2012lru,jiang2018ttl,gast2017ttl}, we define the \emph{approximate recall probability} for each route $p$ by
\begin{equation}
\label{eq:ttl-hit}
h_{k,p}(x) \;\approx\; 1-\exp\!\big(-\pi_{k,x}(p)\,T_k(x)\big),
\end{equation}
where the \emph{characteristic time} $T_k(x)\ge 0$ is chosen to satisfy the cache-size constraint
\begin{equation}
\label{eq:ttl-size}
\sum_{p\in\cP_k} h_{k,p}(x) \;=\; B_k.
\end{equation}
Intuitively, \eqref{eq:ttl-hit} says a route is recalled if it was requested at least once in the ``recent'' time window of length $T_k(x)$.

\paragraph{Step 1b: LRU$\rightarrow$TTL accuracy and an explicit rate under Poissonized requests.}
Step~1 replaces the LRU recall list by a TTL cache with a single \emph{characteristic time} $T_k(x)$.
For an external independent-request cache satisfying the cited hypotheses, individual LRU hit probabilities admit a TTL approximation. This statement does not establish the endogenous-choice closure.
We record a specialization of the convergence and rate results of~\citep{jiang2018ttl}.

\begin{theorem}[Imported marginal-hit bound for an external Poisson cache]
\label{thm:lru-ttl}
Consider an external cache with $n$ objects, independent positive-rate Poisson requests, and $B(n)\to\infty$ with $B(n)<n$. Assume conditions (C1) and (P1) of \citet{jiang2018ttl}; these control the cache/catalog ratio and the popularity tail near the cache size. Their condition (R6) holds for Poisson requests. Let $T_n$ solve $\sum_p(1-e^{-\lambda_pT_n})=B(n)$. Proposition 5.2 and Example 5.3 of that work imply
\begin{equation}\label{eq:lru-ttl-rate}
\max_p|H_p^{\mathrm{LRU}}-(1-e^{-\lambda_pT_n})|
=O\!\left(\sqrt{\log B(n)/B(n)}\right).
\end{equation}
This bounds individual hit probabilities for the external cache. It does not bound the joint menu law, endogenous route choices, or the error of an equilibrium fixed point.
\end{theorem}
\begin{proof}
Apply the cited Poisson specialization under its stated cache-size and popularity assumptions. At Poisson arrivals, hit and arbitrary-time occupancy probabilities agree. For general stationary requests those quantities differ, so this statement deliberately makes no non-Poisson extension. See Appendix~\ref{app:proof-thm-lru-ttl}.
\end{proof}

\paragraph{Limits of the cache analogy.}
Independent request streams make external TTL occupancy indicators independent at an arbitrary fixed time. Theorem~\ref{thm:lru-ttl} controls marginal hits only; it does not imply independence of LRU indicators. Further, union with one categorical surfaced route creates dependence even if the underlying TTL occupancies are independent. Additional coupling or joint-law bounds are necessary for an end-to-end result.

\paragraph{Step 2: from recall frequencies to an effective salience.}
Define the stationary \emph{availability probability} of route $p$ as
\begin{equation}
\label{eq:availability}
\eta_{k,p}(x)
\;\triangleq\;
\Pr\{p\in A\}
\;=\;
h_{k,p}(x) + \big(1-h_{k,p}(x)\big)\rho_k(p),
\end{equation}
since $p$ is available either because it is recalled, or because it is surfaced when not recalled.

The unconditional stationary choice probability under random availability is
\begin{equation}
\label{eq:random-attention-logit}
\pi_{k,x}(p)
\;=\;
\E\left[\frac{\mathbf{1}\{p\in A\}\exp(-\beta L_p(x))}
{\sum_{r\in A}\exp(-\beta L_r(x))}\right].
\end{equation}
This is a ``random attention'' logit: logit applied to a random menu $A$.

\begin{proposition}[Availability-weighted logit approximation]
\label{prop:awl-approx}
Fix $k$ and $x$ and assume all routes have strictly positive availability $\eta_{k,p}(x)>0$.
Let $w_p\triangleq \exp(-\beta L_p(x))$ and let $Z\triangleq \sum_{r\in A} w_r$ denote the random denominator in \eqref{eq:random-attention-logit}. Assume $A$ is almost surely nonempty (equivalently, $Z>0$ almost surely).
If $Z$ concentrates around its mean $\mu\triangleq \E[Z]=\sum_r \eta_{k,r}(x)w_r$ in the sense that the coefficient of variation $\mathrm{cv}(Z)\triangleq \sqrt{\mathrm{Var}(Z)}/\mu$ is small and $\max_p w_p/\mu$ is small, then
\begin{equation}
\label{eq:awl}
\pi_{k,x}(p)
\;=\;
\frac{\eta_{k,p}(x)\exp(-\beta L_p(x))}
{\sum_{r\in\cP_k}\eta_{k,r}(x)\exp(-\beta L_r(x))}
\;+\;
O\!\left(\mathrm{cv}(Z)^2+\max_{r}\frac{w_r}{\mu}\right).
\end{equation}
In particular, when the available set contains many routes with non-negligible weight (e.g., large $B_k$), the error term in \eqref{eq:awl} is small.
\end{proposition}

\begin{proof}[Proof idea.]
Write $\pi(p)=\E[\mathbf{1}\{p\in A\}w_p/Z]$ with $Z=\sum_{r\in A}w_r$ and expand $1/Z$ around $1/\mu$ on the high-probability event $|Z-\mu|\le \mu/2$. Control the remainder using the coefficient of variation of $Z$ and bound the effect of including a single term $w_p$ via $w_p/\mu$. This yields the stated error bound.
Full proof is deferred to Appendix~\ref{app:proof-prop-awl-approx}.
\end{proof}

\begin{assumption}[Independent-availability menu model (with nonempty menu)]
\label{assump:independent-menu}
Fix a commodity $k$ and an aggregate load state $x$.
To avoid degenerate empty-menu events (which make the logit denominator undefined), assume there exists a designated \emph{baseline} route $p_0\in\cP_k$ that is always available, i.e., $p_0\in A$ almost surely (equivalently $\eta_{k,p_0}(x)=1$).
For each other route $p\in\cP_k\setminus\{p_0\}$, the availability indicators
$\{\1\{p\in A\}\}$ are independent Bernoulli random variables with
$\Pr(p\in A)=\eta_{k,p}(x)$.
Equivalently, the menu is generated as $A=\{p_0\}\cup S$ where $S$ includes each $p\neq p_0$ independently with probability $\eta_{k,p}(x)$.

This is a separate menu surrogate. A fixed always-available route differs from the random last-chosen route. Conditioning independent Bernoulli menus on being nonempty changes their marginals and dependence and is not exactly equivalent to the fixed-baseline construction. The preceding LRU model and categorical surfacing do not imply this assumption.

\end{assumption}

\begin{lemma}[A simple concentration proxy under independent availability]
\label{lem:Z-concentration-bern}
Under Assumption~\ref{assump:independent-menu}, let $w_p\triangleq \exp(-\beta L_p(x))$, $w_{\max}\triangleq \max_{p} w_p$, and $Z\triangleq \sum_{r\in A} w_r$.
Then with $\mu\triangleq \E[Z]=\sum_{p}\eta_{k,p}(x)w_p$,
\begin{equation}
\label{eq:bern-var}
\mathrm{Var}(Z)\;\le\; w_{\max}\,\mu,
\qquad\text{and hence}\qquad
\mathrm{cv}(Z)^2 \;\le\; \frac{w_{\max}}{\mu}.
\end{equation}
Moreover, $\max_p w_p/\mu = w_{\max}/\mu$.
\end{lemma}

\begin{proof}
Write $Z=\sum_{p} w_p\,\1\{p\in A\}$.
Independence gives
$\mathrm{Var}(Z)=\sum_{p} w_p^2\,\eta_{k,p}(x)\big(1-\eta_{k,p}(x)\big)\le \sum_{p} w_p^2\,\eta_{k,p}(x)$.
Since $w_p^2\le w_{\max}\,w_p$, we obtain
$\mathrm{Var}(Z)\le w_{\max}\sum_{p}\eta_{k,p}(x)w_p = w_{\max}\mu$.
The remaining statements are immediate from definitions.
\end{proof}

\begin{corollary}[An explicit $1/B$-type rate under bounded costs and large menus]
\label{cor:awl-rate}
Fix $k$ and assume the feasible load region $\cX$ is compact so that path costs are uniformly bounded:
$L_p(x)\in[L_{\min},L_{\max}]$ for all $x\in\cX$ and all $p\in\cP_k$.
Then $w_p\in[w_{\min},w_{\max}]$ with $w_{\min}=e^{-\beta L_{\max}}$ and $w_{\max}=e^{-\beta L_{\min}}$.

Assume the independent menu model of Assumption~\ref{assump:independent-menu} and $\mathbb{E}|A|\ge B$. Then
$\mu=\E[Z]\ge B\,w_{\min}$ and hence
\begin{equation}
\label{eq:dominance-B}
\frac{w_{\max}}{\mu}\;\le\;\frac{w_{\max}}{B\,w_{\min}}
\;=\;\frac{e^{\beta(L_{\max}-L_{\min})}}{B}.
\end{equation}
Combining Proposition~\ref{prop:awl-approx} with Lemma~\ref{lem:Z-concentration-bern} yields an explicit large-$B$ regime in which the availability-weighted logit surrogate \eqref{eq:awl} is accurate:
the error for each route separately is of order $e^{\beta(L_{\max}-L_{\min})}/B$ whenever $L_{\max}-L_{\min}$ is bounded.
\end{corollary}

\begin{proof}[Proof idea.]
Bounded costs imply $w_p\in[w_{\min},w_{\max}]$. Under independent availability and $\mathbb{E}|A|\ge B$, $\mu=\E[Z]\ge Bw_{\min}$, hence $w_{\max}/\mu\le e^{\beta(L_{\max}-L_{\min})}/B$. Lemma~\ref{lem:Z-concentration-bern} gives $\mathrm{cv}(Z)^2\le w_{\max}/\mu$, and substituting into Proposition~\ref{prop:awl-approx} yields an $O(1/B)$ error.
Full proof is deferred to Appendix~\ref{app:proof-cor-awl-rate}.
\end{proof}

\begin{remark}[Stability and approximation error]
\label{cor:lru-combined-rate}
The marginal-hit bound in Theorem~\ref{thm:lru-ttl} and the coordinatewise AWL bound do not establish a dimension-free $\ell_1$ bound for the endogenous TTL--salience fixed point. Marginal availability does not determine joint menus; summing coordinate bounds can introduce a catalog-size factor; and a fixed-point perturbation requires stability. In particular, the two bounds cannot be combined into an $O(B^{-1})+O(\sqrt{\log B/B})$ endogenous-choice error rate without additional assumptions.

To make the stability requirement explicit, suppose that the exact map $T$ is a contraction with factor $\kappa<1$ and $\sup_x\|\widehat T(x)-T(x)\|\le\delta$. An approximate fixed point satisfying $\|\widehat T(\widehat x)-\widehat x\|\le r$ then obeys
\[
\|\widehat x-x^\star\|\le\frac{\delta+r}{1-\kappa}.
\]
This follows by inserting $\widehat T(\widehat x)$ and $T(\widehat x)$, applying the triangle inequality and contraction, and rearranging. If social cost is $L_{\mathrm{SC}}$-Lipschitz in the same norm, its error is at most $L_{\mathrm{SC}}(\delta+r)/(1-\kappa)$. A small numerical residual alone supplies neither a uniform model-error bound $\delta$ nor a contraction factor; these remain to be established for endogenous LRU routing.
\end{remark}

\begin{remark}[Interpretation]
Equation \eqref{eq:awl} shows that, beyond the exact $B_k=1$ case, salience-weighted logit remains a principled approximation: the effective salience weight is the stationary availability probability $\eta_{k,p}(x)$ induced by the memory policy.
Combining \eqref{eq:ttl-hit}--\eqref{eq:ttl-size} with \eqref{eq:awl} yields an explicit fixed-point surrogate for LRU-type memory with $B_k>1$ that can be used for equilibrium computation and for policy design.
\end{remark}

\begin{algorithm}[t]
\caption{Scalable TTL--salience surrogate for $B_k>1$ (mean-field equilibrium solver)}
\label{alg:ttl-salience}
\begin{algorithmic}[1]
\STATE \textbf{Input:} initial $x^{(0)}\in\cX$, damping $\eta\in(0,1]$, tolerances $\varepsilon_x,\varepsilon_\pi$
\STATE Initialize $\pi_k^{(0)}$ on $\cP_k$ for each commodity $k$ (e.g., uniform)
\FOR{$t=0,1,2,\dots$ until convergence}
    \FOR{each commodity $k\in\cK$}
        \STATE If $B_k=|\cP_k|$, set $h_{k,p}^{(t)}=1$ directly; otherwise find $T_k^{(t)}\ge 0$ such that $\sum_{p\in\cP_k}\big(1-e^{-\pi_k^{(t)}(p)\,T_k^{(t)}}\big)=B_k$ (bisection)
        \STATE If $B_k<|\cP_k|$, set $h_{k,p}^{(t)} \leftarrow 1-e^{-\pi_k^{(t)}(p)\,T_k^{(t)}}$ for all $p\in\cP_k$ \hfill (TTL recall)
        \STATE $\eta_{k,p}^{(t)} \leftarrow h_{k,p}^{(t)} + (1-h_{k,p}^{(t)})\,\rho_k(p)$ \hfill (availability)
        \STATE $\pi_k^{(t+1)}(p) \propto \eta_{k,p}^{(t)} \exp\!\big(-\beta L_p(x^{(t)})\big)$ for all $p\in\cP_k$ \hfill (AWL update)
    \ENDFOR
    \STATE Form path flows $f^{(t+1)}$ from $\pi^{(t+1)}$ and compute $\tilde x^{(t+1)} \leftarrow x(f^{(t+1)})$
    \STATE Damped update: $x^{(t+1)} \leftarrow (1-\eta)x^{(t)}+\eta\,\tilde x^{(t+1)}$
    \IF{$\|\tilde x^{(t+1)}-x^{(t)}\|_\infty \le \varepsilon_x$ \textbf{and} $\max_{k}\|\pi_k^{(t+1)}-\pi_k^{(t)}\|_1 \le \varepsilon_\pi$}
        \STATE \textbf{break}
    \ENDIF
\ENDFOR
\STATE \textbf{Output:} $(x^{(t+1)},\pi^{(t+1)},h^{(t)},\eta^{(t)})$
\end{algorithmic}
\end{algorithm}

\noindent
Algorithm~\ref{alg:ttl-salience} is a heuristic coupled iteration; convergence is not supplied by the fixed-salience potential. Initialize positive route probabilities, impose an iteration limit, use log-domain evaluation, and check both undamped residuals. When $B_k=|\cP_k|$, no finite characteristic time exists, so the explicit full-cache branch is required. The method avoids enumerating the factorial-sized memory state space (Remark~\ref{rem:state-space}) by tracking only per-route frequencies.
It is the computational primitive used in our synthetic validations (Section~\ref{sec:experiments}) and should be validated separately before use in governed policy optimization for $B_k>1$.

\begin{theorem}[Uniqueness and potential-based computation]
\label{thm:swsue-unique-stable}
Under the conditions of Proposition~\ref{prop:swsue-potential}, SW-SUE exists and is unique.
Moreover, any algorithm that globally minimizes $\Phi_s$ over $\cF$ (with assumptions and step sizes that ensure objective convergence; this is not an unconditional guarantee for standard gradient methods at the entropy boundary) converges to the unique equilibrium.
\end{theorem}

\begin{proof}[Proof idea.]
Proposition~\ref{prop:swsue-potential} shows SW-SUE coincides with the unique minimizer of a strictly convex potential $\Phi_s$. Existence and uniqueness are therefore immediate. Any method that globally minimizes $\Phi_s$ converges to this unique minimizer by definition of global convergence for convex optimization.
Full proof is deferred to Appendix~\ref{app:proof-thm-swsue-unique-stable}.
\end{proof}

\begin{remark}[Connection to the explicit memory model]
The explicit memory model in Section~\ref{sec:dynamic-recall-model} induces \emph{random} consideration sets.
In regimes with fast mixing (Assumption~\ref{ass:doeblin}) and high-frequency AI surfacing, these random consideration effects can often be summarized as persistent route-specific salience biases (e.g., via stationary retrieval frequencies).
Theorems~\ref{prop:swsue-potential}--\ref{thm:swsue-unique-stable} therefore provide a tractable ``design layer'' for memory/guidance optimization even when the underlying microdynamics are more complex.
\end{remark}

\subsection{Limitations of the salience abstraction and the richer-memory bridge}
\label{sec:limitations}
The stationary salience model is deliberately a reduced form.
It is the right ``design layer'' if the platform can directly manipulate exposure/utility biases (ranking, recommendations, interface architecture), and it is \emph{exactly} micro-founded for $B=1$ (Theorem~\ref{thm:micro2salience} and Corollary~\ref{cor:micro-realizability}).
For richer memories ($B>1$), the LRU$\rightarrow$TTL$\rightarrow$availability$\rightarrow$salience bridge (Section~\ref{sec:micro2salience-B}) should be interpreted as a mean-field approximation and may be inaccurate outside the regimes that make its assumptions plausible.

\paragraph{When the bridge can fail.}
The technically vulnerable points are:
(i) \emph{non-i.i.d.\ requests / slow mixing}: the Poissonization device and caching approximations assume the chosen-route sequence is close to an i.i.d.\ label stream over the time scale relevant for recall; this can fail under strong day-to-day nonstationarity or path-dependent learning;
(ii) \emph{correlated availability}: LRU recall indicators are generally correlated across routes, and the independent-menu model (Assumption~\ref{assump:independent-menu}) is not justified by the marginal-hit theorem;
(iii) \emph{small menus or steep utilities}: the random-denominator bounds in Proposition~\ref{prop:awl-approx} can have large constants when $\beta(L_{\max}-L_{\min})$ is large, so ``$O(1/B)$'' rates can be numerically loose unless costs are well behaved; and
(iv) \emph{beliefs vs.\ consideration}: we model memory as shaping consideration sets but not beliefs about $L_p(x)$ (Section~\ref{sec:dynamic-recall-model}); settings in which memory stores cost estimates or induces systematic belief biases require an additional information layer.

\paragraph{Scope of the design results.}
Accordingly, the strongest implementability and governance theorems in Sections~\ref{sec:implementability} and~\ref{sec:design} are stated for the stationary salience layer.
They can be operationalized directly as interface/ranking interventions, and (in the micro model) they are guaranteed to be realizable without loss only in the $B=1$ regime.
Characterizing the \emph{micro-realizable} subset of salience vectors for richer eviction/retention policies, and proving micro-level implementability beyond $B=1$, are important directions for future work.

\section{Value of Recall for Routing}
\label{sec:vor}
We now propose a routing analogue of the Value of Recall (VoR) metric used in imperfect-recall games and memory-aware network dynamics~\citep{alqithami2026snam}.

Let $J(\delta)$ denote a welfare objective under equilibrium at recall factor $\delta$; for example:
\[
J(\delta) \;=\; C\big(f^{\star}(\delta)\big),
\]
the total latency at the (appropriate) equilibrium $f^\star(\delta)$.

\begin{definition}[Value of Recall for routing (VoR-R)]
Assume $J(1)>0$ and define
\[
\mathrm{VoR}_{\text{routing}}(\delta)
\;=\;
\frac{J(\delta)}{J(1)}.
\]
Values $\mathrm{VoR}_{\text{routing}}(\delta) < 1$ indicate that imperfect recall (at $\delta$) \emph{improves} welfare relative to perfect recall.
\end{definition}

\begin{remark}
If one prefers a ``higher is better'' utility convention, for positive costs use the reciprocal $J(1)/J(\delta)$ and state that it reverses the normalization. Algebraically this also equals $U(1)/U(\delta)$ for $U=-J$; it is not the same numerical ratio as above. This is related to the ratio form in~\citep{alqithami2026snam}.
\end{remark}

\section{Recall Braess Paradox: Definition and Sharp Instances}
\label{sec:rbp}
This section isolates the core paradoxical effect we care about: \emph{equilibrium delay can increase when recall improves}, even though the physical network is unchanged.
Unlike classic Braess' paradox, the mechanism here is informational and policy-driven.

\subsection{Welfare and the Recall Braess Paradox (RBP)}
For any edge-load vector $x\in\cX$, define the standard total travel time (social cost)
\begin{equation}
\label{eq:social-cost}
\mathrm{SC}(x)\;\triangleq\;\sum_{e\in E} x_e\,\ell_e(x_e).
\end{equation}

Let $\theta$ denote a recall-policy parameter (e.g., memory budget, reset rate, summary frequency, or a parameter of the surfacing distribution).
Assume for each $\theta$ the FWE is unique, yielding $x^\star(\theta)$.

For each $\theta$, let $A_{k,\theta}$ denote the (random) available set presented to a representative infinitesimal traveler of commodity $k$ at stationarity (the realized menu $\cA(m,q)$ in the micro model, or the deterministic set $\cP_k(\theta)$ in the static recall baseline).
Define the \emph{menu-inclusion probabilities}
\[
h_{k,\theta}(p)\;\triangleq\;\Pr\{p\in A_{k,\theta}\},\qquad p\in\cP_k.
\]
We use the following canonical recall-richness order.

\begin{definition}[Recall-richness order (menu inclusion)]
\label{def:recall-richness}
We say $\theta_2$ is \emph{(weakly) more recall-rich} than $\theta_1$, written $\theta_2\succeq \theta_1$, if for every commodity $k$ and every path $p\in\cP_k$,
\[
h_{k,\theta_2}(p)\;\ge\; h_{k,\theta_1}(p).
\]
We write $\theta_2\succ \theta_1$ if $\theta_2\succeq \theta_1$ and strict inequality holds for at least one pair $(k,p)$.
A stronger sufficient condition is \emph{almost sure menu inclusion}: $A_{k,\theta_1}\subseteq A_{k,\theta_2}$ almost surely for all $k$, which implies $\theta_2\succeq \theta_1$.
\end{definition}

\begin{definition}[Recall Braess Paradox (social form)]
\label{def:rbp}
A \emph{Recall Braess Paradox} occurs if there exist $\theta_1,\theta_2$ such that $\theta_2\succ \theta_1$ in the recall-richness order of Definition~\ref{def:recall-richness}, yet
\[
\mathrm{SC}\!\left(x^\star(\theta_2)\right)\;>\;\mathrm{SC}\!\left(x^\star(\theta_1)\right).
\]
\end{definition}

The next results show that RBP is not a pathological artifact of exotic topologies: it can arise on the simplest possible series-parallel network.

\subsection{A sharp analytic instance: Pigouvian forgetting on the Pigou network}\label{sec:rbp-pigou}
Consider the two-link Pigou network: a single OD pair with total demand $1$ and two parallel edges.
Edge $a$ has constant latency $\ell_a(x)=1$ and edge $b$ has latency $\ell_b(x)=x$.

Under full recall (standard Wardrop), all flow routes to edge $b$ and the equilibrium social cost is $\mathrm{SC}=1$.
The system optimum splits flow evenly ($x_a=x_b=\tfrac12$) and achieves $\mathrm{SC}=\tfrac34$.

We now introduce a \emph{uniform recall-suppression} policy parameterized by $\alpha\in[0,1]$:
in each period, an $\alpha$ fraction of agents do not recall (or are not shown) edge $b$ and therefore must choose edge $a$; the remaining $1-\alpha$ fraction have full recall and choose selfishly.
This is a static restricted-menu policy. It is not generally realizable by the unchanged full-support LRU-plus-one-surfaced-route kernel: that kernel always allows every route to be surfaced. The exact formulas below are for deterministic ICWE, not finite-temperature FWE.

\begin{theorem}[Pigouvian forgetting implements the system optimum]
\label{thm:pigou}
In the Pigou network with $\ell_a(x)=1$ and $\ell_b(x)=x$, the induced (information/recall constrained) equilibrium under recall-suppression level $\alpha$ routes $x_a=\alpha$ and $x_b=1-\alpha$, and the resulting social cost is
\[
\mathrm{SC}(\alpha)\;=\;\alpha + (1-\alpha)^2 \;=\; 1-\alpha+\alpha^2.
\]
The minimizer is $\alpha^\star=\tfrac12$, yielding $\mathrm{SC}(\alpha^\star)=\tfrac34$ (the system optimum).
Moreover, $\alpha=0\succ \alpha=\tfrac12$ in the sense of Definition~\ref{def:recall-richness}, yet improving recall from $\alpha=\tfrac12$ to $\alpha=0$ \emph{increases} equilibrium social cost from $\tfrac34$ to $1$; hence a Recall Braess Paradox occurs even on a series-parallel network.
\end{theorem}

\begin{proof}
Given $\alpha$, an $\alpha$ mass must use edge $a$.
All remaining mass has access to both edges.
Since $\ell_b(1-\alpha)=1-\alpha\le 1=\ell_a(\alpha)$, all unconstrained agents strictly prefer $b$ (or are indifferent at $\alpha=0$).
Thus $(x_a,x_b)=(\alpha,1-\alpha)$.
Substituting into \eqref{eq:social-cost} yields $\mathrm{SC}(\alpha)=\alpha\cdot 1 + (1-\alpha)\cdot(1-\alpha)$.
The minimizer follows by differentiation.
\end{proof}

\begin{remark}[Why this is a ``Braess'' phenomenon]
Classic Braess' paradox requires a non-series-parallel structure for two-terminal networks, whereas the Pigou network is series-parallel.
The inefficiency here is not caused by adding physical capacity; it is caused by \emph{improving recall}, which shifts behavior toward the selfish equilibrium and away from the system optimum.
\end{remark}

\subsection{Generalization: one constant link and one increasing link}
The Pigou example is not an isolated curiosity.
It illustrates a general mechanism: when private and social marginal costs differ, calibrated information friction can mimic a corrective Pigouvian toll.

\begin{proposition}[Optimal recall suppression on a generalized Pigou network]
\label{prop:general-pigou}
Consider a two-link parallel network with demand $1$, where $\ell_a(x)=c$ is constant and $\ell_b(x)=g(x)$ is continuously differentiable, strictly increasing, and satisfies $g(0)=0$.
Suppose further that $g(1)\le c$ so that the full-recall Wardrop equilibrium routes all flow on $b$.
Let $y^\mathrm{opt}\in(0,1)$ denote the system-optimal flow on edge $b$, i.e., a minimizer of
\[
\min_{y\in[0,1]} \; c(1-y) + y\,g(y).
\]
Then the recall-suppression policy with level $\alpha^\star = 1-y^\mathrm{opt}$ implements the system optimum as an equilibrium and weakly improves over full recall. Strict improvement requires $\mathrm{SC}(y^\mathrm{opt})<\mathrm{SC}(1)$; convexity of $g$ together with the stated strict increase is a sufficient additional condition.
\end{proposition}

\begin{proof}[Proof idea.]
At $\alpha^\star=1-y^\mathrm{opt}$, strict monotonicity gives $g(y^\mathrm{opt})<g(1)\le c$, so all unsuppressed mass selects the increasing link. This implements the chosen global minimizer. Its cost is no greater than the full-recall cost; strictness needs an actual strict objective gap. The objective is not generally a quadratic in $\alpha$ and its minimizer need not be $1/2$.
Full proof is deferred to Appendix~\ref{app:proof-prop-general-pigou}.

\end{proof}

Proposition~\ref{prop:general-pigou} shows that calibrated recall friction can act as a Pigouvian correction in the simplest canonical networks.
We next (i) identify a network class theorem showing that social-cost RBP is essentially ubiquitous whenever users have \emph{any} nontrivial route choice, and (ii) formalize \emph{implementability} as an inverse-equilibrium problem for memory/guidance policies.

\subsection{Network classes: ubiquity and immunity for social RBP}
\label{sec:rbp-network-class}
We now show that, under the social-cost notion in Definition~\ref{def:rbp}, paradoxical non-monotonicity in recall is not confined to exotic graph topologies.

\begin{definition}[Choiceful vs.\ series two-terminal networks]
A two-terminal directed network $(G,s,t)$ is \emph{series} if there is a unique simple $s$--$t$ path.
Otherwise, it is \emph{choiceful} (there exist at least two distinct simple $s$--$t$ paths).
\end{definition}

\begin{theorem}[Ubiquity of social Recall Braess Paradox]
\label{thm:rbp-ubiquity}
Let $(G,s,t)$ be a choiceful two-terminal network.
In the deterministic static-menu ICWE model, there exist continuous nondecreasing edge latencies $(\ell_e)_{e\in E}$, total demand $d=1$, and two recall policies $\theta_1,\theta_2$ such that $\theta_2\succ \theta_1$ in the menu-inclusion order of Definition~\ref{def:recall-richness} and yet $\mathrm{SC}\!\left(x^\star(\theta_2)\right)>\mathrm{SC}\!\left(x^\star(\theta_1)\right)$.
Equivalently, social-cost RBP can occur on \emph{every} choiceful two-terminal network.

Conversely, if $(G,s,t)$ is series (unique $s$--$t$ path), then for any latencies and any recall policy with nonempty feasible menus, the induced flow is unique and invariant; social-cost RBP cannot occur.
\end{theorem}

\begin{proof}[Proof idea.]
Choose two distinct simple paths, give a distinct exclusive edge on each path the Pigou latency, and restrict all static menus to these two paths. This embeds the required static-menu instance even when the union supports other paths. This yields an instance with higher cost under increased recall; conversely, series networks admit no choices so recall policies cannot affect flows.
Full proof is deferred to Appendix~\ref{app:proof-thm-rbp-ubiquity}.
\end{proof}

\begin{remark}[Interpretation]
Theorem~\ref{thm:rbp-ubiquity} separates our social-cost notion from topology-based immunity results for classic Braess' paradox (which does not occur on series-parallel networks) and from group-harm informational paradox notions.
Under social cost, \emph{any} nontrivial route choice admits a Pigou-type externality that can be mitigated by calibrated information friction.
\end{remark}

\section{Numerical experiments}\label{sec:experiments}
\subsection{Experimental setup}\label{sec:exp_goals}\label{sec:exp_setup}
We evaluate six questions: how stationary predictions compare with lagged population dynamics; which approximation stages introduce error; how memory capacity changes stationary and dynamic costs; how an influence budget limits optimal design; whether inverse salience implements sampled targets under tying constraints; and how a specialized series-parallel solver compares with explicit path enumeration. All plotted quantities come from the specified computational models and measured solver runtimes; the network experiments do not use observed traveler choices.

\paragraph{Simulation and estimands.}
Each agent starts with an independent uniform ordered memory of $B$ distinct routes. At each period it receives one independent uniformly surfaced route, chooses from the distinct union of memory and that route using logit, and moves the chosen route to the front of memory. An unchosen surfaced route is not retained. All agents choose synchronously using the preceding period's loads; loads are then recomputed without damping. This is a lagged finite-population process, not a solver for the stationary FWE equations.
For retained periods $t=1,\ldots,T$, we report separately
\[
\overline{\mathrm{SC}}=\frac1T\sum_t\mathrm{SC}(x_t),\qquad
\mathrm{SC}(\bar x)=\mathrm{SC}\!\left(\frac1T\sum_t x_t\right),\qquad
\mathrm{SC}(x^\star).
\]
The first is mean daily cost, the second is cost evaluated at the mean flow, and the third is the cost of a computed stationary equilibrium. Nonlinear congestion makes these quantities different. Burn-in and small differences between half-window means do not certify stationary equilibrium or dynamic stability.

\paragraph{Uncertainty and equilibrium computation.}
Stochastic experiments use ten independent seeds. Unless identified as deterministic, entries of the form mean $\pm$ half-width use pointwise 95\% Student-$t$ intervals across these ten seed-level statistics, not across correlated periods or agents. These intervals describe replication variability; they do not cover model bias, incomplete mixing, or behavioral misspecification. Numerical equilibrium calls require an \emph{undamped} map residual $\|T(q)-q\|_\infty\le10^{-9}$ in within-commodity path shares. Damped iteration uses step $0.1$ and at most 20,000 iterations, with a residual-checked root fallback. The exact finite-memory solver uses the finite-state stationary law; $B=1$ and full recall use their analytic reductions.

\paragraph{Inputs and units.}
Synthetic networks have unit demand. Pigou has latencies $x$ and $1$; Braess has paths $(su,ut)$, $(sv,vt)$, $(su,uv,vt)$ and edge latencies $x,1,1,x,0$ in the order $su,ut,sv,vt,uv$. Sioux Falls uses the pinned TransportationNetworks data \citep{bstablerTransportationNetworks}, all 76 directed edges, the 50 largest-demand OD pairs, and eight free-flow shortest simple routes per pair. The selected demand totals 118,900 in the dataset's units, with no background traffic and no further demand scaling. All nodes may be traversed. BPR parameters are taken from the input files. Numerical TNTP time units are retained; $\beta=5$ is per TNTP time unit and is not behaviorally calibrated. Social cost therefore has units of dataset demand times TNTP time. The full route catalog and coefficients are provided with the code.

Pigou and Braess simulations use 8,000 agents, 3,000 periods, and 1,000 burn-in periods. Sioux Falls uses 800 agents per OD (40,000 total), 2,000 periods, and 500 burn-in periods; each agent represents its OD demand divided by 800. The constant-cost experiment uses 12,000 agents, 4,000 periods, and 1,000 burn-in periods. All use $\beta=5$ except the constant-cost experiment, which uses $\beta=4$.

\FloatBarrier
\subsection{Exp-1: stationary predictions and lagged congestion dynamics}\label{sec:exp_micro_vs_salience_congestion}
Pigou uses full memory $B=2$. Sioux Falls uses $B=2$ out of eight routes per OD. We compare the simulation time mean with (i) the endogenous TTL--availability--salience fixed point and (ii) an oracle-availability equilibrium obtained by freezing salience at availability marginals measured in the simulation. This oracle uses observed marginal availability; joint-menu dependence is still approximated. A third diagnostic, local AWL, applies availability-weighted logit once at the mean flow using the same measured marginals.

\begin{table}[!htbp]
\centering\fontsize{8}{9.2}\selectfont
\begin{tabular}{@{}lrrrr@{}}
\toprule
Instance & $\overline{\mathrm{SC}}$ & $\mathrm{SC}(\bar x)$ & TTL & Oracle \\
\midrule
Pigou & $0.820055 \pm 0.000025$ & $0.819937 \pm 0.000025$ & 0.819960 & 0.819960 \\
Sioux Falls ($10^6$) & $2.130624 \pm 0.000080$ & $0.952318 \pm 0.000008$ & 0.889140 & $0.8891404 \pm 0.0000003$ \\
\bottomrule
\end{tabular}

\caption{Exp-1 costs. Sioux Falls entries are divided by $10^6$; Pigou is unscaled. Uncertainty is the 95\% interval half-width across ten seeds. TTL is a deterministic stationary prediction. The oracle is a separately solved equilibrium for each seed's measured availability.}
\label{tab:exp1-costs}
\end{table}
\begin{figure}[!htbp]
\centering\includegraphics[width=.94\linewidth]{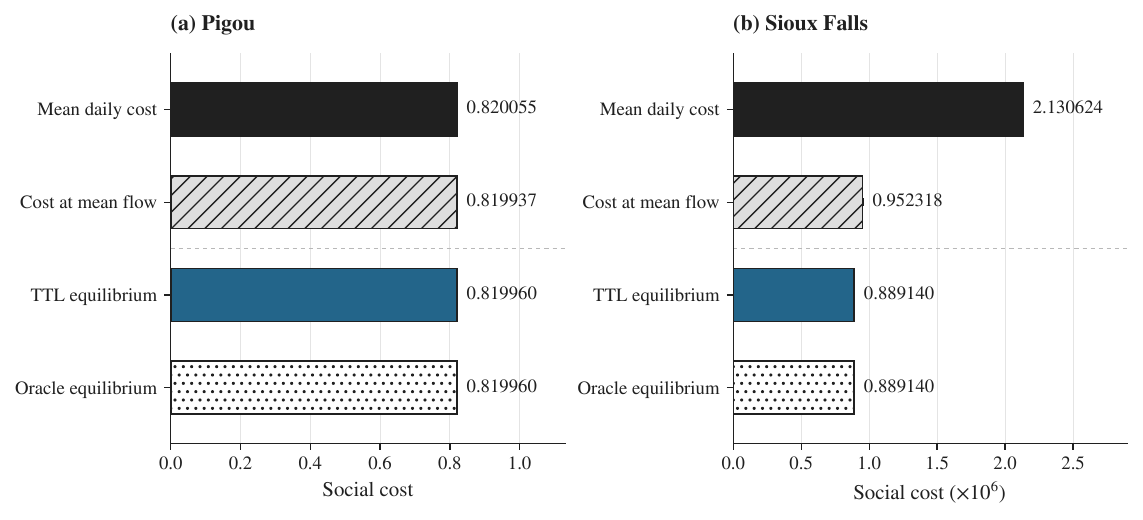}
\caption{Stationary predictions and simulated congestion costs at $\beta=5$. (a) Pigou with full memory ($B=2$); (b) Sioux Falls with $B=2$ of eight routes per OD. The first two bars report mean daily cost, $\overline{\mathrm{SC}}$, and cost at the mean flow, $\mathrm{SC}(\bar x)$. The remaining bars report the TTL equilibrium and an equilibrium with availability marginals frozen at their simulated values (oracle). Numbers beside bars give the plotted values; error bars are 95\% intervals across ten seeds and may be smaller than the bar outlines. The panels use different units. These comparisons combine dynamic variation with model approximation; they do not isolate TTL equilibrium error.}
\label{fig:exp1-costs}
\end{figure}

On Pigou the TTL and exact full-recall predictions coincide at cost $0.819959693$. The mean daily simulation cost is $0.820054946$, with a 95\% half-width of approximately $0.0000246$. This full-memory baseline agrees with the analytic reduction. On Sioux Falls, the mean daily cost is $2.130624$ million, whereas cost at the mean flow is $0.952318$ million. The TTL and frozen-oracle predictions are both approximately $0.889140$ million. Their proximity to each other does not validate either against the lagged dynamics. In particular, the large difference between the first two estimands demonstrates why evaluating only a mean flow can conceal substantial congestion variation.

For share comparisons, we use $\sum_k(d_k/D)\|q_k-\bar q_k\|_1$, where $q_k$ is a within-OD share vector and $D=\sum_k d_k$. The Sioux Falls TTL distance from the simulated mean is approximately $0.300175$, and the oracle distance is $0.299874$ (Table~\ref{tab:exp1-shares}). These are discrepancies between stationary predictions and dynamic averages. An exact 400-route finite-memory FWE was not computed, so these figures do not isolate TTL's equilibrium-model error. Half-window drift measures finite-window variation only.
\begin{table}[!htbp]
\centering\fontsize{8}{9.2}\selectfont
\begin{tabular}{@{}lrrrr@{}}
\toprule
Instance & TTL distance & Oracle distance & Local AWL distance & Window drift \\
\midrule
Pigou & $0.000119 \pm 0.000073$ & $0.000119 \pm 0.000073$ & $0.000227 \pm 0.000139$ & $0.000173 \pm 0.000077$ \\
Sioux Falls & $0.300175 \pm 0.000007$ & $0.299874 \pm 0.000007$ & $0.317466 \pm 0.000007$ & $0.000105 \pm 0.000020$ \\
\bottomrule
\end{tabular}

\caption{Exp-1 demand-weighted share diagnostics (mean $\pm$ 95\% half-width). Local AWL is evaluated at the observed mean flow; oracle distance uses the separately solved frozen-availability equilibrium. Window drift compares the first and second halves of the retained sample.}
\label{tab:exp1-shares}
\end{table}

\FloatBarrier
\subsection{Exp-2: separating memory, availability, and TTL errors}\label{sec:exp_assumptions}
The endogenous constant-cost instance has six routes with costs $(0.5,0.7,0.9,1.1,1.3,1.5)$, $\beta=4$, uniform surfacing, and $B=3$. The 120-state ordered-memory chain yields reference stationary recall and choice probabilities. We separately compare simulation with that reference, TTL recall evaluated at the \emph{exact} choice vector, AWL using \emph{exact} availability marginals, and the endogenous TTL fixed point.
\begin{table}[!htbp]
\centering\small\begin{tabular}{@{}lr@{}}
\toprule
Comparison & $\ell_1$ discrepancy \\
\midrule
Simulation choice vs. exact choice & $0.000307 \pm 0.000067$ \\
Simulation recall vs. exact recall & $0.001613 \pm 0.000788$ \\
TTL recall at exact choices vs. exact recall & 0.162964 \\
AWL using exact availability vs. exact choice & 0.003239 \\
TTL equilibrium vs. exact choice & 0.025885 \\
\bottomrule
\end{tabular}

\caption{Exp-2 separates approximation stages. The first two rows summarize ten simulations; the remaining rows are deterministic comparisons with the exact finite-state reference. Each row compares the indicated quantity with its exact counterpart; the errors are not additive.}
\label{tab:exp2-stages}
\end{table}
\begin{figure}[!htbp]
\centering\includegraphics[width=.94\linewidth]{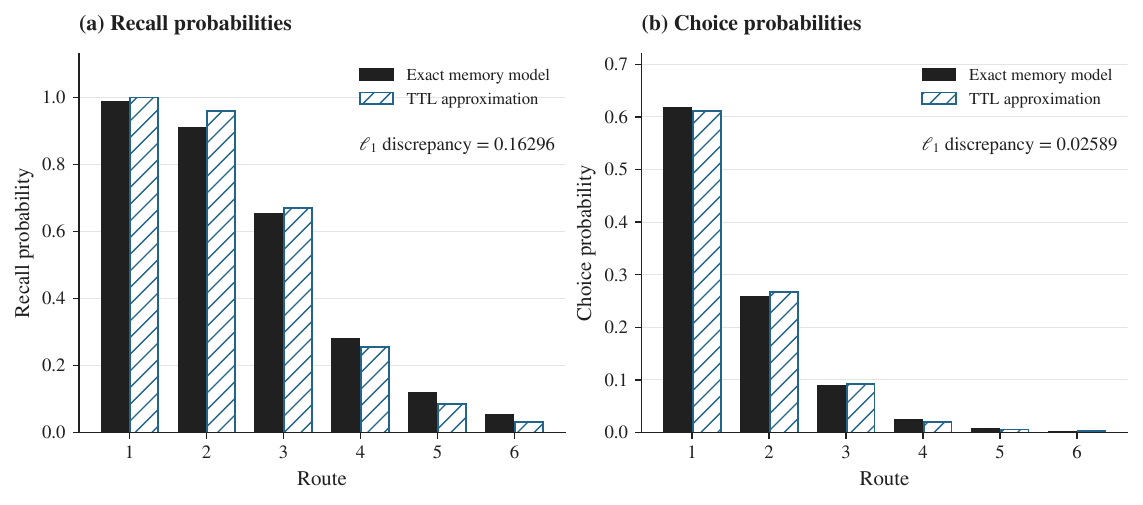}
\caption{Two stages of the larger-memory approximation on six constant-cost routes, with $B=3$, $\beta=4$, and uniform surfacing. Solid bars use the exact 120-state memory chain; hatched bars use TTL. (a) Recall probabilities, with TTL evaluated at the exact choice vector. Recall probabilities need not sum to one: several routes can be remembered together. (b) Route-choice probabilities from the exact model and the endogenous TTL equilibrium, each summing to one. Routes are ordered by increasing cost $(0.5,0.7,0.9,1.1,1.3,1.5)$. The annotated $\ell_1$ discrepancies are sums of absolute coordinate differences. All values are deterministic calculations.}
\label{fig:exp2-endogenous}
\end{figure}
The simulation choice discrepancy averages $0.000307$, while the TTL equilibrium choice discrepancy is $0.025885$. Exact-marginal AWL has a smaller but nonzero discrepancy of $0.003239$. TTL recall differs from exact endogenous recall by $0.162964$ in $\ell_1$. Thus observed availability, endogenous choice, and independent-request cache approximation must be distinguished even on a small instance.

To isolate the external-cache setting, we also simulate IID request streams with catalog sizes $n\in\{6,20,100,500\}$, capacities $\max\{1,\lfloor0.3n\rfloor\}$, and Zipf exponents $z\in\{0,0.6,1.2,2\}$. Each of ten seeds supplies one million requests with the first 100,000 discarded. We record pre-request occupancy, per-object request counts, and conditional hits separately. Table~\ref{tab:exp2-external} reports maximum absolute marginal occupancy discrepancy from the characteristic-time prediction. For discrete versus Poisson TTL, both formulas use the same rounded positive integer timer: $1-(1-p)^\tau$ versus $1-e^{-p\tau}$.
\begin{table}[!htbp]
\centering\fontsize{8}{9.2}\selectfont\begin{tabular}{@{}lrrrrr@{}}
\toprule
Catalog & Capacity & $z=0$ & $z=0.6$ & $z=1.2$ & $z=2$ \\
\midrule
6 & 1 & $0.0007 \pm 0.0002$ & $0.0166 \pm 0.0003$ & $0.0452 \pm 0.0004$ & $0.0796 \pm 0.0005$ \\
20 & 6 & $0.0027 \pm 0.0004$ & $0.0067 \pm 0.0009$ & $0.0122 \pm 0.0008$ & $0.0169 \pm 0.0012$ \\
100 & 30 & $0.0076 \pm 0.0008$ & $0.0084 \pm 0.0010$ & $0.0103 \pm 0.0011$ & $0.0277 \pm 0.0033$ \\
500 & 150 & $0.0201 \pm 0.0022$ & $0.0207 \pm 0.0018$ & $0.0354 \pm 0.0028$ & $0.1500 \pm 0.0082$ \\
\bottomrule
\end{tabular}

\caption{External IID LRU cache: maximum marginal occupancy discrepancy from TTL, mean $\pm$ 95\% half-width over ten streams. Fixed simulation length, rare requests, and correlated occupancy can limit accuracy at large catalog sizes; this table does not estimate an asymptotic convergence rate.}
\label{tab:exp2-external}
\end{table}
\begin{figure}[!htbp]
\centering\includegraphics[width=.94\linewidth]{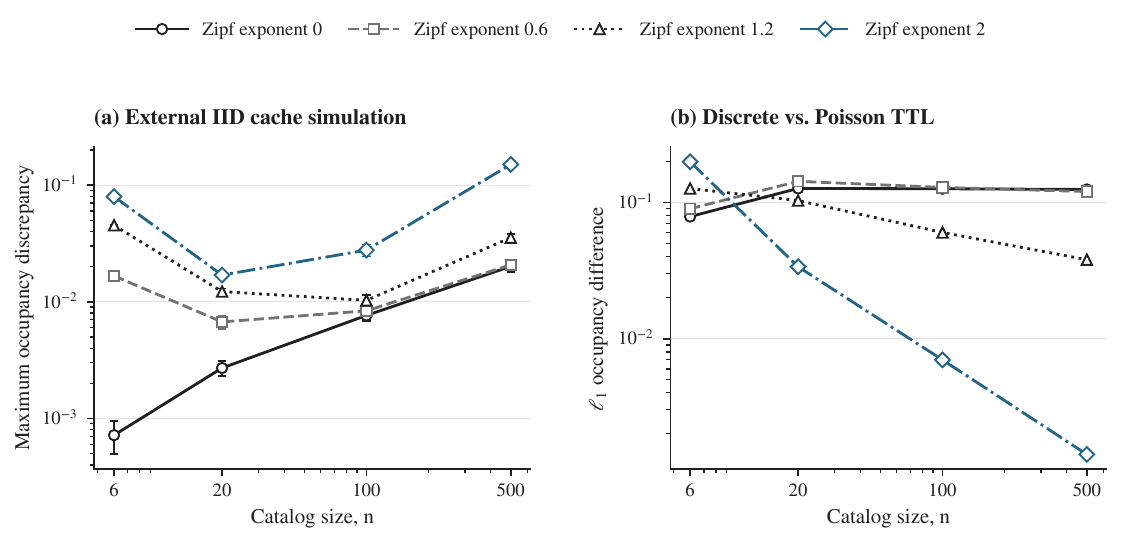}
\caption{External-cache comparisons under independent requests, with memory capacity $B=\max\{1,\lfloor0.3n\rfloor\}$ and Zipf exponents distinguished by line style and marker. (a) Maximum absolute difference between simulated LRU occupancy and the TTL prediction, averaged over ten streams of one million requests (100,000 discarded); error bars are 95\% seed intervals. (b) Exact $\ell_1$ difference between discrete TTL occupancy $1-(1-p)^\tau$ and Poisson TTL occupancy $1-e^{-p\tau}$, using the same rounded positive integer timer. Both axes are logarithmic. Connecting segments guide the eye between evaluated catalog sizes; neither panel estimates an endogenous-routing convergence rate.}
\label{fig:exp2-external}
\end{figure}
The empirical discrepancy does not decrease uniformly with catalog size. For example, at $n=500$ and $z=2$ it is about $0.1500$, substantially larger than at $n=20$. These finite-horizon observations combine approximation error with Monte Carlo and mixing effects; they neither establish nor refute the asymptotic theorem under its hypotheses. In particular, uncertainty across seeds does not bound a common initialization bias.

Finally, identical availability marginals do not determine logit choices. With weights $(1,2,3,4)$, choose menus $\{1,2\}$ and $\{3,4\}$ equiprobably in one model, and $\{1,3\}$ and $\{2,4\}$ equiprobably in another. All four availability marginals equal $1/2$, but the resulting choice vectors are $(1/6,1/3,3/14,2/7)$ and $(1/8,1/6,3/8,1/3)$, whose $\ell_1$ distance is $5/12$. This exact counterexample explains why marginal-hit theory alone cannot supply the missing endogenous joint-menu bound.

\FloatBarrier
\subsection{Exp-3: Braess equilibria and oscillatory dynamics}\label{sec:exp_rbp}
On the three-path Braess instance, we solve the stationary FWE for $B=1,2,3$ using the exact memory law and compare it with independent lagged-population simulations. Uniform-surfacing last-choice memory and full recall both induce ordinary logit at fixed loads; therefore their stationary equilibrium costs must agree. The computed costs are $1.671202613$, $1.691247314$, and $1.671202613$, respectively. The increase from $B=1$ to $B=2$ is about $1.20\%$. A uniform start and two additional interior starts give maximum coordinate spreads below $1.13\times10^{-9}$; this numerical check is not a general uniqueness proof.
\begin{table}[!htbp]
\centering\fontsize{8}{9.2}\selectfont\begin{tabular}{@{}lrrrr@{}}
\toprule
$B$ & Equilibrium cost & Residual & Mean daily cost & Cost at mean flow \\
\midrule
1 & 1.671202613 & $8.87\times10^{-10}$ & $1.671252 \pm 0.000076$ & $1.671218 \pm 0.000076$ \\
2 & 1.691247314 & $9.17\times10^{-10}$ & $1.691326 \pm 0.000045$ & $1.691242 \pm 0.000047$ \\
3 & 1.671202613 & $8.87\times10^{-10}$ & $1.680816 \pm 0.000229$ & $1.660650 \pm 0.000219$ \\
\bottomrule
\end{tabular}

\caption{Exp-3 stationary equilibrium and dynamic costs. Residuals apply to stationary equations. Dynamic columns report mean $\pm$ 95\% half-width over ten finite-population simulations and are not interchangeable with equilibrium costs.}
\label{tab:exp3}
\end{table}
\begin{figure}[!htbp]
\centering\includegraphics[width=\linewidth]{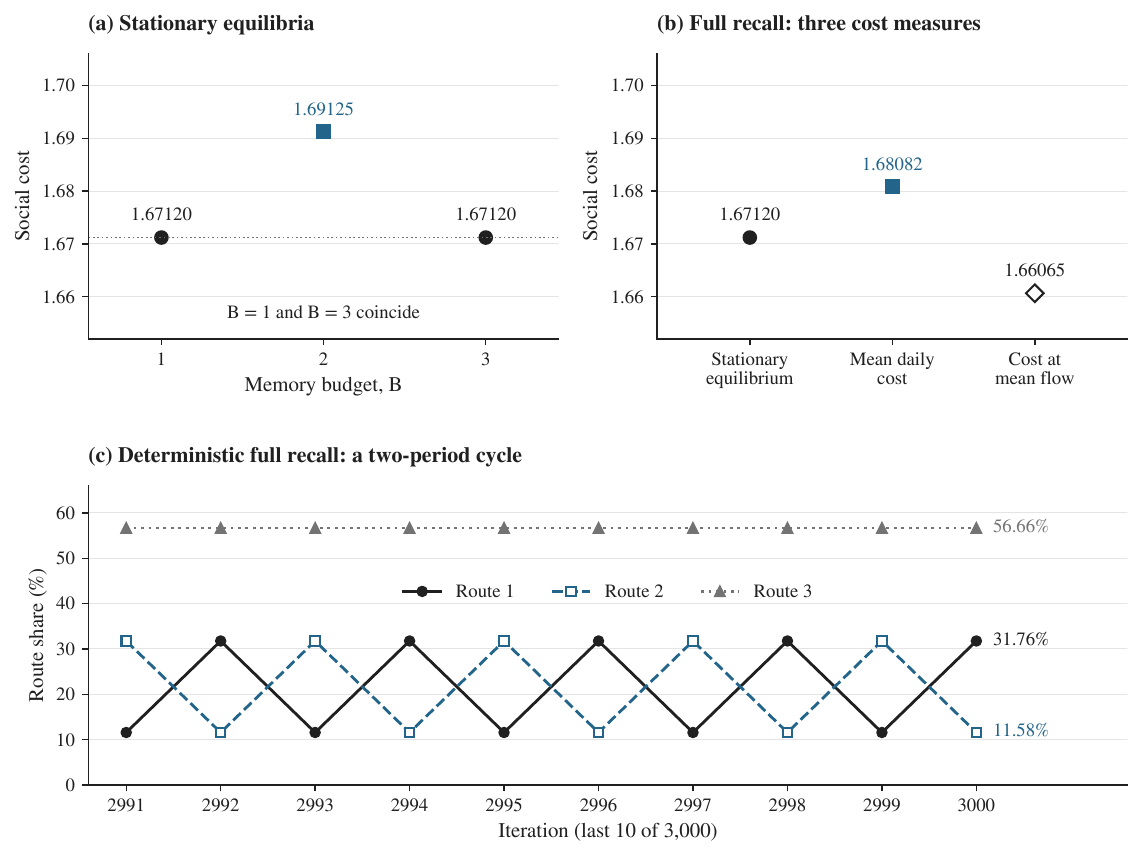}
\caption{Equilibrium and dynamics on the three-route Braess network ($\beta=5$, uniform surfacing). (a) Stationary costs at discrete memory budgets. The equality at $B=1$ and $B=3$ follows from their common stationary logit rule. (b) Full-recall equilibrium cost compared with mean daily cost and cost at the mean flow from ten finite-population simulations; error bars are 95\% seed intervals, smaller than the markers at this scale. These are different cost measures. (c) Last ten iterations of deterministic full-recall adjustment, $q_{t+1}=\operatorname{softmax}[-5L(q_t)]$, initialized at $(0.208,0.206,0.586)$. Routes 1 and 2 exchange shares each iteration while route 3 stays at $56.66\%$; shares sum to one at every step. The connecting lines join discrete iterates. This cycle results from synchronous undamped adjustment and is distinct from a stationary equilibrium. The budget comparison alone does not establish a nested-menu Recall Braess Paradox.}
\label{fig:exp3}
\end{figure}
At full recall, the finite-population mean daily cost is approximately $1.680816$, whereas cost at its mean flow is approximately $1.660650$. Deterministic full-recall iteration from $(0.208,0.206,0.586)$ converges numerically to a two-cycle: the final lag-one share difference is $0.201814$, while the lag-two difference is zero to recorded precision. Over the last 1,500 periods, mean cost is $1.680894482$ and cost at mean flow is $1.660530015$. Neither is the stationary equilibrium cost. 

Memory-budget nonmonotonicity alone does not establish the menu-inclusion order in Definition~\ref{def:rbp}. The invariant menu laws at different capacities have not been shown to satisfy that order. The separate deterministic Pigou construction in Section~\ref{sec:rbp-pigou} does have nested menus: expanding from a population with half restricted to the constant route to full menus for everyone raises cost from $0.75$ to $1$. That is the analytic static-menu paradox; it is not a finite-temperature LRU simulation result.

\FloatBarrier
\subsection{Exp-4: a globally solved influence-budget frontier}\label{sec:exp_governed_design}
For unit-demand Pigou with $\beta=5$, let $x$ denote the variable-route share and define
\[
g(x)=\log\frac{x}{1-x}+5(x-1).
\]
A path-salience ratio budget $R$ permits precisely the interval $|g(x)|\le\log R$. Since $g$ is strictly increasing, its endpoints can be obtained by bracketing. The social cost $x^2-x+1$ is minimized globally by projecting $1/2$ onto this interval. 
\begin{table}[!htbp]
\centering\small\begin{tabular}{@{}lrrr@{}}
\toprule
$R$ & Variable-route share & Optimal cost & Required ratio \\
\midrule
1 & 0.764499 & 0.819960 & 1 \\
1.1 & 0.755406 & 0.815232 & 1.1 \\
1.25 & 0.743023 & 0.809060 & 1.25 \\
1.5 & 0.725013 & 0.800631 & 1.5 \\
2 & 0.695849 & 0.788357 & 2 \\
3 & 0.653439 & 0.773544 & 3 \\
5 & 0.598374 & 0.759677 & 5 \\
10 & 0.521929 & 0.750481 & 10 \\
12.1825 & 0.500000 & 0.750000 & 12.1825 \\
20 & 0.500000 & 0.750000 & 12.1825 \\
\bottomrule
\end{tabular}

\caption{Exp-4 globally optimal Pigou design. These are deterministic results. The required ratio can be below the available budget once the system optimum becomes feasible.}
\label{tab:exp4}
\end{table}
\begin{figure}[!htbp]
\centering\includegraphics[width=.94\linewidth]{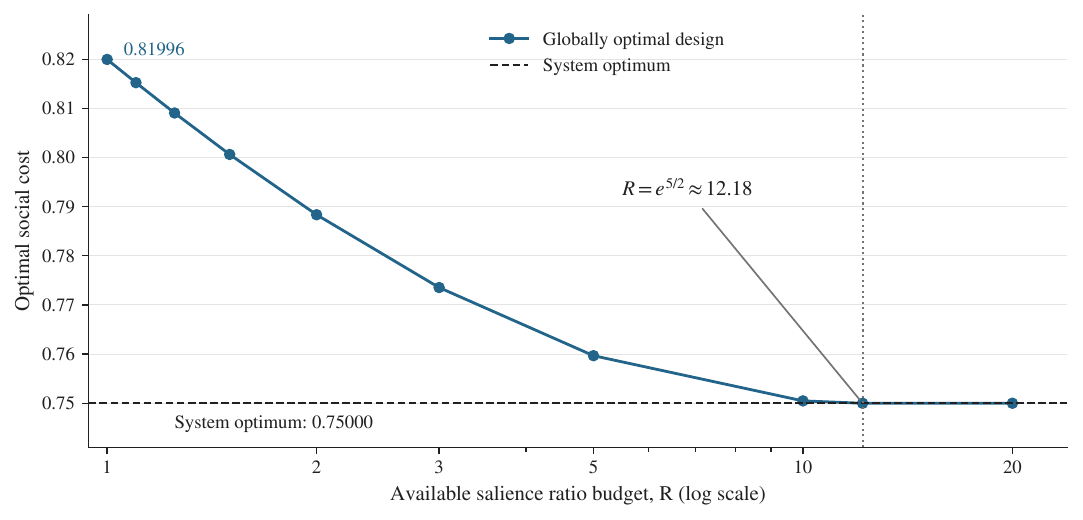}
\caption{Globally optimal influence-budget frontier for unit-demand Pigou at $\beta=5$. Markers give the evaluated budgets; connecting segments guide the eye. At each budget $R$, the minimizer is the projection of the system-optimal variable-route share $1/2$ onto the feasible interval $|\log[x/(1-x)]+5(x-1)|\le\log R$. The system-optimal cost $0.75$ first becomes attainable at $R=e^{5/2}\approx12.18$ and remains optimal for larger budgets. The horizontal axis is logarithmic. This deterministic frontier concerns the specified two-route instance.}
\label{fig:exp4}
\end{figure}
The optimal cost decreases from $0.819960$ at $R=1$ to $0.759677$ at $R=5$, reaches $0.75$ at $R=e^{5/2}$, and remains there for larger budgets. Reconstructing inverse salience and solving the forward equilibrium recovers every target to within $4.5\times10^{-10}$ in maximum share error. The frontier illustrates the scalar design construction on Pigou.

\FloatBarrier
\subsection{Exp-5: target implementation and tying constraints}\label{sec:exp_implementability}
For each of ten seeds we draw 30 independent target share vectors from $\mathrm{Dirichlet}(1,1,1)$ on the Braess route catalog, giving 300 distinct targets. We compute required log-salience $\log\bar f+\beta L(\bar x)$ up to an additive gauge and its unrestricted required ratio. We then check membership in two explicit tying classes: edge-additive salience, with matrix equal to the transposed path-edge incidence matrix, and the one-parameter column $(1,-1,0)^\top$. Membership uses least squares with an additional all-ones gauge and maximum residual tolerance $10^{-9}$; the ratio budget $R=5$ is checked separately.
\begin{table}[!htbp]
\centering\small\begin{tabular}{@{}lr@{}}
\toprule
Quantity & Value \\
\midrule
Distinct sampled targets & 300 \\
Minimum required ratio & 1.285151 \\
Median required ratio & 57.519078 \\
90th percentile required ratio & 1003.868097 \\
Maximum required ratio & 94874.270438 \\
Edge-additive feasible at $R=5$ & 36/300 (12\%) \\
One-parameter feasible at $R=5$ & 0/300 \\
Maximum unrestricted forward share error & $9.898\times10^{-10}$ \\
\bottomrule
\end{tabular}

\caption{Exp-5 descriptive results for 300 targets, each assessed in both tying classes. Quantiles refer to the sampled Dirichlet targets.}
\label{tab:exp5}
\end{table}
\begin{figure}[!htbp]
\centering\includegraphics[width=.94\linewidth]{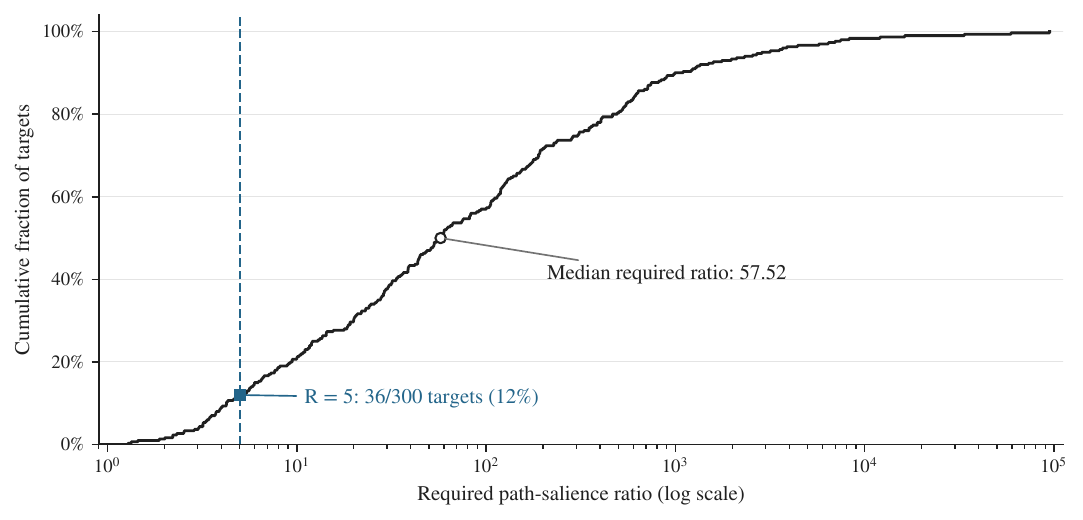}
\caption{Empirical distribution of the unrestricted salience ratio required to implement 300 Braess target flows drawn from $\mathrm{Dirichlet}(1,1,1)$ at $\beta=5$. At horizontal coordinate $R$, the step curve gives the fraction of sampled targets requiring a ratio no larger than $R$. The budget $R=5$ admits 36 targets (12\%); the median required ratio is $57.52$. The horizontal axis is logarithmic. The edge-additive class has full row rank on this route catalog, so these ratios also determine feasibility in that class; they do not describe the restrictive one-parameter tying class.}
\label{fig:exp5}
\end{figure}
The edge-additive matrix has rank three and can express every interior target's salience when edge coefficients are unrestricted. Its maximum tying residual is $1.07\times10^{-14}$. Only 36 of 300 targets (12\%) satisfy $R\le5$; the median required ratio is $57.5191$. None of the sampled targets satisfies the restrictive one-parameter class and the budget together. The one-parameter class contains feasible targets, but its lower-dimensional image makes exact membership restrictive under continuous sampling. Forward equilibrium recovers all unrestricted inverse-salience targets with maximum share error below $9.90\times10^{-10}$. This forward check concerns unrestricted salience, not a forward realization of rejected tied policies.

\FloatBarrier
\subsection{Exp-6: complete solves on a diamond-chain family}\label{sec:exp_scalability}
A chain of $k$ parallel pairs has $2k$ edges and $2^k$ routes. With unit demand, unit salience, and $\beta=5$, independent affine slopes and intercepts are drawn uniformly from $[0.2,1]$ and $[0,0.5]$, respectively. The specialized split solver solves $k$ strictly monotone scalar equilibrium equations. The comparator constructs all path-incidence columns and solves the corresponding logit edge-load fixed point. The benchmark measures complete equilibrium solves on this separable family.

We time five repetitions of each complete solver call, including path construction and allocation for enumeration. Shared coefficient generation and a two-diamond warm-up are excluded. Timings use an Intel Xeon Platinum 8573C with numerical libraries restricted to one thread. We report the median and observed range across repetitions.
\begin{table}[!htbp]
\centering\small\begin{tabular}{@{}lrrrr@{}}
\toprule
$k$ & Routes & Split (ms) & Enumeration (ms) & Max. load difference \\
\midrule
4 & 16 & 0.240 & 0.694 & $3.81\times10^{-10}$ \\
6 & 64 & 0.236 & 0.758 & $4.35\times10^{-10}$ \\
8 & 256 & 0.286 & 0.904 & $3.35\times10^{-10}$ \\
10 & 1024 & 0.320 & 1.476 & $4.26\times10^{-10}$ \\
12 & 4096 & 0.392 & 5.196 & $4.65\times10^{-10}$ \\
14 & 16384 & 0.485 & 26.160 & $4.42\times10^{-10}$ \\
16 & 65536 & 0.623 & 122.269 & $5.86\times10^{-10}$ \\
18 & 262144 & 0.729 & 513.200 & $4.20\times10^{-10}$ \\
\bottomrule
\end{tabular}

\caption{Exp-6 median complete-solve times over five repetitions and maximum observed difference between final edge loads.}
\label{tab:exp6}
\end{table}
\begin{figure}[!htbp]
\centering\includegraphics[width=.94\linewidth]{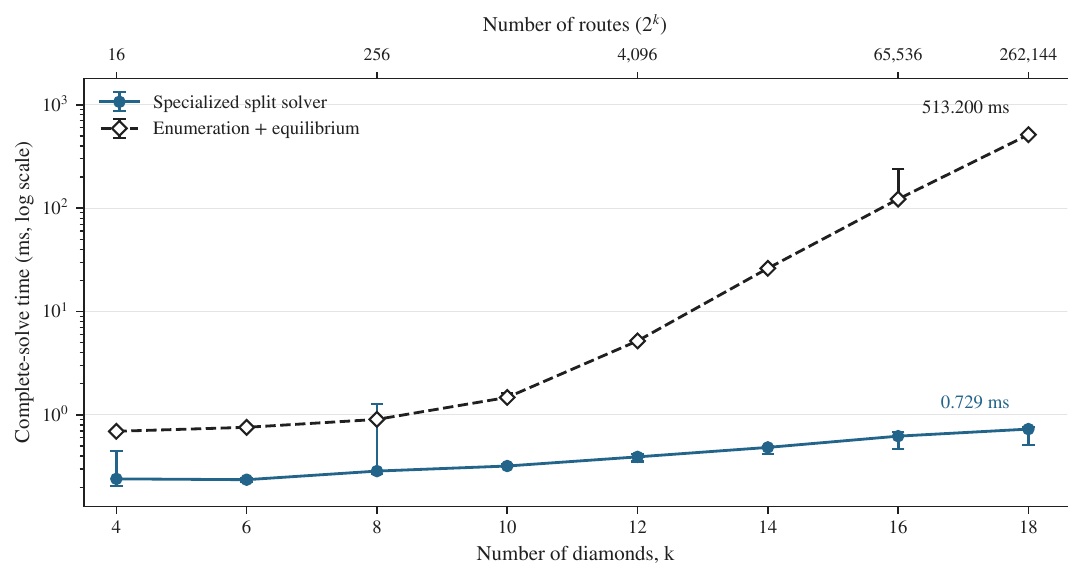}
\caption{Complete equilibrium-solve times on diamond chains with $k$ parallel pairs and $2^k$ routes. Markers show medians of five repetitions; vertical bars show the observed minimum and maximum, not confidence intervals. Enumeration includes construction of the full path catalog, while the split solver exploits the separable chain structure. Shared coefficient generation and warm-up are excluded. Timings use one numerical-library thread on an Intel Xeon Platinum 8573C; the vertical axis is logarithmic. At $k=18$, both methods solve the 262,144-route instance, with medians of $0.729$ ms and $513.200$ ms. This benchmark concerns the specified equilibrium solvers, not a general constrained design optimizer.}
\label{fig:exp6}
\end{figure}
At $k=18$ (262,144 routes), the specialized split solve takes a median $0.729$ ms, versus $513.200$ ms for enumeration and solving. Both methods satisfy their numerical residual checks, and maximum edge-load disagreement across all sizes is $5.87\times10^{-10}$. The result demonstrates the advantage of exploiting this family's decomposition, while the general linear-time partition-oracle result remains distinct from the complexity of a complete nonlinear solve.

\FloatBarrier
\subsection{Summary of findings}
The experiments show three effects: synchronous adjustment can sustain oscillations around a stationary model; joint-menu dependence affects the accuracy of marginal-availability approximations; and influence budgets and tying constraints restrict implementable flows. The Pigou frontier quantifies the benefit of additional influence, while the diamond-chain benchmark illustrates the computational benefit of exploiting network decomposition.

\section{Discussion and limitations}\label{sec:discussion}
The salience model gives exact equilibrium and design results under its assumptions. Applying these results to richer memory dynamics or observed route choice raises two further questions.

\paragraph{Behavioral scope of the salience abstraction.}
Stationary salience policies capture persistent, average effects of recall and guidance on menus,
but do not directly model transient learning, endogenous formation of attention, or rich correlation
structures in route availability. The approximation pipeline in Section~\ref{sec:micro2salience-B}
partially addresses this by endogenizing availability through TTL equations; nevertheless, regimes
with highly concentrated popularity, near-deterministic choice (large $\beta$), or strong overlap
across alternatives can require micro simulation or richer surrogates.

\paragraph{From synthetic to real networks.}
The Sioux Falls experiment retains a restricted set of OD demands and uses an uncalibrated choice-sensitivity parameter. Its results describe the computational model; empirical validation would require observed route choices and a calibrated demand and information model. Finite simulation windows leave mixing uncertain, and the absence of an exact large-network memory equilibrium prevents isolation of the TTL equilibrium approximation error.

\section{Conclusion}
This paper develops a policy-controlled notion of \emph{endogenous information} in routing games, bridging micro-level memory dynamics (LRU recall and surfacing) with macro-level equilibrium analysis and design.
The key technical idea is to expose an \emph{inverse map} from target flows to required salience (log-salience equals $\log f+\beta L$ up to scale), which yields implementability tests and transforms governed policy design into optimization over implementable flows.
On a nontrivial network class (parallel networks), bounded-influence optimal design reduces to a one-dimensional search plus convex subproblems.

The experiments distinguish stationary cost from dynamic averages, illustrate inverse design under influence and tying constraints, and demonstrate computational savings on a diamond-chain family. The larger-memory comparisons identify joint-menu dependence and equilibrium stability as obstacles to an endogenous TTL approximation bound. Further research directions are: (i) extending the tractable design results beyond parallel and series-parallel networks (e.g., graphs with low treewidth, decomposable route-set representations, or column-generation over path features); (ii) incorporating richer, feature-based tying constraints into efficient solvers; and (iii) validating the model on real route-choice datasets with controlled information interventions.

\section*{Data and code availability}
Code and benchmark inputs are available at \url{https://github.com/alqithami/Routing-Games}. The accompanying code-and-results archive contains the full experiment outputs, configurations, and reproduction instructions.


\appendix

\section{Deferred and extended proofs}
\label{app:proofs}
This appendix provides the deferred proofs, using the notation of the main text. The following standing assumptions apply throughout; individual results state any additional hypotheses.

\subsection{Standing assumptions and conventions for Appendix~\ref{app:proofs}}
\label{app:standing-assumptions}
Unless explicitly stated otherwise, the arguments in this appendix use the following standing assumptions.
\begin{enumerate}
    \item[\textbf{A1}] \textbf{Finite path sets.} For each commodity $k\in\cK$, the path set $\cP_k$ is finite; hence the induced memory state space $\cM_k$ is finite.
    \item[\textbf{A2}] \textbf{Latency regularity.} Each edge latency $\ell_e(\cdot)$ is continuous and nondecreasing on $[0,D]$, where $D=\sum_{k\in\cK} d_k$ is the total demand. When strict monotonicity or Lipschitz regularity is needed, it is stated explicitly.
    \item[\textbf{A3}] \textbf{Finite logit temperature.} In the dynamic-recall model, $\beta<\infty$ so every available route has strictly positive choice probability. In salience-weighted SUE results we assume $\beta>0$.
    \item[\textbf{A4}] \textbf{Nonempty availability.} Random menus/available sets are almost surely nonempty, so denominators such as $Z=\sum_{r\in A} w_r$ are well-defined. (This is automatic in the endogenous-recall model because the surfaced route is always available.)
    \item[\textbf{A5}] \textbf{Full-support surfacing.} In the endogenous-recall Markov chain, each surfacing distribution $\rho_k$ has full support on $\cP_k$.
\end{enumerate}

\subsection{Proof of Lemma~\ref{lem:ergodic}}
\label{app:proof-lem-ergodic}
\begin{proof}
Fix a commodity $k$ and a congestion vector $x\in\cX$.
The state space $\cM_k$ is finite.
It therefore suffices to show that the Markov chain with kernel $P_x^{(k)}$ is (i) irreducible and (ii) aperiodic.

\paragraph{Irreducibility.}
Let $m,m'\in\cM_k$ be arbitrary.
Write $m'=(p'^{(1)},\dots,p'^{(B_k)})$.
Consider the event that for $j=1,2,\dots,B_k$ the surfaced route equals
\[
q^{(j)} \;=\; p'^{(B_k-j+1)}.
\]
Because $\rho_k$ has full support on $\cP_k$, this surfacing sequence occurs with strictly positive probability
$\prod_{j=1}^{B_k}\rho_k(q^{(j)})>0$.
Conditional on this surfacing sequence, at each step $j$ the route $q^{(j)}$ belongs to the available set $\cA(m^{(j)},q^{(j)})$ by \eqref{eq:available-set}.
Since $\beta<\infty$, the logit rule \eqref{eq:logit-choice} assigns every available route strictly positive probability, hence
\[
\sigma\!\big(q^{(j)}\mid m^{(j)},q^{(j)},x\big) \;>\;0.
\]
Therefore the joint event that the traveler chooses the surfaced route at every step also has strictly positive probability (the product of these positive terms).

Under the LRU update map $U_k$ in Definition~\ref{def:lru}, whenever the chosen route is not already in the current recalled list it is inserted at the front (and the last element is dropped).
Starting from an arbitrary initial state $m$, after choosing $q^{(1)}=p'^{(B_k)}$ the first element becomes $p'^{(B_k)}$; after choosing $q^{(2)}=p'^{(B_k-1)}$ the first two elements become $(p'^{(B_k-1)},p'^{(B_k)})$; continuing for $B_k$ steps yields exactly
\[
m^{(B_k+1)} \;=\; (p'^{(1)},p'^{(2)},\dots,p'^{(B_k)}) \;=\; m'.
\]
Hence $m'$ is reachable from $m$ with positive probability in $B_k$ steps, implying irreducibility.

\paragraph{Aperiodicity.}
Fix any state $m=(p^{(1)},\dots,p^{(B_k)})$.
For any surfaced route $q$, the route $p^{(1)}$ is always available because $p^{(1)}\in\cS(m)\subseteq \cA(m,q)$.
Under logit with finite $\beta$, $\sigma(p^{(1)}\mid m,q,x)>0$.
Moreover, under LRU, choosing the most-recent route leaves the ordered list unchanged: $U_k(m,p^{(1)})=m$.
Therefore
\[
P_x^{(k)}(m\mid m) \;\ge\; \sum_{q\in\cP_k}\rho_k(q)\,\sigma(p^{(1)}\mid m,q,x) \;>\;0,
\]
so the chain admits a self-loop at every state and is aperiodic.

Since $\cM_k$ is finite and the chain is irreducible and aperiodic, it is ergodic and admits a unique stationary distribution $\pi_k(x)$ with full support.
\end{proof}

\subsection{Proof of Lemma~\ref{lem:pi-continuous}}
\label{app:proof-lem-pi-continuous}
\begin{proof}
Fix a commodity $k$ and let $n\triangleq |\cM_k|$.
For each $x\in\cX$, let $P_x$ denote the $n\times n$ transition matrix of the chain (with rows indexed by $m\in\cM_k$).
Under Assumption~\ref{ass:regularity}, each entry of $P_x$ is a continuous function of $x$ because it is a finite sum of terms of the form $\rho_k(q)\sigma(p\mid m,q,x)$ and $x\mapsto \sigma(\cdot\mid m,q,x)$ is continuous (indeed smooth) whenever $\beta<\infty$ and $x\mapsto L_p(x)$ is continuous.

By Lemma~\ref{lem:ergodic}, for every $x\in\cX$ the chain is ergodic and therefore has a unique stationary distribution.
Let $\pi(x)\in\R^n$ denote the stationary distribution written as a column vector.
It is the unique solution to the linear system
\begin{equation}
\label{eq:stationary-linear-system}
P_x^\top \pi(x)=\pi(x), \qquad \mathbf{1}^\top \pi(x)=1.
\end{equation}
Define the matrix $B_x\in\R^{n\times n}$ by taking $I-P_x^\top$ and replacing its last row by $\mathbf{1}^\top$, and define $b\in\R^n$ by $b_n=1$ and $b_i=0$ for $i<n$.
Then \eqref{eq:stationary-linear-system} is equivalent to
\begin{equation}
\label{eq:Bx}
B_x \pi(x)=b.
\end{equation}
Because the stationary distribution is unique, \eqref{eq:Bx} has a unique solution, hence $B_x$ is invertible.
The mapping $x\mapsto B_x$ is continuous, and matrix inversion is continuous on the set of invertible matrices.
Therefore $x\mapsto \pi(x)=B_x^{-1}b$ is continuous on $\cX$.
\end{proof}

\subsection{Proof of Proposition~\ref{prop:xi-nonempty}}
\label{app:proof-prop-xi-nonempty}
\begin{proof}
Fix a memory profile $\mu$.
Define the map $T_\mu:\cX\to\cX$ by
\[
T_\mu(x)\;\triangleq\; x\!\left(f(\mu,x)\right),
\]
where $f(\mu,x)$ is defined in \eqref{eq:path-flow-induced} and $x(f)$ is defined in \eqref{eq:edge-loads}.

\paragraph{Step 1: $T_\mu$ maps $\cX$ into itself.}
For each commodity $k$, the induced path flows satisfy $\sum_{p\in\cP_k} f_{k,p}(\mu,x)=d_k$ because $\sigma(\cdot\mid m,q,x)$ is a probability distribution over $\cA(m,q)$ and the outer sums in \eqref{eq:path-flow-induced} preserve total mass.
Therefore, for any edge $e$,
\[
0 \;\le\; x_e\!\left(f(\mu,x)\right)
\;=\;
\sum_{k\in\cK}\sum_{p\in\cP_k} f_{k,p}(\mu,x)\,\mathbf{1}\{e\in p\}
\;\le\;
\sum_{k\in\cK}\sum_{p\in\cP_k} f_{k,p}(\mu,x)
\;=\;
\sum_{k\in\cK} d_k
\;=\; D.
\]
Hence $T_\mu(x)\in [0,D]^{|E|}=\cX$.

\paragraph{Step 2: continuity.}
Under Assumption~\ref{ass:regularity}, the latency functions are continuous and the logit rule is continuous in the cost vector.
Thus $(\mu,x)\mapsto f(\mu,x)$ is continuous because it is a finite sum of continuous functions.
The edge-load map $f\mapsto x(f)$ is linear, hence continuous.
Therefore $T_\mu$ is continuous.

\paragraph{Step 3: apply Brouwer.}
The set $\cX$ is compact and convex.
By Brouwer's fixed point theorem, the continuous map $T_\mu:\cX\to\cX$ admits a fixed point $x\in\cX$ satisfying $x=T_\mu(x)$, which is exactly the congestion consistency condition \eqref{eq:congestion-fixed-point}.
Therefore $\Xi(\mu)$ is nonempty.
\end{proof}

\subsection{Proof of Theorem~\ref{thm:fwe-existence}}
\label{app:proof-thm-fwe-existence}
\begin{proof}
Define $\pi(x)\triangleq(\pi_k(x))_{k\in\cK}$ and recall the reduced map
\[
T(x)\;\triangleq\; x\!\left(f(\pi(x),x)\right).
\]

\paragraph{Step 1: continuity of $T$.}
By Lemma~\ref{lem:pi-continuous}, each $x\mapsto \pi_k(x)$ is continuous; thus $x\mapsto \pi(x)$ is continuous.
As in the proof of Proposition~\ref{prop:xi-nonempty}, the map $(\mu,x)\mapsto f(\mu,x)$ is continuous and $f\mapsto x(f)$ is linear.
Therefore $T$ is continuous on $\cX$.

\paragraph{Step 2: $T$ maps $\cX$ into itself.}
The same argument as in Proposition~\ref{prop:xi-nonempty} shows that for any $x\in\cX$, all edge loads in $T(x)$ lie in $[0,D]$, hence $T(x)\in\cX$.

\paragraph{Step 3: apply Brouwer and verify the FWE conditions.}
By Brouwer's fixed point theorem, there exists $x^\star\in\cX$ such that $x^\star=T(x^\star)$.
Define $\mu^\star\triangleq \pi(x^\star)$.
By definition of $\pi_k(x^\star)$, for each commodity $k$ the distribution $\mu_k^\star$ is stationary for the memory kernel induced by $x^\star$, i.e., it satisfies the stationarity requirement in Definition~\ref{def:fwe}.
Moreover, $x^\star=T(x^\star)$ is precisely the within-period congestion consistency condition with $\mu^\star$.
Hence $(x^\star,\mu^\star)$ is a forgetful Wardrop equilibrium.
\end{proof}

\subsection{Proof of Proposition~\ref{prop:swsue-potential}}
\label{app:proof-prop-swsue-potential}
\begin{proof}
Fix a salience policy $s$ and recall the feasible set
\[
\cF \;=\; \Big\{f\ge 0:\; \sum_{p\in\cP_k} f_{k,p}=d_k\ \ \forall k\in\cK\Big\}.
\]

\paragraph{Step 1: strict convexity and existence of a unique minimizer.}
The Beckmann term $f\mapsto \sum_{e\in E}\int_0^{x_e(f)}\ell_e(u)\,du$ is convex because each $\ell_e$ is nondecreasing and $x_e(f)$ is linear in $f$.
The entropic term
\[
f \;\mapsto\; \frac{1}{\beta}\sum_{k,p} f_{k,p}\big(\log f_{k,p}-\log s_{k,p}\big)
\]
is strictly convex on $\cF$ when $\beta>0$ because $z\mapsto z\log z$ is strictly convex on $\R_{>0}$ and $s_{k,p}>0$.
Therefore $\Phi_s$ is strictly convex on the convex set $\cF$ and can have at most one minimizer.
Moreover, $\cF$ is a product of simplices and is compact, and $\Phi_s$ is continuous on $\cF$ (taking $0\log 0=0$), so a minimizer exists.
Hence $\Phi_s$ admits a unique minimizer $f^\star\in\cF$.

\paragraph{Step 2: first-order conditions yield SW-SUE.}
We show that the minimizer is interior, i.e., $f_{k,p}^\star>0$ for all $k,p$.
Fix any commodity $k$ and any path $p$.
If $f_{k,p}^\star=0$, then for any $\epsilon>0$ and any $r\in\cP_k$ with $f_{k,r}^\star>0$ (which must exist because $\sum_p f_{k,p}^\star=d_k>0$), consider shifting $\epsilon$ mass from $r$ to $p$ while preserving feasibility.
The directional derivative of the entropic term in the $p$-direction diverges to $-\infty$ as $f_{k,p}\downarrow 0$, implying that such a perturbation strictly decreases $\Phi_s$ for sufficiently small $\epsilon$, contradicting optimality.
Therefore $f^\star$ is interior.

For interior points, the KKT conditions are necessary and sufficient.
Let $x^\star=x(f^\star)$.
The derivative of the Beckmann term with respect to $f_{k,p}$ equals the path latency:
\[
\frac{\partial}{\partial f_{k,p}} \sum_{e\in E}\int_0^{x_e(f)}\ell_e(u)\,du
\;=\;
\sum_{e\in p}\ell_e(x^\star_e)
\;=\; L_p(x^\star).
\]
The derivative of the entropic term is $\frac{1}{\beta}(\log f_{k,p}^\star + 1 - \log s_{k,p})$.
Introducing Lagrange multipliers $\lambda_k$ for the equality constraints $\sum_{p\in\cP_k} f_{k,p}=d_k$, the first-order conditions read
\[
L_p(x^\star) + \frac{1}{\beta}\big(\log f_{k,p}^\star + 1 - \log s_{k,p}\big) + \lambda_k = 0
\qquad \forall k\in\cK,\ \forall p\in\cP_k.
\]
Rearranging yields
\[
f_{k,p}^\star \;=\; s_{k,p}\,\exp\!\big(-\beta(L_p(x^\star)+\lambda_k)\big)\,e^{-1}.
\]
Using $\sum_{p\in\cP_k} f_{k,p}^\star=d_k$ to eliminate $\lambda_k$ shows that for each commodity $k$,
\[
f_{k,p}^\star
\;=\;
d_k\,
\frac{s_{k,p}\exp(-\beta L_p(x^\star))}
{\sum_{r\in\cP_k} s_{k,r}\exp(-\beta L_r(x^\star))},
\]
which is precisely the SW-SUE condition in Definition~\ref{def:swsue}.

\paragraph{Step 3: equivalence.}
Conversely, any flow satisfying the SW-SUE fixed point satisfies the KKT system above and therefore minimizes $\Phi_s$.
By uniqueness of the minimizer, the SW-SUE is unique and equals $f^\star$.
\end{proof}

\subsection{Proof of Lemma~\ref{lem:pi-lip}}
\label{app:proof-lem-pi-lip}
\begin{proof}
Fix a commodity $k$ and write $P_x$ for $P_x^{(k)}$ and $\pi(x)$ for $\pi_k(x)$ to lighten notation.
Assumption~\ref{ass:doeblin} implies the minorization
\[
P_x(\cdot\mid m) \;=\; \varepsilon_k \nu_k(\cdot) + (1-\varepsilon_k)\,\widetilde P_x(\cdot\mid m)
\qquad \forall m\in\cM_k,
\]
for some stochastic kernel $\widetilde P_x$ (obtained by renormalizing the residual probabilities).

\paragraph{Step 1: a contraction inequality under Doeblin.}
Let $\alpha,\beta$ be any two distributions on $\cM_k$.
Then
\[
\alpha P_x - \beta P_x
\;=\;
(1-\varepsilon_k)\big(\alpha \widetilde P_x - \beta \widetilde P_x\big),
\]
because the common term $\varepsilon_k\nu_k$ cancels.
Since multiplication by a stochastic matrix is nonexpansive in $\ell_1$,
$\|\alpha \widetilde P_x - \beta \widetilde P_x\|_1 \le \|\alpha-\beta\|_1$, hence
\begin{equation}
\label{eq:doeblin-contraction}
\|\alpha P_x - \beta P_x\|_1 \;\le\; (1-\varepsilon_k)\,\|\alpha-\beta\|_1.
\end{equation}

\paragraph{Step 2: perturbation bound for stationary distributions.}
Let $x,y\in\cX$.
Using stationarity, $\pi(x)=\pi(x)P_x$ and $\pi(y)=\pi(y)P_y$, we write
\[
\pi(x)-\pi(y)
\;=\;
\pi(x)P_x - \pi(y)P_y
\;=\;
(\pi(x)-\pi(y))P_x \;+\; \pi(y)(P_x-P_y).
\]
Taking $\ell_1$ norms and applying \eqref{eq:doeblin-contraction} gives
\[
\|\pi(x)-\pi(y)\|_1
\;\le\;
\|(\pi(x)-\pi(y))P_x\|_1 + \|\pi(y)(P_x-P_y)\|_1
\;\le\;
(1-\varepsilon_k)\|\pi(x)-\pi(y)\|_1 + \|\pi(y)(P_x-P_y)\|_1.
\]
Rearranging yields
\[
\|\pi(x)-\pi(y)\|_1 \;\le\; \frac{1}{\varepsilon_k}\,\|\pi(y)(P_x-P_y)\|_1.
\]
Finally, $\pi(y)$ is a convex combination of rows, hence
\[
\|\pi(y)(P_x-P_y)\|_1
\;=\;
\Big\|\sum_{m\in\cM_k}\pi(y)(m)\big(P_x(\cdot\mid m)-P_y(\cdot\mid m)\big)\Big\|_1
\;\le\;
\sup_{m\in\cM_k}\big\|P_x(\cdot\mid m)-P_y(\cdot\mid m)\big\|_1,
\]
which proves the first inequality in Lemma~\ref{lem:pi-lip}.

\paragraph{Step 3: an explicit Lipschitz constant in $\|x-y\|_\infty$.}
Fix $m\in\cM_k$.
Using the definition \eqref{eq:memory-kernel} and the fact that the update map $U_k$ is deterministic, the distribution $P_x(\cdot\mid m)$ is a pushforward of the logit choice distribution on the available set.
Therefore, for each surfaced route $q$,
\[
\big\|P_x(\cdot\mid m,q)-P_y(\cdot\mid m,q)\big\|_1
\;\le\;
\big\|\sigma(\cdot\mid m,q,x)-\sigma(\cdot\mid m,q,y)\big\|_1,
\]
and after averaging over $q\sim\rho_k$ we obtain
\[
\big\|P_x(\cdot\mid m)-P_y(\cdot\mid m)\big\|_1
\;\le\;
\sup_{q\in\cP_k}\big\|\sigma(\cdot\mid m,q,x)-\sigma(\cdot\mid m,q,y)\big\|_1.
\]
By Lemma~\ref{lem:logit-lip}, the logit map is $\beta$-Lipschitz from $\|\cdot\|_\infty$ costs to $\|\cdot\|_1$ probabilities.
For any route $p$, the path cost is $L_p(x)=\sum_{e\in p}\ell_e(x_e)$.
Under Assumption~\ref{ass:lipschitz},
\[
|L_p(x)-L_p(y)|
\;\le\;
\sum_{e\in p} |\ell_e(x_e)-\ell_e(y_e)|
\;\le\;
\sum_{e\in p} L\,|x_e-y_e|
\;\le\;
|p|\,L\,\|x-y\|_\infty
\;\le\;
H L\,\|x-y\|_\infty.
\]
Combining these bounds yields
\[
\sup_{m\in\cM_k}\big\|P_x(\cdot\mid m)-P_y(\cdot\mid m)\big\|_1
\;\le\;
\beta\,H L\,\|x-y\|_\infty.
\]
Substituting into the first inequality proves the second inequality with $C_k=\beta H L/\varepsilon_k$.
\end{proof}

\subsection{Proof of Proposition~\ref{prop:general-pigou}}
\label{app:proof-prop-general-pigou}
\begin{proof}
Since $y^\mathrm{opt}<1$ and $g$ is strictly increasing, $g(y^\mathrm{opt})<g(1)\le c$. Under $\alpha^\star=1-y^\mathrm{opt}$, the suppressed mass uses $a$ and all remaining mass strictly prefers $b$ at the proposed flow, giving the claimed ICWE. Global optimality implies weak improvement. It implies strict improvement only when the objective values differ. If $g$ is convex and strictly increasing, $y\mapsto yg(y)$ is strictly convex on the positive interval, so an interior minimizer is unique and its value is strictly less than at $y=1$.
\end{proof}

\subsection{Proof of Theorem~\ref{thm:rbp-ubiquity}}
\label{app:proof-thm-rbp-ubiquity}
\begin{proof}
Choose two distinct simple $s$--$t$ paths $P,Q$. Each contains an edge absent from the other; otherwise a simple path with the same endpoints would properly contain the other. Select $e_a\in P\setminus Q$ and $e_b\in Q\setminus P$. Put latency $1$ on $e_a$, latency $x$ on $e_b$, and zero on all remaining edges. Restrict the static menus to $\{P\}$ for a mass $\alpha$ and $\{P,Q\}$ for the other mass. Even if the union supports other paths, those paths are outside these admissible menus. The two admissible costs are exactly $1$ and $x_Q$. The Pigou calculation gives $\mathrm{SC}(\alpha)=\alpha+(1-\alpha)^2$; expanding menus from $\alpha=1/2$ to $0$ raises cost from $3/4$ to $1$. No shared-suffix assertion or physical-edge exclusion argument is needed. This proves existence for static nested menus, not for fixed-full-support LRU memory-budget changes. With a unique feasible simple path and nonempty menus, every traveler uses that path and cost is invariant.
\end{proof}

\subsection{Proof of Theorem~\ref{thm:sp-partition}}
\label{app:proof-thm-sp-partition}
\begin{proof}
Let $(G,s,t)$ be a two-terminal series-parallel (SP) network and fix generalized edge costs $\{c_e\}$.
For any two-terminal subnetwork $H$ with terminals $(s_H,t_H)$, define its partition function
\[
Z_H \;\triangleq\; \sum_{p\in\cP(H)} \exp\Big(-\beta\sum_{e\in p} c_e\Big),
\]
where $\cP(H)$ is the set of all $s_H$--$t_H$ paths in $H$.
At the root $H=G$, this coincides with \eqref{eq:partition}.

\paragraph{Bottom-up recursion for $Z_H$.}
We proceed by structural induction along the SP decomposition tree.

\emph{Leaf edge.}
If $H$ is a single edge $e$, then $\cP(H)=\{e\}$ and $Z_H=\exp(-\beta c_e)$.

\emph{Series composition.}
Suppose $H=H_1\otimes H_2$, meaning that $t_{H_1}$ is identified with $s_{H_2}$ and every $s_H$--$t_H$ path is the concatenation of a path in $H_1$ and a path in $H_2$.
Then
\[
Z_H
=
\sum_{p_1\in\cP(H_1)}\sum_{p_2\in\cP(H_2)}
\exp\Big(-\beta\sum_{e\in p_1}c_e\Big)\exp\Big(-\beta\sum_{e\in p_2}c_e\Big)
=
Z_{H_1}Z_{H_2}.
\]

\emph{Parallel composition.}
Suppose $H=H_1\oplus H_2$, meaning that $s_{H_1}=s_{H_2}=s_H$ and $t_{H_1}=t_{H_2}=t_H$ and $\cP(H)=\cP(H_1)\cup\cP(H_2)$ (disjoint union).
Then
\[
Z_H
=
\sum_{p\in\cP(H_1)}\exp\Big(-\beta\sum_{e\in p}c_e\Big)
+
\sum_{p\in\cP(H_2)}\exp\Big(-\beta\sum_{e\in p}c_e\Big)
=
Z_{H_1}+Z_{H_2}.
\]

Thus, one can compute $Z_H$ for every node $H$ in one bottom-up traversal, in total time linear in the number of nodes, i.e., $O(|E|)$.

\paragraph{Edge marginals via a top-down probability pass.}
Let $P_G$ denote the Gibbs distribution on $s$--$t$ paths in $G$ with weight proportional to $\exp(-\beta\sum_{e\in p}c_e)$.
For each node $H$ in the SP tree, define $w_H$ as the probability that a random path $p\sim P_G$ traverses the subnetwork $H$ (equivalently, that $p$ lies in the set of $s_H$--$t_H$ paths of $H$ once restricted to $H$).
At the root, $w_G=1$.

We propagate $w_H$ down the tree as follows.

\emph{Series node $H=H_1\otimes H_2$.}
Every $s_H$--$t_H$ path is a concatenation of an $H_1$ path and an $H_2$ path.
Therefore, conditional on traversing $H$, the random path necessarily traverses both children.
Hence $w_{H_1}=w_{H_2}=w_H$.

\emph{Parallel node $H=H_1\oplus H_2$.}
Conditional on traversing $H$, the random path chooses either $H_1$ or $H_2$.
The probability of choosing $H_1$ equals the total Gibbs weight of paths in $H_1$ divided by that in $H$, i.e., $Z_{H_1}/Z_H$; similarly for $H_2$.
Hence
\[
w_{H_1}=w_H\frac{Z_{H_1}}{Z_H},\qquad
w_{H_2}=w_H\frac{Z_{H_2}}{Z_H}.
\]

By induction on the tree, these recursions compute $w_H$ for every node in time $O(|E|)$.

Finally, for a leaf edge $e$ (viewed as a leaf subnetwork), the event that the random path uses edge $e$ is exactly the event that it traverses the leaf node corresponding to $e$.
Therefore the edge marginal in \eqref{eq:edge-marginal} satisfies $\pi_e=w_e$.
This yields all edge marginals in linear time.
\end{proof}

\subsection{Proof of Proposition~\ref{prop:icwe-existence}}
\label{app:proof-prop-icwe-existence}
\begin{proof}
Fix recall sets $\{\cP_i\}_{i=1}^I$ and demands $\{d_i\}_{i=1}^I$.
Let $\cF_R$ denote the feasible set of (type-indexed) path flows
\[
\cF_R
\;\triangleq\;
\Big\{(f^{(i)}_p)_{i,p}:\ f^{(i)}_p\ge 0,\ \sum_{p\in \cP_i} f^{(i)}_p=d_i\ \text{for each }i\Big\}.
\]
This set is nonempty, compact, and convex.

For $f\in\cF_R$, define edge loads $x(f)$ by \eqref{eq:edge-loads} and the Beckmann potential
\[
\Psi(f)\;\triangleq\;\sum_{e\in E}\int_{0}^{x_e(f)} \ell_e(u)\,du.
\]
Because each $\ell_e$ is continuous, $\Psi$ is continuous.
Because each $\ell_e$ is nondecreasing, $x_e\mapsto \int_0^{x_e}\ell_e(u)\,du$ is convex, and since $f\mapsto x(f)$ is linear, $\Psi$ is convex on $\cF_R$.
Hence $\Psi$ attains its minimum over $\cF_R$ at some $f^\star$.

We now show that $f^\star$ is an ICWE/RCWE in the sense of the ICWE/RCWE definition in Section~7.2.
The partial derivative of $\Psi$ with respect to a path flow coordinate $f^{(i)}_p$ is
\[
\frac{\partial \Psi}{\partial f^{(i)}_p}(f)
\;=\;
\sum_{e\in p}\ell_e\!\big(x_e(f)\big)
\;=\;
L_p\!\big(x(f)\big),
\]
because increasing $f^{(i)}_p$ by an infinitesimal amount increases $x_e$ by the same amount on every edge $e\in p$.
Applying the KKT conditions to the convex program $\min_{f\in \cF_R}\Psi(f)$ yields multipliers $\{\lambda_i\}_{i=1}^I$ such that for each type $i$ and each $p\in \cP_i$,
\[
L_p\!\big(x(f^\star)\big)\ \ge\ \lambda_i,
\qquad\text{with equality whenever }f^{\star(i)}_p>0.
\]
Equivalently, every used path in $\cP_i$ has minimum cost within $\cP_i$, which is exactly the ICWE/RCWE condition.

Finally, assume there is a single origin--destination pair and each $\ell_e$ is strictly increasing on $[0,d]$.
Then each map $x_e\mapsto \int_0^{x_e}\ell_e(u)\,du$ is strictly convex, so the aggregate potential is strictly convex in the edge-load vector $x$.
If $f^\star$ and $\tilde f^\star$ are two minimizers, then their edge loads must coincide: otherwise strict convexity would imply
\[
\Psi\!\big(\tfrac12 f^\star+\tfrac12 \tilde f^\star\big)\ <\ \tfrac12\Psi(f^\star)+\tfrac12\Psi(\tilde f^\star),
\]
contradicting optimality.
Thus equilibrium edge loads are unique (though path flows may not be).
\end{proof}

\subsection{Proof of Corollary~\ref{cor:micro2salience-equilibrium}}
\label{app:proof-cor-micro2salience-equilibrium}
\begin{proof}
Fix a congestion vector $x$.
For $B_k=1$, Theorem~\ref{thm:micro2salience} shows that the stationary per-period route-choice probability for commodity $k$ is
\[
\pi_{k,x}(p)
\;=\;
\frac{\rho_k(p)\exp(-\beta L_p(x))}
{\sum_{r\in\cP_k}\rho_k(r)\exp(-\beta L_r(x))}.
\]
In the coupled routing model, the induced stationary flow therefore satisfies
\[
f_{k,p}
\;=\;
d_k\,\pi_{k,x}(p)
\;=\;
d_k\,\frac{\rho_k(p)\exp(-\beta L_p(x))}
{\sum_{r\in\cP_k}\rho_k(r)\exp(-\beta L_r(x))}.
\]
Since the scaling of salience weights cancels in the logit rule, taking $s_{k,p}\propto \rho_k(p)$ yields exactly the SW-SUE fixed-point condition \eqref{eq:swsue}.
Thus the flow component of any coupled equilibrium is an SW-SUE with salience weights $s_{k,p}\propto \rho_k(p)$.
Uniqueness follows from Proposition~\ref{prop:swsue-potential}, which shows the SW-SUE is the unique minimizer of a strictly convex potential.
\end{proof}

\subsection{Proof of Corollary~\ref{cor:micro-realizability}}
\label{app:proof-cor-micro-realizability}
\begin{proof}
Fix strictly positive salience weights $\{s_{k,p}\}$ and define $\rho_k(p)=s_{k,p}/\sum_{r\in\cP_k}s_{k,r}$.
For each commodity $k$ with $B_k=1$, Theorem~\ref{thm:micro2salience} shows that the stationary within-period choice law induced by the micro model is a salience-weighted logit model with salience proportional to $\rho_k$.
Because salience is defined only up to a per-commodity multiplicative constant, the choice law coincides exactly with the salience model with weights $s_k$.
Applying Corollary~\ref{cor:micro2salience-equilibrium} then yields that the induced network equilibrium is the unique SW-SUE for salience $s$.

For the final statement, let $\bar x$ be any interior flow that is implementable by stationary salience in the sense of Theorem~\ref{thm:implementability-salience}, and let $s$ be a salience vector that implements $\bar x$.
Choosing $\rho$ proportional to $s$ as above makes the $B_k=1$ micro model induce the same SW-SUE, hence implement $\bar x$ without monetary tolls.
\end{proof}

\subsection{Proof of Corollary~\ref{cor:diffuse-poisson}}
\label{app:proof-cor-diffuse-poisson}
\begin{proof}
For each route $p$, Lemma~\ref{lem:poissonization-accuracy} gives
\[
0\le H^{\mathrm{WS}}_p(W)-H^{\mathrm{TTL}}_p(W)
\le \frac{\pi_p^2 W}{1-\pi_p}.
\]
Therefore,
\[
\sup_p \big|H^{\mathrm{WS}}_p(W)-H^{\mathrm{TTL}}_p(W)\big|
\le
\frac{1}{1-\pi_{\max}}\sup_p \pi_p^2 W.
\]
Using $\pi_p^2W = (\pi_p W)\pi_p\le \big(\sup_r \pi_r W\big)\pi_{\max}$ and the assumption $\sup_p \pi_p W=O(1)$ yields
\[
\sup_p \big|H^{\mathrm{WS}}_p(W)-H^{\mathrm{TTL}}_p(W)\big|
\le
O(1)\cdot \frac{\pi_{\max}}{1-\pi_{\max}}
=
O(\pi_{\max}),
\]
as claimed.
\end{proof}

\subsection{Proof of Theorem~\ref{thm:lru-ttl}}
\label{app:proof-thm-lru-ttl}
\begin{proof}The main-text statement is the cited Proposition 5.2 specialized to independent Poisson requests under the stated assumptions. Its conclusion concerns marginal cache hits, not the endogenous joint choice law.
\end{proof}

\subsection{Proof of Proposition~\ref{prop:awl-approx}}
\label{app:proof-prop-awl-approx}
\begin{proof}
Fix $k$ and $x$, and write $w_p=\exp(-\beta L_p(x))$.
Let $I_p\triangleq \mathbf{1}\{p\in A\}$, let $Z\triangleq \sum_{r\in A} w_r$ denote the random denominator in \eqref{eq:random-attention-logit}, and let $\mu\triangleq \E[Z]=\sum_r \eta_{k,r}(x)w_r$.
By Assumption~\textbf{A4} (nonempty availability), $Z>0$ almost surely.

For a fixed route $p\in\cP_k$,
\[
\pi_{k,x}(p)-\frac{\eta_{k,p}(x)w_p}{\mu}
=
\E\!\left[I_p w_p\left(\frac{1}{Z}-\frac{1}{\mu}\right)\right]
=
\E\!\left[I_p w_p\frac{\mu-Z}{\mu Z}\right].
\]
Let $\mathcal{E}\triangleq \{Z\ge \mu/2\}$.

\paragraph{On $\mathcal{E}$.}
Since $Z\ge \mu/2$, we have $1/Z\le 2/\mu$ and therefore
\[
\left|I_p w_p\frac{\mu-Z}{\mu Z}\right|
\le
\frac{2w_p}{\mu^2}|Z-\mu|.
\]
Taking expectations and using $\E|Z-\mu|\le \sqrt{\mathrm{Var}(Z)}$ yields
\[
\E\!\left[\left|I_p w_p\frac{\mu-Z}{\mu Z}\right|\mathbf{1}_{\mathcal{E}}\right]
\le
\frac{2w_p}{\mu}\,\mathrm{cv}(Z)
\le
2\,\frac{w_{\max}}{\mu}\,\mathrm{cv}(Z),
\]
where $w_{\max}\triangleq \max_{r\in\cP_k} w_r$.

\paragraph{On $\mathcal{E}^c$.}
On $\{Z<\mu/2\}$ we use the simple bound
\[
\left|I_p w_p\frac{\mu-Z}{\mu Z}\right|
\le
I_p w_p\frac{\mu}{\mu Z}
=
\frac{I_p w_p}{Z}
\le 1,
\]
because whenever $I_p=1$ the denominator satisfies $Z\ge w_p$.
Hence
\[
\E\!\left[\left|I_p w_p\frac{\mu-Z}{\mu Z}\right|\mathbf{1}_{\mathcal{E}^c}\right]
\le
\Pr(Z<\mu/2)
\le
\Pr(|Z-\mu|\ge \mu/2)
\le
4\,\mathrm{cv}(Z)^2,
\]
by Chebyshev's inequality.

Combining the two bounds yields
\begin{equation}
\label{eq:awl-proof-intermediate}
\left|\pi_{k,x}(p)-\frac{\eta_{k,p}(x)w_p}{\mu}\right|
\le
2\,\frac{w_{\max}}{\mu}\,\mathrm{cv}(Z)
+
4\,\mathrm{cv}(Z)^2.
\end{equation}

Finally, if $w_{\max}/\mu\ge 1$, then the claimed error bound
$O(\mathrm{cv}(Z)^2+w_{\max}/\mu)$ holds trivially because the left-hand side is at most $1$.
If instead $w_{\max}/\mu\le 1$, then
$2(w_{\max}/\mu)\mathrm{cv}(Z)\le (w_{\max}/\mu)^2+\mathrm{cv}(Z)^2 \le (w_{\max}/\mu)+\mathrm{cv}(Z)^2$,
so \eqref{eq:awl-proof-intermediate} implies
\[
\pi_{k,x}(p)
=
\frac{\eta_{k,p}(x)w_p}{\mu}
+
O\!\left(\mathrm{cv}(Z)^2+\frac{w_{\max}}{\mu}\right).
\]
Substituting $\mu=\sum_r \eta_{k,r}(x)w_r$ yields \eqref{eq:awl}.
\end{proof}

\subsection{Proof of Corollary~\ref{cor:awl-rate}}

\label{app:proof-cor-awl-rate}
\begin{proof}
The condition $\mathbb{E}|A|\ge B$ gives $\mu\ge Bw_{\min}$. Independence implies $\mathrm{cv}(Z)^2\le w_{\max}/\mu$. Proposition~\ref{prop:awl-approx} then gives the stated coordinatewise bound. This is not a dimension-free $\ell_1$ estimate and does not apply automatically to the fixed-size LRU menu.
\end{proof}

\subsection{Explanation of Remark~\ref{cor:lru-combined-rate}}
\label{app:proof-cor-lru-combined-rate}
\begin{proof}
The marginal and coordinatewise bounds do not establish a combined endogenous-choice rate. To see the missing information, compare two equally likely two-route menus on four routes with weights $(1,2,3,4)$. Menus $\{1,2\},\{3,4\}$ and menus $\{1,3\},\{2,4\}$ give identical inclusion marginals $(1/2,1/2,1/2,1/2)$ but route-choice vectors differing by $5/12$ in $\ell_1$. Thus marginal-hit agreement alone cannot justify the required choice-law bound. This example identifies a missing implication; it is not a counterexample to the external cache theorem.
\end{proof}

\subsection{Proof of Theorem~\ref{thm:swsue-unique-stable}}
\label{app:proof-thm-swsue-unique-stable}
\begin{proof}
Proposition~\ref{prop:swsue-potential} shows that an SW-SUE is exactly a minimizer of the strictly convex potential $\Phi_s$ over the compact convex feasible set $\cF$.
Since $\Phi_s$ is continuous, a minimizer exists; since $\Phi_s$ is strictly convex, the minimizer is unique.
Therefore the SW-SUE exists and is unique.

For the algorithmic statement, let $\{f^{(t)}\}$ be the iterates of any method that globally minimizes $\Phi_s$ over $\cF$ (in the sense that $\Phi_s(f^{(t)})\downarrow \min_{f\in\cF}\Phi_s(f)$ and every limit point is a minimizer).
Because the minimizer is unique, every limit point equals the unique minimizer $f^\star$, and hence $f^{(t)}\to f^\star$.
\end{proof}

\subsection{Proof of Lemma~\ref{lem:logit-lip}}
\label{app:proof-lem-logit-lip}
\begin{proof}
Let $\cA$ be a finite action set and define $\sigma(\cdot\mid c)\in\Delta(\cA)$ by
$\sigma(a\mid c)=\exp(-\beta c_a)/\sum_{b\in\cA}\exp(-\beta c_b)$.
Fix $c,c'\in\R^{|\cA|}$ and set $h=c'-c$.
By the mean value theorem,
\[
\sigma(\cdot\mid c')-\sigma(\cdot\mid c)
=
\int_0^1 J(c+t h)\,h\,dt,
\]
where $J(\cdot)$ is the Jacobian of $\sigma$.
A direct calculation gives
\[
J_{ab}(c)
=
\frac{\partial \sigma(a\mid c)}{\partial c_b}
=
-\beta\,\sigma(a\mid c)\big(\mathbf{1}\{a=b\}-\sigma(b\mid c)\big).
\]
Fix any vector $v$ with $\|v\|_\infty\le 1$ and let $m\triangleq \sum_{b}\sigma(b\mid c)\,v_b$.
Then
\[
(J(c)v)_a
=
-\beta\,\sigma(a\mid c)\big(v_a-m\big),
\]
so
\[
\|J(c)v\|_1
=
\beta\sum_a \sigma(a\mid c)\,|v_a-m|
=
\beta\,\E\big[|V-m|\big],
\]
where $V$ is a random variable supported on $\{v_a\}$ with law $\sigma(\cdot\mid c)$.
Since $V\in[-1,1]$ and $m=\E[V]$, Cauchy--Schwarz yields
$\E[|V-m|]\le \sqrt{\E[(V-m)^2]}=\sqrt{\mathrm{Var}(V)}\le 1$.
Therefore $\|J(c)v\|_1\le \beta$ for all $\|v\|_\infty\le 1$, i.e., the operator norm $\|J(c)\|_{\infty\to 1}\le \beta$ uniformly in $c$.

Finally,
\[
\|\sigma(\cdot\mid c')-\sigma(\cdot\mid c)\|_1
\le
\int_0^1 \|J(c+t h)\|_{\infty\to 1}\,dt\ \|h\|_\infty
\le
\beta\,\|c'-c\|_\infty,
\]
as claimed.
\end{proof}

\subsection{Proof of Proposition~\ref{prop:T-lip}}
\label{app:proof-prop-T-lip}
\begin{proof}
Fix $x,y\in\cX$.
Recall $T(x)=x(f(\pi(x),x))$, where edge loads are linear in path flows \eqref{eq:edge-loads}.
Let $f_x\triangleq f(\pi(x),x)$ and $f_y\triangleq f(\pi(y),y)$.
Then for each edge $e$,
\[
|T_e(x)-T_e(y)|
=
|x_e(f_x)-x_e(f_y)|
\le
\sum_{k}\sum_{p\in\cP_k} |(f_x)_{k,p}-(f_y)_{k,p}|
=
\|f_x-f_y\|_1,
\]
so $\|T(x)-T(y)\|_\infty\le \|f_x-f_y\|_1$.

We decompose
\[
\|f_x-f_y\|_1
\le
\|f(\pi(x),x)-f(\pi(x),y)\|_1
+
\|f(\pi(x),y)-f(\pi(y),y)\|_1.
\]

\paragraph{Step 1: sensitivity to congestion for fixed memory.}
Fix $\mu$ and consider the map $x\mapsto f(\mu,x)$.
For each commodity $k$, each memory state $m$, and each surfaced candidate $q$, the logit choice probabilities $\sigma(\cdot\mid m,q,x)$ depend on $x$ only through the menu path costs.
Under Assumption~\ref{ass:lipschitz}, each edge latency is $L$-Lipschitz on $\cX$ and each path contains at most $H$ edges, hence
\[
\max_{p\in\cP_k}|L_p(x)-L_p(y)|\le H L\,\|x-y\|_\infty.
\]
Applying Lemma~\ref{lem:logit-lip} on the menu action set yields that, for each $(m,q)$,
\[
\big\|\sigma(\cdot\mid m,q,x)-\sigma(\cdot\mid m,q,y)\big\|_1
\le
\beta\,H\,L\,\|x-y\|_\infty.
\]
Averaging over $m\sim \mu_k$ and $q\sim \rho_k$ and multiplying by demand $d_k$ in \eqref{eq:path-flow-induced} gives
\[
\|f_k(\mu,x)-f_k(\mu,y)\|_1
\le
d_k\,\beta\,H\,L\,\|x-y\|_\infty.
\]
Summing over $k$ yields
\begin{equation}
\label{eq:f-cong-lip}
\|f(\mu,x)-f(\mu,y)\|_1
\le
D\,\beta\,H\,L\,\|x-y\|_\infty.
\end{equation}

\paragraph{Step 2: sensitivity to memory for fixed congestion.}
Fix $y$ and two memory profiles $\mu,\mu'$.
For each commodity $k$, the induced flow $f_k(\mu,y)$ is a convex combination of probability vectors indexed by memory states $m\in\cM_k$.
Therefore
\[
\|f_k(\mu,y)-f_k(\mu',y)\|_1
\le
d_k\,\|\mu_k-\mu_k'\|_1.
\]
Taking $\mu=\pi(x)$ and $\mu'=\pi(y)$ and using Lemma~\ref{lem:pi-lip} yields
\begin{equation}
\label{eq:f-mem-lip}
\|f(\pi(x),y)-f(\pi(y),y)\|_1
\le
\sum_k d_k\,\|\pi_k(x)-\pi_k(y)\|_1
\le
D\Big(\max_k C_k\Big)\|x-y\|_\infty.
\end{equation}

Combining \eqref{eq:f-cong-lip} and \eqref{eq:f-mem-lip} yields
\[
\|T(x)-T(y)\|_\infty
\le
D\Big(\beta\,H\,L+\max_k C_k\Big)\|x-y\|_\infty,
\]
which proves the claim with $\kappa$ as in \eqref{eq:kappa}.
\end{proof}

\subsection{Proof of Theorem~\ref{thm:unique}}
\label{app:proof-thm-unique}
\begin{proof}
By Proposition~\ref{prop:T-lip}, $T$ is $\kappa$-Lipschitz on $\cX$ with $\kappa<1$.
Because $\cX$ is compact (hence complete) under $\|\cdot\|_\infty$, $T$ is a contraction mapping on a complete metric space.
Banach's fixed point theorem therefore implies that $T$ has a unique fixed point $x^\star\in\cX$.

Moreover, for any initialization $x^{(0)}\in\cX$, the iterates $x^{(t+1)}=T(x^{(t)})$ converge to $x^\star$ at a geometric rate:
$\|x^{(t)}-x^\star\|_\infty\le \kappa^t\|x^{(0)}-x^\star\|_\infty$.
Setting $\mu_k^\star=\pi_k(x^\star)$ then yields a unique FWE $(x^\star,\mu^\star)$.

Finally, the flow iteration induced by $x^{(t)}$ converges as well because $f(\pi(x),x)$ is continuous in $x$ under Assumptions~\ref{ass:regularity}, \ref{ass:lipschitz}, and Lemma~\ref{lem:pi-continuous}.

For the coupled dynamics set $a=\beta HL$, $\epsilon=\min_k\varepsilon_k$, $X_t=\|x_t-x^\star\|_\infty$, and $M_t=\max_k\|\mu_{k,t}-\mu_k^\star\|_1$. The preceding sensitivity bounds imply
\[
 X_{t+1}\le(1-\eta+\eta Da)X_t+\eta D M_t,
 \qquad M_{t+1}\le aX_t+(1-\epsilon)M_t.
\]
The condition $\kappa=Da(1+1/\epsilon)<1$ allows a number $r$ with $a/\epsilon<r<(1-Da)/D$. For $a=0$, take any $0<r<1/D$. In the weighted maximum norm $\max\{X_t,M_t/r\}$ the contraction factor is at most
\[
 \max\{1-\eta+\eta Da+\eta Dr,\;1-\epsilon+a/r\}<1.
\]
This proves the coupled assertion for $0<\eta\le1$. These bounds require one-step uniform mixing of the actual kernel. They are not implied by finite-state ergodicity alone.
\end{proof}

\subsection{Proof of Theorem~\ref{thm:implicit}}
\label{app:proof-thm-implicit}
\begin{proof}
Define $F:\cX\times \Theta\to\R^{|E|}$ by
\[
F(x,\theta)\triangleq T_\theta(x)-x.
\]
By assumption, $F$ is continuously differentiable.
Moreover,
\[
\nabla_x F(x,\theta)=\nabla_x T_\theta(x)-I.
\]
At $(x^\star,\theta)$, the nondegeneracy condition \eqref{eq:regular-fixed-point} is exactly $\det(\nabla_x F(x^\star,\theta))\neq 0$, so $\nabla_x F(x^\star,\theta)$ is invertible.
Therefore, by the implicit function theorem, there exists a neighborhood $U$ of $\theta$ and a unique continuously differentiable map $\theta'\mapsto x^\star(\theta')$ on $U$ such that $F(x^\star(\theta'),\theta')=0$ for all $\theta'\in U$ and $x^\star(\theta)=x^\star$.
Equivalently, $x^\star(\theta')$ is the unique fixed point of $T_{\theta'}$ in a neighborhood of $x^\star$.

Differentiating the identity $F(x^\star(\theta),\theta)=0$ with respect to $\theta$ and rearranging yields
\[
\nabla_x F(x^\star,\theta)\,\frac{dx^\star}{d\theta}
+
\nabla_\theta F(x^\star,\theta)
=
0,
\]
so
\[
\frac{dx^\star}{d\theta}
=
-\big(\nabla_x F(x^\star,\theta)\big)^{-1}\nabla_\theta F(x^\star,\theta)
=
\big(I-\nabla_x T_\theta(x^\star)\big)^{-1}\nabla_\theta T_\theta(x^\star),
\]
which is \eqref{eq:implicit-derivative}.
\end{proof}

\subsection{Proof of Corollary~\ref{cor:group-blind}}
\label{app:proof-cor-group-blind}
\begin{proof}
At an SW-SUE with salience utilities $u$, the route-share vector for group $k$ is
\[
\frac{f_{k,p}}{d_k}
=
\frac{\exp(u_{k,p})\exp(-\beta L_p(x))}
{\sum_{r\in \cP}\exp(u_{k,r})\exp(-\beta L_r(x))}.
\]
If salience is group-blind, then $u_{k,p}=u_{k',p}$ for all $k,k'$ and all $p$, and the right-hand side is independent of $k$.
Therefore $f_{k,p}/d_k=f_{k',p}/d_{k'}$ for all groups and routes, proving the first claim.
The non-implementability of targets with heterogeneous route shares across groups follows immediately.
\end{proof}

\subsection{Proof of Corollary~\ref{cor:sp-local-budgets}}
\label{app:proof-cor-sp-local-budgets}
\begin{proof}
Theorem~\ref{thm:sp-implementability} characterizes implementability on an SP network via the node-wise inverse-bias recursion \eqref{eq:sp-inverse-delta}.
At any parallel node $H$, the required bias is
\[
\delta_H
=
\log\frac{\bar d_{H_L}}{\bar d_{H_R}}
-\big(V_{H_L}(\bar x,\delta)-V_{H_R}(\bar x,\delta)\big).
\]
Imposing the per-node ratio budget $|\delta_H|\le \log R_H$ is therefore equivalent to the local inequality in the statement.

Conversely, if the local inequalities hold for every parallel node, then the recursively defined $\delta$ is feasible under the budgets.
Applying Theorem~\ref{thm:sp-implementability} with this feasible $\delta$ yields a decomposition-tied salience vector that implements $\bar x$.
The linear-time claim follows because both the forward computation of $(V_H)$ and the backward computation of $(\delta_H)$ traverse the SP decomposition tree once, i.e., in $O(|E|)$ time.
\end{proof}

\subsection{Proof of Proposition~\ref{prop:sp-gradient}}
\label{app:proof-prop-sp-gradient}
\begin{proof}
The second-variation identity in Theorem~\ref{thm:sp-split-convex} handles changes in both $y_H$ and its incoming flow $d_H$. It is positive definite after fixing the root demand. On a compact convex interior domain, continuity and compactness give a positive lower eigenvalue bound. Lipschitz latencies give a Lipschitz gradient for the Beckmann terms, while the local entropy derivatives are bounded away from zero branch flows. Hence the objective is $L$-smooth and $\mu$-strongly convex on that domain. Standard projected-gradient analysis gives the stated linear rate to the domain-constrained minimizer. Equality with the unrestricted equilibrium requires that this equilibrium lie in the domain. Flow and reverse derivative passes each visit $O(|E|)$ nodes; this does not bound the cost of projection.
\end{proof}

\subsection{Proof of Corollary~\ref{cor:implement-so}}
\label{app:proof-cor-implement-so}
\begin{proof}
For an interior optimum apply Theorem~\ref{thm:implementability-salience}. Otherwise fix a desired welfare tolerance $\varepsilon>0$. Continuity on the compact feasible flow set gives a $\delta>0$ such that $\|f-f^{\mathrm{SO}}\|_1<\delta$ implies $|\mathrm{SC}(f)-\mathrm{SC}(f^{\mathrm{SO}})|<\varepsilon$. Mix $f^{\mathrm{SO}}$ with a strictly positive feasible flow using a sufficiently small mixing coefficient, then implement that mixture by Theorem~\ref{thm:implementability-salience}. Flow and welfare tolerances need not be numerically equal.
\end{proof}

\bibliographystyle{plainnat}
\bibliography{references}

\begin{thebibliography}{48}
\providecommand{\natexlab}[1]{#1}
\providecommand{\url}[1]{\texttt{#1}}
\expandafter\ifx\csname urlstyle\endcsname\relax
  \providecommand{\doi}[1]{doi: #1}\else
  \providecommand{\doi}{doi: \begingroup \urlstyle{rm}\Url}\fi

\bibitem[Acemoglu et~al.(2018)Acemoglu, Makhdoumi, Malekian, and
  Ozdaglar]{acemoglu2018ibp}
Daron Acemoglu, Ali Makhdoumi, Azarakhsh Malekian, and Asuman Ozdaglar.
\newblock Informational {Braess}' paradox: The effect of information on traffic
  congestion.
\newblock \emph{Operations Research}, 66\penalty0 (4):\penalty0 893--917, 2018.
\newblock \doi{10.1287/opre.2017.1712}.

\bibitem[Alqithami(2025)]{alqithami2025fifa}
Saad Alqithami.
\newblock Forgetful but faithful: A cognitive memory architecture and benchmark
  for privacy-aware generative agents.
\newblock \emph{arXiv preprint arXiv:2512.12856}, 2025.
\newblock URL \url{https://arxiv.org/abs/2512.12856}.

\bibitem[Alqithami(2026)]{alqithami2026snam}
Saad Alqithami.
\newblock Dynamic homophily with imperfect recall: Modeling resilience in
  adversarial networks.
\newblock \emph{Social Network Analysis and Mining}, 16, 2026.
\newblock \doi{10.1007/s13278-025-01483-2}.
\newblock Article 5.

\bibitem[Beckmann et~al.(1956)Beckmann, McGuire, and Winsten]{beckmann1956}
Martin Beckmann, C.~B. McGuire, and C.~B. Winsten.
\newblock \emph{Studies in the Economics of Transportation}.
\newblock Yale University Press, 1956.

\bibitem[Ben-Akiva and Lerman(1985)]{benakiva1985}
Moshe Ben-Akiva and Steven~R. Lerman.
\newblock \emph{Discrete Choice Analysis: Theory and Application to Travel
  Demand}.
\newblock MIT Press, 1985.

\bibitem[Bergemann and Morris(2019)]{bergemann2019}
Dirk Bergemann and Stephen Morris.
\newblock Information design: A unified perspective.
\newblock \emph{Journal of Economic Literature}, 57\penalty0 (1):\penalty0
  44--95, 2019.
\newblock \doi{10.1257/jel.20181489}.

\bibitem[Biega et~al.(2018)Biega, Gummadi, and Weikum]{biega2018}
Asia~J. Biega, Krishna~P. Gummadi, and Gerhard Weikum.
\newblock Equity of attention: Amortizing individual fairness in rankings.
\newblock In \emph{Proceedings of the 41st International ACM SIGIR Conference
  on Research and Development in Information Retrieval (SIGIR)}, pages
  405--414. ACM, 2018.
\newblock \doi{10.1145/3209978.3210063}.

\bibitem[Braess(1968)]{braess1968}
Dietrich Braess.
\newblock {\"U}ber ein paradoxon aus der verkehrsplanung.
\newblock \emph{Unternehmensforschung}, 12:\penalty0 258--268, 1968.
\newblock \doi{10.1007/BF01918335}.

\bibitem[Caplin and Dean(2015)]{caplin2015}
Andrew Caplin and Mark Dean.
\newblock Revealed preference, rational inattention, and costly information
  acquisition.
\newblock \emph{American Economic Review}, 105\penalty0 (7):\penalty0
  2183--2203, 2015.
\newblock \doi{10.1257/aer.20140117}.

\bibitem[Cascetta and Cantarella(1991)]{cascetta1991}
Ennio Cascetta and Giulio~Erberto Cantarella.
\newblock A day-to-day and within-day dynamic stochastic assignment model.
\newblock \emph{Transportation Research Part A: General}, 25\penalty0
  (5):\penalty0 277--291, 1991.
\newblock \doi{10.1016/0191-2607(91)90144-F}.

\bibitem[Cattaneo et~al.(2020)Cattaneo, Ma, Masatlioglu, and
  Suleymanov]{cattaneo2020ram}
Matias~D. Cattaneo, Xinwei Ma, Yusufcan Masatlioglu, and Elchin Suleymanov.
\newblock A random attention model.
\newblock \emph{Journal of Political Economy}, 128\penalty0 (7):\penalty0
  2796--2836, 2020.
\newblock \doi{10.1086/706861}.

\bibitem[Celis et~al.(2018)Celis, Straszak, and Vishnoi]{celis2018}
L.~Elisa Celis, Damian Straszak, and Nisheeth~K. Vishnoi.
\newblock {Ranking with Fairness Constraints}.
\newblock In Ioannis Chatzigiannakis, Christos Kaklamanis, D\'{a}niel Marx, and
  Donald Sannella, editors, \emph{45th International Colloquium on Automata,
  Languages, and Programming (ICALP 2018)}, volume 107 of \emph{Leibniz
  International Proceedings in Informatics (LIPIcs)}, pages 28:1--28:15,
  Dagstuhl, Germany, 2018. Schloss Dagstuhl -- Leibniz-Zentrum f{\"u}r
  Informatik.
\newblock ISBN 978-3-95977-076-7.
\newblock \doi{10.4230/LIPIcs.ICALP.2018.28}.
\newblock URL
  \url{https://drops.dagstuhl.de/entities/document/10.4230/LIPIcs.ICALP.2018.28}.

\bibitem[Che et~al.(2002)Che, Tung, and Wang]{che2002}
Hao Che, Ye~Tung, and Zhijun Wang.
\newblock Hierarchical web caching systems: Modeling, design and experimental
  results.
\newblock \emph{IEEE Journal on Selected Areas in Communications}, 20\penalty0
  (7):\penalty0 1305--1314, 2002.
\newblock \doi{10.1109/JSAC.2002.801752}.

\bibitem[Christodoulou and Koutsoupias(2005)]{christodoulou2005}
George Christodoulou and Elias Koutsoupias.
\newblock The price of anarchy of finite congestion games.
\newblock In \emph{Proceedings of the Thirty-Seventh Annual ACM Symposium on
  Theory of Computing (STOC)}, pages 67--73. ACM, 2005.
\newblock \doi{10.1145/1060590.1060600}.

\bibitem[Correa et~al.(2004)Correa, Schulz, and Stier-Moses]{correa2004}
Jos{\'e}~R. Correa, Andreas~S. Schulz, and Nicol{\'a}s~E. Stier-Moses.
\newblock Selfish routing in capacitated networks.
\newblock \emph{Mathematics of Operations Research}, 29\penalty0 (4):\penalty0
  961--976, 2004.
\newblock \doi{10.1287/moor.1040.0098}.

\bibitem[Correa et~al.(2008)Correa, Schulz, and Stier-Moses]{correa2008}
Jos{\'e}~R. Correa, Andreas~S. Schulz, and Nicol{\'a}s~E. Stier-Moses.
\newblock A geometric approach to the price of anarchy in nonatomic congestion
  games.
\newblock \emph{Games and Economic Behavior}, 64\penalty0 (2):\penalty0
  457--469, 2008.
\newblock \doi{10.1016/j.geb.2008.01.001}.

\bibitem[Dafermos and Nagurney(1984)]{dafermos1984}
Stella Dafermos and Anna Nagurney.
\newblock On some traffic equilibrium theory paradoxes.
\newblock \emph{Transportation Research Part B: Methodological}, 18\penalty0
  (2):\penalty0 101--110, 1984.
\newblock \doi{10.1016/0191-2615(84)90023-7}.

\bibitem[Dafermos and Sparrow(1969)]{dafermos1969}
Stella~C. Dafermos and Frederick~T. Sparrow.
\newblock The traffic assignment problem for a general network.
\newblock \emph{Journal of Research of the National Bureau of Standards,
  Section B: Mathematical Sciences}, 73B\penalty0 (2):\penalty0 91--118, 1969.
\newblock \doi{10.6028/jres.073b.010}.

\bibitem[Dial(1971)]{dial1971}
Robert~B. Dial.
\newblock A probabilistic multipath traffic assignment model which obviates
  path enumeration.
\newblock \emph{Transportation Research}, 5\penalty0 (2):\penalty0 83--111,
  1971.
\newblock \doi{10.1016/0041-1647(71)90012-8}.

\bibitem[Duffin(1965)]{duffin1965}
R.~J. Duffin.
\newblock Topology of series-parallel networks.
\newblock \emph{Journal of Mathematical Analysis and Applications},
  10:\penalty0 303--318, 1965.
\newblock \doi{10.1016/0022-247X(65)90125-3}.

\bibitem[Dughmi(2017)]{dughmi2017}
Shaddin Dughmi.
\newblock Algorithmic information structure design: a survey.
\newblock \emph{SIGecom Exch.}, 15\penalty0 (2):\penalty0 2--24, feb 2017.
\newblock \doi{10.1145/3055589.3055591}.
\newblock URL \url{https://doi.org/10.1145/3055589.3055591}.

\bibitem[Fagin(1977)]{fagin1977}
Ronald Fagin.
\newblock Asymptotic miss ratios over independent references.
\newblock \emph{Journal of Computer and System Sciences}, 14\penalty0
  (2):\penalty0 222--250, 1977.
\newblock ISSN 0022-0000.
\newblock \doi{10.1016/S0022-0000(77)80014-7}.
\newblock URL
  \url{https://www.sciencedirect.com/science/article/pii/S0022000077800147}.

\bibitem[Fricker et~al.(2012)Fricker, Robert, and Roberts]{fricker2012lru}
Christine Fricker, Philippe Robert, and James Roberts.
\newblock A versatile and accurate approximation for lru cache performance.
\newblock In \emph{2012 24th International Teletraffic Congress (ITC 24)},
  pages 1--8, 2012.
\newblock URL \url{https://arxiv.org/abs/1202.3974}.

\bibitem[Gast and Van~Houdt(2017)]{gast2017ttl}
Nicolas Gast and Benny Van~Houdt.
\newblock {TTL} approximations of the cache replacement algorithms {LRU(m)} and
  {h-LRU}.
\newblock \emph{Performance Evaluation}, 117:\penalty0 33--57, 2017.
\newblock \doi{10.1016/j.peva.2017.09.002}.

\bibitem[Jiang et~al.(2018)Jiang, Nain, and Towsley]{jiang2018ttl}
Bo~Jiang, Philippe Nain, and Don Towsley.
\newblock On the convergence of the ttl approximation for an lru cache under
  independent stationary request processes.
\newblock \emph{ACM Trans. Model. Perform. Eval. Comput. Syst.}, 3\penalty0
  (4), sep 2018.
\newblock ISSN 2376-3639.
\newblock \doi{10.1145/3239164}.
\newblock URL \url{https://doi.org/10.1145/3239164}.

\bibitem[Kamenica and Gentzkow(2011)]{kamenica2011}
Emir Kamenica and Matthew Gentzkow.
\newblock Bayesian persuasion.
\newblock \emph{American Economic Review}, 101\penalty0 (6):\penalty0
  2590--2615, 2011.
\newblock \doi{10.1257/aer.101.6.2590}.

\bibitem[Koutsoupias and Papadimitriou(1999)]{koutsoupias1999}
Elias Koutsoupias and Christos Papadimitriou.
\newblock Worst-case equilibria.
\newblock In Christoph Meinel and Sophie Tison, editors, \emph{Proceedings of
  the 16th Annual Symposium on Theoretical Aspects of Computer Science
  (STACS)}, pages 404--413, Berlin, Heidelberg, 1999. Springer Berlin
  Heidelberg.
\newblock ISBN 978-3-540-49116-3.
\newblock \doi{10.1007/3-540-49116-3_38}.

\bibitem[Li et~al.(2024)Li, Wang, and Nie]{wardropbr2024}
Jiayang Li, Zhaoran Wang, and Yu~Marco Nie.
\newblock Wardrop equilibrium can be boundedly rational: A new behavioral
  theory of route choice.
\newblock \emph{Transportation Science}, 58\penalty0 (5):\penalty0 973--994,
  2024.
\newblock \doi{10.1287/trsc.2023.0132}.

\bibitem[Manzini and Mariotti(2014)]{manzini2014}
Paola Manzini and Marco Mariotti.
\newblock Stochastic choice and consideration sets.
\newblock \emph{Econometrica}, 82\penalty0 (3):\penalty0 1153--1176, 2014.
\newblock \doi{10.3982/ECTA10575}.

\bibitem[Masatlioglu et~al.(2012)Masatlioglu, Nakajima, and
  Ozbay]{masatlioglu2012revealed}
Yusufcan Masatlioglu, Daisuke Nakajima, and Erkut Ozbay.
\newblock Revealed attention.
\newblock \emph{American Economic Review}, 102\penalty0 (5):\penalty0
  2183--2205, 2012.
\newblock \doi{10.1257/aer.102.5.2183}.

\bibitem[Mat\v{e}jka and McKay(2015)]{matejka2015}
Filip Mat\v{e}jka and Alisdair McKay.
\newblock Rational inattention to discrete choices: A new foundation for the
  multinomial logit model.
\newblock \emph{American Economic Review}, 105\penalty0 (1):\penalty0 272--298,
  2015.
\newblock \doi{10.1257/aer.20130047}.

\bibitem[Milchtaich(2006)]{milchtaich2006}
Igal Milchtaich.
\newblock Network topology and the efficiency of equilibrium.
\newblock \emph{Games and Economic Behavior}, 57\penalty0 (2):\penalty0
  321--346, 2006.
\newblock \doi{10.1016/j.geb.2005.09.005}.

\bibitem[Monderer and Shapley(1996)]{monderer1996}
Dov Monderer and Lloyd~S. Shapley.
\newblock Potential games.
\newblock \emph{Games and Economic Behavior}, 14\penalty0 (1):\penalty0
  124--143, 1996.
\newblock \doi{10.1006/game.1996.0044}.

\bibitem[Murchland(1970)]{murchland1970}
J.~D. Murchland.
\newblock Braess's paradox of traffic flow.
\newblock \emph{Transportation Research}, 4\penalty0 (4):\penalty0 391--394,
  1970.
\newblock \doi{10.1016/0041-1647(70)90196-6}.

\bibitem[Pas and Principio(1997)]{pas1997}
Eric~I. Pas and Shari~L. Principio.
\newblock Braess' paradox: Some new insights.
\newblock \emph{Transportation Research Part B: Methodological}, 31\penalty0
  (3):\penalty0 265--276, 1997.
\newblock \doi{10.1016/S0191-2615(96)00024-0}.

\bibitem[Patriksson(2015)]{patriksson1994}
M.~Patriksson.
\newblock \emph{The Traffic Assignment Problem: Models and Methods}.
\newblock Dover Publications, 2015.
\newblock ISBN 9780486802275.
\newblock URL \url{https://store.doverpublications.com/products/9780486802275}.

\bibitem[Rosenthal(1973)]{rosenthal1973}
Robert~W. Rosenthal.
\newblock A class of games possessing pure-strategy nash equilibria.
\newblock \emph{International Journal of Game Theory}, 2:\penalty0 65--67,
  1973.
\newblock \doi{10.1007/BF01737559}.

\bibitem[Roughgarden(2005)]{roughgarden2005book}
Tim Roughgarden.
\newblock \emph{Selfish Routing and the Price of Anarchy}.
\newblock MIT Press, 2005.

\bibitem[Roughgarden and Tardos(2002)]{roughgarden2002}
Tim Roughgarden and {\'E}va Tardos.
\newblock How bad is selfish routing?
\newblock \emph{Journal of the ACM}, 49\penalty0 (2):\penalty0 236--259, 2002.
\newblock \doi{10.1145/506147.506153}.

\bibitem[Sheffi(1985)]{sheffi1985}
Yosef Sheffi.
\newblock \emph{Urban Transportation Networks: Equilibrium Analysis with
  Mathematical Programming Methods}.
\newblock Prentice-Hall, 1985.

\bibitem[Singh and Joachims(2018)]{singh2018}
Ashudeep Singh and Thorsten Joachims.
\newblock Fairness of exposure in rankings.
\newblock In \emph{Proceedings of the 24th ACM SIGKDD International Conference
  on Knowledge Discovery \& Data Mining (KDD)}, pages 2219--2228. ACM, 2018.
\newblock \doi{10.1145/3219819.3220088}.

\bibitem[Smith(1979)]{smith1979}
Michael~J. Smith.
\newblock The existence, uniqueness and stability of traffic equilibria.
\newblock \emph{Transportation Research Part B: Methodological}, 13\penalty0
  (4):\penalty0 295--304, 1979.
\newblock \doi{10.1016/0191-2615(79)90022-5}.

\bibitem[Steinberg and Zangwill(1983)]{steinberg1983}
Richard Steinberg and Willard~I. Zangwill.
\newblock The prevalence of braess' paradox.
\newblock \emph{Transportation Science}, 17\penalty0 (3):\penalty0 301--318,
  1983.
\newblock \doi{10.1287/trsc.17.3.301}.

\bibitem[{Transportation Networks for Research Core
  Team}(n.d.)]{bstablerTransportationNetworks}
{Transportation Networks for Research Core Team}.
\newblock Transportationnetworks: A repository of transportation network
  datasets.
\newblock GitHub repository, n.d.
\newblock URL \url{https://github.com/bstabler/TransportationNetworks}.
\newblock Accessed: 2026-01-25.

\bibitem[Wang et~al.(2025)Wang, Duan, He, Lyu, Chen, Zheng, Yao, and
  Ma]{Wang_2025}
Leizhen Wang, Peibo Duan, Zhengbing He, Cheng Lyu, Xin Chen, Nan Zheng, Li~Yao,
  and Zhenliang Ma.
\newblock Agentic large language models for day-to-day route choices.
\newblock \emph{Transportation Research Part C: Emerging Technologies},
  180:\penalty0 105307, 2025.
\newblock ISSN 0968-090X.
\newblock \doi{10.1016/j.trc.2025.105307}.
\newblock URL
  \url{https://www.sciencedirect.com/science/article/pii/S0968090X25003110}.

\bibitem[Wardrop(1952)]{wardrop1952}
John~Glen Wardrop.
\newblock Some theoretical aspects of road traffic research.
\newblock \emph{Proceedings of the Institution of Civil Engineers}, 1\penalty0
  (3):\penalty0 325--362, 1952.
\newblock \doi{10.1680/ipeds.1952.11259}.

\bibitem[Zehlike et~al.(2017)Zehlike, Bonchi, Castillo, Hajian, Megahed, and
  Baeza-Yates]{zehlike2017}
Meike Zehlike, Francesco Bonchi, Carlos Castillo, Sara Hajian, Mohamed Megahed,
  and Ricardo Baeza-Yates.
\newblock Fa*ir: A fair top-k ranking algorithm.
\newblock In \emph{Proceedings of the 2017 ACM on Conference on Information and
  Knowledge Management}, CIKM '17, pages 1569--1578, New York, NY, USA, 2017.
  Association for Computing Machinery.
\newblock ISBN 9781450349185.
\newblock \doi{10.1145/3132847.3132938}.
\newblock URL \url{https://doi.org/10.1145/3132847.3132938}.

\bibitem[Zhou et~al.(2022)Zhou, Nguyen, and Xu]{zhou2022aid}
Chenghan Zhou, Thanh~H. Nguyen, and Haifeng Xu.
\newblock Algorithmic information design in multi-player games: Possibilities
  and limits in singleton congestion.
\newblock In \emph{Proceedings of the 23rd ACM Conference on Economics and
  Computation}, EC '22, page 869, New York, NY, USA, 2022. Association for
  Computing Machinery.
\newblock ISBN 9781450391504.
\newblock \doi{10.1145/3490486.3538238}.
\newblock URL \url{https://doi.org/10.1145/3490486.3538238}.

\end{thebibliography}

\end{document}